\documentclass[11pt]{article}   
  \usepackage[margin=0.88in]{geometry} 
  \usepackage{listings}
\usepackage[colorlinks,
citecolor=blue
]{hyperref}  
\usepackage{mathrsfs} 
\usepackage{kpfonts,comment}   
\usepackage{defs}
\usepackage{float}
\usepackage{subcaption}
\renewcommand*{\thefootnote}{\fnsymbol{footnote}}
\usepackage{lipsum}
\newcommand\blfootnote[1]{%
	\begingroup
	\renewcommand\thefootnote{}\footnote{#1}%
	\addtocounter{footnote}{-1}%
	\endgroup
}

\newcommand{\kzedit}[1]{{\color{blue}#1}}

\newcommand{\cpdelete}[1]{{\color{lightgray}#1}}
\newcommand{\arxiv}[1]{#1}
\newcommand{\neurips}[1]{}
\newcommand{\numbering}[1]{#1}
\newcommand{\tnumbering}[1]{}

\renewcommand{\arraystretch}{2} 

\title{Multi-Player Zero-Sum  Markov Games with \\
Networked  
{Separable}  Interactions} 

\begin{document} 

\author{\vspace{10pt} 
Chanwoo Park\footnote{Alphabetical order.} \and \qquad   Kaiqing Zhang\footnotemark[\value{footnote}] \and \qquad  Asuman Ozdaglar
\blfootnote{C.\ Park, and A.\ Ozdaglar are with Massachusetts Institute of Technology, Cambridge, MA, 02139.
		K.~Zhang is with the University of Maryland, College Park, MD, 20742.
		{E-mails:  cpark97@mit.edu, kaiqing@\{umd,mit\}.edu, asuman@mit.edu}.
	}
}
\date{July 13, 2023}

\maketitle
\numbering{
\begin{abstract}   
{We study a new class of  Markov games, \textit{(multi-player) zero-sum Markov Games} with {\it Networked separable interactions} (zero-sum NMGs), to model the local interaction structure in non-cooperative multi-agent sequential decision-making. We define a zero-sum NMG as a model where {the payoffs of the auxiliary games associated with each state are zero-sum and} have some separable (i.e., polymatrix) structure across the neighbors over some interaction network. 
We first identify the necessary and sufficient conditions under which an MG can be presented as a zero-sum NMG, and show that the set of Markov coarse correlated equilibrium (CCE) collapses to the set of Markov Nash equilibrium (NE) in these games, in that the {product of} per-state marginalization of the former for all players yields the latter. Furthermore,  we show that finding approximate Markov \emph{stationary}  CCE  in infinite-horizon discounted zero-sum NMGs is \texttt{PPAD}-hard, unless the underlying network has a ``star topology''. Then, we propose fictitious-play-type dynamics, the classical learning dynamics in normal-form games, for zero-sum NMGs, and establish convergence guarantees to Markov stationary NE under a star-shaped network structure. Finally, in light of the hardness result, we focus on computing a Markov \emph{non-stationary} NE and provide finite-iteration guarantees for a series of value-iteration-based algorithms. We also provide numerical experiments to corroborate our theoretical results.}  
\end{abstract}
}
\tnumbering{
{We study a new class of  Markov games, \textit{(multi-player) zero-sum Markov Games} with {\it Networked separable interactions} (zero-sum NMGs), to model the local interaction structure in non-cooperative multi-agent sequential decision-making. We define a zero-sum NMG as a model where {the payoffs of the auxiliary games associated with each state are zero-sum and} have some separable (i.e., polymatrix) structure across the neighbors over some interaction network. 
We first identify the necessary and sufficient conditions under which an MG can be presented as a zero-sum NMG, and show that the set of Markov coarse correlated equilibrium (CCE) collapses to the set of Markov Nash equilibrium (NE) in these games, in that the {product of} per-state marginalization of the former for all players yields the latter. Furthermore,  we show that finding approximate Markov \emph{stationary}  CCE  in infinite-horizon discounted zero-sum NMGs is \texttt{PPAD}-hard, unless the underlying network has a ``star topology''. Then, we propose fictitious-play-type dynamics, the classical learning dynamics in normal-form games, for zero-sum NMGs, and establish convergence guarantees to Markov stationary NE under a star-shaped network structure.
 Finally, in light of the hardness result, we focus on computing a Markov \emph{non-stationary} NE and provide finite-iteration guarantees for a series of value-iteration-based algorithms. We also provide numerical experiments to corroborate our theoretical results.}  

}

\neurips{
\vspace{-0.1in}
}
\tnumbering{\chapter{Introduction}}
\numbering{\section{Introduction}}

Nash Equilibrium (NE) has been broadly used as a solution concept in game theory, since the seminal works of \cite{von1947theory,nash1950equilibrium}. 
Perhaps equally important, NE 
is also 
deeply rooted in the prediction and analysis of \emph{learning dynamics} in multi-agent strategic environments: it may appear as a natural outcome of many non-equilibrating learning processes of multiple agents interacting with each other  
\cite{cesa2006prediction,fudenberg1998theory}. A prominent example of  such learning processes  is fictitious play (FP) \cite{brown1951iterative,robinson1951iterative}, 
in which myopic agents estimate the opponents' play using history, and then choose a best-response action (based on their payoff matrix) against this estimate, as if the opponents use it as their stationary strategy. The focus of these studies has initially been on the convergence to NE in zero-sum games (see \cite{brown1951iterative, robinson1951iterative}, and also \cite{ref:Fudenberg93, fudenberg1998theory}), and in games with aligned objective (identical-interest and potential games, see  \cite{monderer1996fictitious}). Since then, FP has been shown to converge to NE in more important classes of games, including 2x$n$ games \cite{miyasawa1961convergence, berger2005fictitious}, ``one-against-all'' games \cite{sela1999fictitious}, and zero-sum polymatrix games  \cite{ewerhart2020fictitious}, justifying the prediction power of NE in learning in normal-form/matrix games.  

Some of these results  have recently been extended to the {\it stochastic game} (also known as the  {\it Markov} game (MG)) setting, a model for multi-agent sequential decision-making with state transition dynamics, first introduced in \cite{shapley1953stochastic}. In particular, \cite{leslie2020best,sayin2021decentralized,sayin2022fictitious,baudin2022fictitious,chen2023finite}  have studied best-response-type learning dynamics in two-player zero-sum MGs, and \cite{sayin2022fictitiousb,baudin2022fictitious} have studied that in  multi-player identical-interest games. 
Following the same path as studying matrix games, one natural question arises: 
\textit{Are there other classes of MGs beyond two-player zero-sum and identical-interest cases that allow natural learning dynamics, e.g., fictitious play, to justify NE as the long-run emerging outcome?} 

\neurips{
\vspace{-0.03in}
}

On the other end of the spectrum, it is well-known that for {\it general-sum}  normal-form games, the special case of MGs without the state transition dynamics,  \emph{computing} an NE is  intractable 
\cite{daskalakis2009complexity,chen2009settling}. Relaxed solution concepts as (coarse) correlated equilibrium ((C)CE) have thus been favored when it comes to equilibrium computation for general-sum, multi-player games  \cite{papadimitriou2008computing,cesa2006prediction}. Encouragingly, when the interactions among players have some  networked separable structure, also known as being \emph{polymatrix}, computing NE may be made tractable even in  multi-player settings. This has been instantiated in the seminal works  \cite{cai2011minmax,cai2016zero} in the  normal-form game setting, which showed that any CCE collapses to the NE in such games when the payoffs are zero-sum. Thus, any algorithms that can efficiently compute the CCE in such games will lead to the efficient computation of the NE. 

\neurips{
\vspace{-0.03in}
}

In fact, besides being of theoretical interest, multi-player zero-sum games with networked separable interactions also find a range of applications, including  security games \cite{cai2016zero}, fashion games \cite{cao2013fashion, cao2014fashion, zhang2018fashion}, and resource  allocation problems \cite{bergman1998separable}. 
These examples, oftentimes, naturally  involve some \emph{state} transition that captures the dynamics of the evolution of the environment in practice. 
For example, in the security game, the protection level or immunity of a target increases as a function of the number of past attacks, leading to a smaller probability of an attack on the target being successful. Hence, it is imperative to study such multi-player zero-sum games with state transitions. As eluded in the recent results \cite{jin2022complexity,daskalakis2022complexity}, such a transition from \emph{stateless} to \emph{stateful} cases may not always yield straightforward and expected results, e.g., computing stationary CCE can be computationally intractable in  stochastic games, in stark contrast to the normal-form case where CCE can be efficiently computed. 
This naturally prompts another question: {\it Are there other types of (multi-player) MGs that may circumvent the computational hardness of computing  NE/CCE?}  
\neurips{
\vspace{-0.05in}
}

In an effort to address these two questions, we introduce a new class of Markov games -- {\it (multi-player) zero-sum Markov games with networked separable interactions} (zero-sum NMGs). We summarize our contributions as follows, and defer a more detailed literature review to \Cref{appendix:literature-review}. 

\neurips{
\vspace{-0.14in}
}

\paragraph{Contributions.} 
First, we introduce a new class of {non-cooperative} Markov games: (multi-player) 
 zero-sum MGs with Networked separable interactions (zero-sum NMGs),    wherein the payoffs of the \emph{auxiliary-games}  associated with each state, i.e., the sum of instantaneous reward and expectation of any estimated state-value functions, possess the  multi-player zero-sum and {networked separable (i.e., polymatrix)}   structure as in \cite{bergman1987methods, bergman1998separable,cai2011minmax,cai2016zero} for normal-form games, a  strict generalization of the latter. We also provide structural results on the reward and transition dynamics of the game, 
 as well as examples of this class of games.  Specifically, for a Markov game to qualify as a zero-sum NMG, if and only if its reward function has the zero-sum polymatrix structure, and its transition dynamics is an {\it ensemble} of multiple single-controller transition dynamics that are sampled randomly at each state (see \Cref{rem:implication} for more details). This transition dynamics covers the common ones in the MG literature, including the single-controller and turn-based dynamics. 
Second, we show that Markov CCE and Markov NE \emph{collapse} in that the {product of} per-state marginal distributions of the former yields the latter,  
making it sufficient to focus on the former in equilibrium computation. We then show the {\tt PPAD}-hardness \cite{papadimitriou1994complexity} of computing the \emph{Markov stationary} equilibrium, a natural solution concept in infinite-horizon discounted MGs{, unless the underlying network {has} a star-topology.} 
This is in contrast to the normal-form case where CCE is always computationally tractable. Third, we study the fictitious-play property  \cite{monderer1996fictitious} of zero-sum NMGs, showing that the fictitious-play dynamics  \cite{sayin2022fictitious,sayin2022fictitiousb} converges to the Markov stationary NE, for zero-sum NMGs {with a star-shaped network structure}. 
Finally, 
in light of the hardness of computing stationary equilibria, we develop a series of value-iteration-based algorithms for computing a Markov \emph{non-stationary} {NE} of zero-sum NMGs, with finite-iteration guarantees. We also provide numerical experiments to corroborate our theoretical results in \Cref{appendix:experiment}. We hope our results serve as a starting point for studying this networked separable interaction structure in non-cooperative Markov games.

\neurips{
\vspace{-0.14in}
}

\paragraph{Notation.} 
For a real number $c$, we use $(c)_+$ to denote $\max\{c,0\}$. For an event $\cE$, we use $\pmb{1}(\cE)$ to denote the indicator function such that $\pmb{1}(\cE)=1$ if $\cE$ is true, and $\pmb{1}(\cE)=0$ otherwise. 
{We define multinomial distribution with probability $(w_i)_{i\in\cN}$ as Multinomial$\big((w_i)_{i\in\cN}\big)$. We denote the uniform distribution over a set $\cS$ as Unif($\cS$). We denote the Bernoulli distribution with probability $p$ as Bern($p$). The $\text{sgn}$ function is defined as $\text{sgn}(x) = 2\times\pmb{1}(x\geq0) - 1$. The KL-divergence between two probability distributions $p, q$ is denoted as $\mathrm{KL}(p, q) = \EE_p [\log(p/q)]$.} 
For a graph $\cG = (\cN, \cE)$, we denote the set of neighboring  nodes of node $i \in \cN$ as $\cE_i$ (without  including $i$). {The maximum norm of a matrix $X\in\RR^{m\times n}$, denoted as $\norm{X}_{\max}$, 
is defined as $\norm{X}_{\max} := \max_{i \in [m],j\in [n]} |X_{i,j}|$.}

\numbering{\section{Preliminaries}}
\tnumbering{\chapter{Preliminaries}}

\neurips{
\vspace{-0.1in}
}

\numbering{\subsection{Markov games}}
\tnumbering{\section{Markov games}}
\label{sec:pre_mg}

\neurips{
\vspace{-0.05in}
}

We define a Markov game as a tuple 
$(\cN, \cS, \cA, H, (\PP_h)_{h \in [H]}, (r_{h,i})_{h \in [H], i \in \cN}, \gamma)$,  
where {$\cN = [n]$} is the set of players, $\cS$ is the state space with $|\cS| = S$, $\cA_i$ is the action space for player $i$ with $|\cA_i| = A_i$ and $\cA = \prod_{i \in \cN} \cA_i$, $H \leq \infty$ is the length of the horizon, {$\PP_h:\cS\times \cA\to\Delta(\cS)$  captures the state transition dynamics at timestep $h$,} $r_{h,i}\in[0,R]$ is the reward function for player $i$ at timestep $h$, bounded by some $R>0$, and $\gamma \in (0,1]$ is a discount factor. An MG with a finite horizon ($H < \infty$) is also referred to as an \emph{episodic} MG, while an MG with an infinite horizon ($H = \infty$) and $\gamma < 1$ is referred to as an infinite-horizon $\gamma$-discounted MG. {When $H = \infty$, we will consider the transition dynamics and reward functions, denoted by $\PP$ and $(r_i)_{i \in \cN}$, respectively, to be independent of $h$.} Hereafter, we may use \emph{agent} and \emph{player} interchangeably.

\neurips{
\vspace{-0.1in}
}

\paragraph{Policy.}
Consider the stochastic Markov policy for player $i$, denoted by  $\pi_i$, as  
$\pi_i:= \{\pi_{h, i}:  \cS  \to\Delta(\cA_i)\}_{h\in[H]}$. {A \emph{joint} Markov policy is a policy $\pi:= \{\pi_{h}:  \cS  \to\Delta(\cA)\}_{h\in[H]}$, where  $\pi_h:\cS\to\Delta(\cA)$ decides the joint action of all players that can be potentially correlated. A joint Markov policy is a \emph{product}  Markov policy  if $\pi_h:\cS\to\prod_{i\in\cN}\Delta(\cA_i)$ for all $h\in[H]$,} and is denoted as $\pi = \pi_1 \times \pi_2 \times \dots \times \pi_n$. When the policy is  independent of $h$, the policy is called a \emph{stationary} policy. We let $\ba_{h}$ denote the joint action of all agents at timestep $h$. {Unless otherwise noted, we will work with Markov policies throughout.} {We denote $\pi(s) \in \Delta(\cA)$ as the joint  policy at  state $s\in\cS$. }

\neurips{
\vspace{-0.1in}
}

\paragraph{Value function.} For  player $i$, the value function under joint policy $\pi$, at timestep $h$  and state  $s_h$ is defined as 
$V^{\pi}_{h, i}(s_h):= \EE_{\pi}\big[\sum_{h'=h}^H \gamma^{h'-h}r_{h', i}(s_{h'}, \ba_{h'})\,\big| \, s_h \big],$  which denotes the expected cumulative reward for player $i$ at step $h$ if all players adhere to policy $\pi$. {We also define $V_{h, i}^\pi(\rho):=\EE_{s_h \sim \rho}[V_{h,i}^\pi (s_h)]$ for some state distribution $\rho\in\Delta(\cS)$.} We  denote  the $Q$-function for the $i$-th player  under policy $\pi$, at step $h$ and state $s_h$ as $Q^{\pi}_{h,i}(s_h, \ba_h):= \EE_{\pi}\big[\sum_{h'=h}^H \gamma^{h'-h}r_{h', i}(s_{h'}, \ba_{h'})\,\big| \, s_h, \ba_h\big],
$ which determines the expected cumulative reward for the $i$-th player at step $h$, when starting from the state-action pair $(s_h, \ba_h)$.  {For the infinite-horizon discounted setting, we also use $V^\pi_i$ and $Q^\pi_i$ to denote $V^\pi_{1,i}$ and $Q^\pi_{1,i}$ for short, respectively.}

 \neurips{
\vspace{-0.1in}
}

 \paragraph{Approximate 
 equilibrium.}  
Define {an} $\epsilon$-approximate \emph{Markov perfect  Nash equilibrium}   as a product policy $\pi$, which satisfies
$\max_{i \in \cN}\max_{\mu_i \in (\Delta(\cA_i)){^{|\cS|\times H}}} (V_{h,i}^{\mu_i, \pi_{-i}}(\rho) - V_{h,i}^\pi(\rho)) \leq \epsilon$ for \emph{all} $\rho \in \Delta(\cS)$ and $h\in[H]$, 
{where $\pi_{-i}$ represents the marginalized policy of all players except player $i$. 
Define {an} $\epsilon$-approximate \emph{Markov coarse correlated equilibrium}  as a joint policy $\pi$, which satisfies $\max_{i \in \cN} \max_{\mu_i \in (\Delta(\cA_i)){^{|\cS|\times H}}} (V_{h,i}^{\mu_i, \pi_{-i}}(\rho) - V_{h,i}^\pi(\rho)) \leq \epsilon$ for \emph{all} $\rho \in \Delta(\cS)$  and $h\in[H]$. 
In the infinite-horizon setting, they can be equivalently defined as satisfying $\max_{s\in\cS}\max_{i \in \cN} \max_{\mu_i \in (\Delta(\cA_i)){^{|\cS|}}} (V_{i}^{\mu_i, \pi_{-i}}(s) - V_{i}^\pi(s)) \leq \epsilon$. 
{If the above conditions only hold for certain $\rho$ and $h=1$}, we refer to them  as Markov {\it non-perfect} NE and CCE, respectively.  {Unless otherwise noted, we hereafter focus on Markov perfect equilibria, and sometimes refer to them simply as {\it Markov equilibria}  when it is clear from the context.} 
{In the infinite-horizon setting, if additionally, the policy is \emph{stationary}, then they are referred to as {a} \emph{Markov stationary} NE and CCE, respectively. 
}}

\neurips{
\vspace{-0.1in}
}
\tnumbering{\section{Multi-player zero-sum games with networked separable interactions }}
\numbering{
\subsection{Multi-player zero-sum games with networked separable interactions }}\label{sec:mp_zs_matrix_def}

\neurips{
\vspace{-0.1in}
}

As a generalization of \emph{two-player} zero-sum matrix games, \emph{(multi-player)} zero-sum polymatrix games have been introduced in \cite{bergman1987methods, bergman1998separable,cai2011minmax,cai2016zero}. A \emph{polymatrix game}, also known as a \emph{separable network game} is defined by a tuple $(\cG =(\cN, \cE_r), \cA = \prod_{i \in \cN} \cA_i, (r_{i, j})_{(i, j) \in \cE_r})$. Here, $\cG$ is an undirected {connected} graph where $\cN = [n]$ denotes the set of players and $\cE_r \subseteq \cN \times \cN$ denotes the set of edges {describing the rewards' networked structures, where the graph neighborhoods represent the interactions among players.} For each edge, a two-player game is defined for players $i$ and $j$, with action sets $\cA_i$ and $\cA_j$, and reward functions $r_{i,j}: \cA_i \times \cA_j \rightarrow \mathbb{R}$, and similarly for $r_{j,i}$. The reward for player $i$ for a given {joint} action $\ba = (a_i)_{i \in \cN} \in \prod_{i \in \cN} \cA_i$  is calculated as the sum of the rewards for all edges involving player $i$, that is, $r_i(\ba) = \sum_{j: (i, j) \in \cE_r} r_{i,j}(a_i, a_j)$. {To be consistent with our terminology later, hereafter, we also refer to such games as \emph{(multi-player) Games  with Networked separable interactions (NGs)}.} 

\neurips{
\vspace{-0.05in}
}

In a \emph{zero-sum} polymatrix game (i.e.,  a (multi-player) \emph{zero-sum} Game with Networked separable interactions {(zero-sum NG)}), the sum of rewards for all players at  any {joint action} $\ba= (a_i)_{i\in\cN} \in \prod_{i \in \cN} \cA_i$ equals zero{, i.e., $\sum_{i\in\cN}r_i(\ba)=0$.}  
One can define the policy of agent $i$, i.e., $\pi_i\in\Delta(\cA_i)$, so that the agent takes actions by  sampling    $a_i\sim \pi_i(\cdot)$. Note that $\pi_i$ can be viewed as the reduced case of the policy defined in \Cref{sec:pre_mg} when $\cS=\emptyset$ and $H=1$. 
The expected reward for player $i$ under $\pi$ can then be computed as:
\begin{align}\label{equ:expected_reward}
r_i(\pi) &:= \sum_{j: (i, j) \in \cE_r} \sum_{a_i \in \cA_i, a_j \in \cA_j} r_{i,j}\left(a_i, a_j\right) \pi_i(a_i) \pi_j(a_j) = \pi_i^\intercal \br_i \bpi, 
\end{align} 
where $\br_i$ denotes the matrix   $\br_i := (r_{i,1}, \dots, r_{i,(i-1)}, \bm{0}, r_{i, (i+1)}, \dots, r_{i,n}) \in \RR^{|\cA_i| \times \sum_{i \in \cN} |\cA_i|}$ 
and $\pmb{\pi}:= (\pi_1^\intercal, \pi_2^\intercal, \dots, \pi_n^\intercal)^\intercal \in \RR^{\sum_{i \in \cN} |\cA_i| }$. 
We define $\br := (\br_1^\intercal, \br_2^\intercal, \dots, \br_n^\intercal)^\intercal \in \RR^{\sum_{i \in \cN} |\cA_i| \times \sum_{i \in \cN} |\cA_i|}$. Then in this case we have  $\sum_{i\in\cN}r_i(\pi)=0$ for any policy $\pi$. See more prominent  application examples of zero-sum polymatrix games  in \cite{cai2011minmax,cai2016zero}.

\neurips{
\vspace{-0.1in}
}

\tnumbering{\chapter{Multi-Player (Zero-Sum) MGs with Networked Separable Interactions}}
\numbering{\section{Multi-Player (Zero-Sum) MGs with Networked Separable Interactions}}\label{sec:MPGdef}

\neurips{
\vspace{-0.1in}
}

We now introduce our model of multi-player zero-sum MGs with networked separable interactions. 

\neurips{
\vspace{-0.1in}
}

\numbering{{\subsection{Definitions}}}
\tnumbering{{\section{Definitions}}}
\label{sec:definition}
\begin{definition}\label{sec:def_MZNG}
An infinite-horizon $\gamma$-discounted MG is called a \textit{(multi-player) MG with Networked separable interactions (NMG)}  characterized by  a tuple $(\cG = (\cN, \cE_Q), \cS, \cA, \PP, (r_{i})_{i \in \cN}, \gamma)$ 
{if for any function $V: \cS \to \RR$, defining $Q_{i}^V(s, \ba):= r_i(s, \ba) + \gamma \sum_{s' \in \cS} \PP(s'\mid s, \ba) V(s')$, there exist a {set of} functions $(Q_{i,j}^V)_{(i, j) \in \cE_Q}$ {and an undirected connected graph $\cG = (\cN, \cE_Q)$} such that $
    Q_{i}^V(s, \ba) = \sum_{j\in \cE_{Q, i}} Q_{i,j}^V(s, a_i, a_j)$ }
holds for every $i \in \cN$, $s \in \cS$, $\ba \in \cA$, {where $\cE_{Q, i}$ denotes the neighbors of player $i$ induced by the edge set $\cE_{Q}$ (without including $i$)}. {When it is clear from the context, we represent the NMG tuple simply as $\cG = (\cN, \cE_Q)$.} 

\neurips{
\vspace{-0.05in}
}

A finite-horizon MG is called a \textit{(multi-player) MG with Networked separable interactions} {if for any set of functions $V:=\{V_h\}_{h \in [H+1]}$ where $V_h: \cS \to \RR$, defining $Q_{h, i}^V(s, \ba):= r_{h, i}(s, \ba) + \gamma \sum_{s' \in \cS} \PP_h(s'\mid s, \ba) V_{h+1}(s')$, there exist a set of  functions $(Q_{h, i,j}^V)_{(i, j) \in \cE_Q,h\in[H]}$ such that $Q_{h, i}^V(s, \ba) = \sum_{j\in\cE_{Q, i}} Q_{h, i,j}^V(s, a_i, a_j)
$ holds for every $i\in \cN$, $s \in \cS$, $h\in[H]$, $\ba \in \cA$.} 

\neurips{
\vspace{-0.05in}
}

{
A (multi-player) NMG is called a  \emph{{(multi-player)  \textit{zero-sum} MG} with Networked separable interactions (zero-sum NMG)} if additionally {$(\cG =(\cN, \cE_Q), \cA = \prod_{i \in \cN} \cA_i, (r_{i, j}(s):= Q_{i,j}^{\pmb{0}}(s))_{(i, j) \in \cE_Q})$ 
forms a zero-sum NG  for all $s\in\cS$ in the infinite-horizon $\gamma$-discounted case, or  $(\cG =(\cN, \cE_Q), \cA = \prod_{i \in \cN} \cA_i, (r_{h,i, j}(s):= Q_{h,i,j}^{\pmb{0}}(s))_{(i, j) \in \cE_Q})$ 
forms a zero-sum NG for all $s\in\cS$ and $h\in[H]$ in the finite-horizon case.}} 
\end{definition}

\neurips{
\vspace{-0.05in}
}

{Regarding the assumption that the above conditions hold under any (set of) functions  {$V$, 
one may understand this as a structural requirement 
to inherit the polymatrix structure in the Markov game case.} It is natural since {$\{Q_i^V\}_{i \in \cN}$  would play the role of the \emph{payoff matrix} in the normal-form case, when value-(iteration) based algorithms are used to solve the MG.} As our hope is to exploit the networked structure in the payoff matrices to develop efficient algorithms for solving such MGs, {if we do not know \emph{a priori}  which value function estimate $V$ will be encountered in the algorithm update, the networked structure may easily break if we do not assume them to hold for \emph{all} possible $V$.  Moreover, such a definition easily encompasses the normal-form case, by preserving the polymatrix structure of the \emph{reward functions} (when substituting $V$ to be a zero function). Some alternative definition (see Remark \ref{rem:alt_def}) may not necessarily preserve the polymatrix structure of even the reward functions {in a consistent way (see \Cref{ssec:counterexample_alternative} for a concrete example). We thus focus on \Cref{sec:def_MZNG}, which at least covers the polymatrix structure of the reduced case regarding only reward functions.}}

Indeed, such a networked structure in Markov games may be fragile. 
{We now propose both sufficient and \emph{necessary} conditions for the reward function's structure and the transition dynamics of the MG, to be an NMG.}
 Here we focus on the infinite-horizon discounted setting for a simpler exposition. For finite-horizon cases, a similar statement holds, which is deferred to \Cref{appendix:MPGdef}. 
We also defer the full statement and proof of the following result to \Cref{appendix:MPGdef}. {We first introduce the definition of decomposability and the set $\cN_C$, which will be used in establishing the conditions.  For a graph $\cG = (\cN, \cE)$, we define $\cN_C:= \{i \,\given \, (i,j) \in \cE \text{ for all } j \in {\cN}\}$, which may be an empty set if no such node $i$ exists.}} 
\begin{definition}[Decomposability] 
\label{def:decomposability}
    A non-negative function $f:X^{|\cD|} \to \RR^+\cup \{0\}$ is decomposable with respect to a set $\cD \neq \emptyset$ if there exists a set of non-negative functions $(f_i)_{i \in \cD}$ with $f_i:X  \to \RR^+\cup \{0\}$, such that $f(x) = \sum_{i \in \cD} f_i(x_i)$ holds for any $x\in X^{|\cD|}$. A non-negative function $f:X^{|\cD|} \to \RR^+\cup \{0\}$ is decomposable with respect to a set $\cD = \emptyset$, if there exists a non-negative constant $f_o$ such that $f(x) = f_o$ holds for any $x\in X^{|\cD|}$. 
\end{definition}

\begin{restatable}{proposition}{MZNMGcond}
\label{prop:MPGcond}
For a given graph  $\cG=(\cN, \cE_Q)$, an MG $(\cN, \cS, \cA, \PP, (r_{i})_{i \in \cN}, \gamma)$ with more than two players is an NMG with respect to $\cG$ if and only if:
(1) \textbf{$r_i(s, a_i, \cdot)$ is decomposable with respect to} $\cE_{Q, i}$ for each $i \in \cN, s \in \cS, a_i \in \cA_i $, {i.e.,} $r_i(s, \ba) = \sum_{j \in \cE_{Q,i}} r_{i,j}(s, a_i, a_j) $ for a set of functions $\{r_{i,j}(s, a_i, \cdot)\}_{j \in \cE_{Q, i}}$, and {(2) \textbf{the transition dynamics $\PP(s'\given s, \cdot)$ is decomposable with respect to} the set $\cN_C$ of this $\cG$, {i.e.,} $\PP(s'\given s, \ba) = \sum_{i \in \cN_C} \FF_i(s'\given s, a_i)$ for a set of functions $\{\FF_i(s'\given s,\cdot)\}_{i\in\cN_C}$ if $\cN_C \neq \emptyset$, or $\PP(s'\given s, \ba) = \FF_o(s'\given s)$ for some constant function (of $\ba$) $\FF_o(s'\given s)$ {if $\cN_C = \emptyset$}.}     
{Moreover, an MG qualifies as a zero-sum NMG if and only if it satisfies an additional condition: the NG, characterized by $(\cG, \cA, (r_{i,j}(s))_{(i,j) \in \cE_Q})$, is a zero-sum NG {for all $s\in\cS$}.}
{In the case of two players, every (zero-sum) Markov game becomes a (zero-sum) NMG.}  
\end{restatable} 

\neurips{
\vspace{-7pt}
}

\arxiv{
 \begin{proof}[Proof Sketch of \Cref{prop:MPGcond} when $\cN_C \neq \emptyset$]
First, we prove that if an MG satisfies the decomposability of the reward function $r_i(s, a_i. \cdot)$ and the transition {dynamics} $\PP(s'\mid s, \cdot)$, then $Q_{i}^V$ satisfies the following:
\begin{align*}
    &Q_i^V(s, \ba) = \sum_{j\in\cE_{Q, i}}\biggl( r_{i,j}(s, a_i,a_j) + \gamma \sum_{s' \in \cS} \bigl(\lambda_{i,j}(s)  \pmb{1}(i \in \cN_C)\FF_i(s' \mid s, a_i)
    \\
    &\qquad \qquad \qquad \qquad \qquad  + \pmb{1}(j \in \cN_C)\FF_j(s' \mid s, a_j)\bigr) V(s')\biggr) =: \sum_{j \in \cE_{Q, i}} Q_{i,j}^V(s, a_i, a_j),
\end{align*}
where $\lambda_{i,j}(s) = 1/|\cE_{Q,i}|$. We have used the fact that $\cN_C\subseteq \cE_{Q,i}$ for every $i\in\cN$.

Next, we prove the necessary conditions for an MG to be an NMG. By definition we have
\begin{align}
    \sum_{j \in \cE_{Q, i}} (Q_{i,j}^V(s, a_i, a_j) - Q_{i,j}^{V'}(s, a_i, a_j)) = \gamma \langle \PP(\cdot \mid s, \ba), V(\cdot) - V'(\cdot) \rangle \label{eqn:subtract-2}
\end{align}
for any $V, V'$ and any $(s,\ba)$. 

For every $s \in \cS$, define $B_s:\cS \to \RR$ such that $B_s(s')= \pmb{1}(s= s')$. We define $\mathbb{G}_{i,j}(s'\mid s,a_i,a_j) := 1/\gamma \sum_{j \in \cE_{Q, i}}(Q_{i,j}^{B_{s'}}(s, a_i, a_j) - Q_{i,j}^{\pmb{0}}(s, a_i, a_j)) $, then by plugging in $V= B_{s'}$ and $V'= \pmb{0}$, we can derive  $\PP(\cdot \mid s, \ba) = \sum_{j \in \cE_{Q, i}} \mathbb{G}_{i,j}(\cdot \mid s, a_i, a_j)$ for every $i$ from \Cref{eqn:subtract-2}. {Note that the decomposability of $\PP(s'\given s, a_i, \cdot)$ with respect to $\cE_{Q,i}$ above has to hold for all $i\in\cN$.} We  then prove that this indicates $\PP(s'\given s, \cdot)$ is decomposable with respect to $\cN_C$ (see  \Cref{appendix:MPGdef}). Moreover, we have  $$r_{i}(s,\ba)  = \sum_{j \in \cE_{Q, i}} \left( Q_{i,j}^V(s, a_i, a_j) - \gamma   \langle \pmb{1}(j \in \cN_C)\FF_j(\cdot \mid s, a_j), V(\cdot) \rangle\right),$$ which will show that $r_i(s, a_i, \cdot)$ is also decomposable with respect to $\cE_{Q, i}$. 
\end{proof}
}

\begin{remark}[Stronger sufficient condition]
\label{rmk:stronger}
{We note that for an MG $(\cN, \cS, \cA, \PP, (r_{i})_{i \in \cN}, \gamma)$, if for every agent $i$, $r_i(s,a_i, \cdot)$ is decomposable with respect to some $\cE_{r} \subseteq \cE_Q$,  and $\PP(s'\given s, \cdot)$ is decomposable with respect to some $\cN_P\subseteq\cN_C$, then  one can still prove the \emph{if} part, i.e., there exists some $\cG = (\cN, \cE_Q)$ such that the game is an NMG with respect to this $\cG$. See Figure \ref{fig:prop1} for the illustration. This is because by our definition,  being decomposable with respect to a subset implies being decomposable with respect to a larger set, as one can choose the functions $f_i$ for the $i$ in the complement of the subset to be simply zero. We chose to state as in \Cref{prop:MPGcond} just for the purpose of presenting both the \emph{if} and \emph{only if} conditions in a concise and unified way.}
\end{remark}

\begin{remark}[An alternative NMG definition]\label{rem:alt_def}
{Another reasonable definition of NMG may be as follows: if for any \emph{policy} $\pi$, there exist a {set of} functions $(Q_{i,j}^\pi)_{(i, j) \in \cE_Q}$ {and an undirected connected graph $\cG = (\cN, \cE_Q)$} such that $Q_{i}^\pi(s, \ba) = \sum_{j\in \cE_{Q, i}} Q_{i,j}^\pi(s, a_i, a_j)$ holds for every $i \in \cN$, $s \in \cS$, $\ba \in \cA$. Note that such a definition can be useful in developing \emph{policy}-based algorithms (while \Cref{sec:def_MZNG} is more amenable to developing \emph{value}-based algorithms), e.g., policy iteration, policy gradient, actor-critic methods, where the $Q$-value under certain \emph{policy} $\pi$ will appear in the updates and may need to preserve certain decomposability structure, for any policy $\pi$ encountered in the algorithm updates. However, in this case, we cannot always guarantee the decomposability of $\PP(s'\given  s, \cdot)$ or $r_i(s, a_i, \cdot)$. For example, if we assume that $r_i(s, \ba) = 0$ for every $i \in \cN$, $s \in \cS$, $\ba \in \cA$, then $Q^\pi_i(s, \ba) = \sum_{j \in \cE_{Q,i}} 0$ and thus $Q_{i}^\pi(s, a_i, \cdot)$ is always decomposable regardless of  $\PP(s'\given  s, \ba)$. 
However, {interestingly, we can show that} the decomposability of the transition dynamics and the reward function {as in \Cref{prop:MPGcond}} can {still} be guaranteed, {{as long as some  {\it degenerate} cases as above do not occur. In particular, if there exist no $i \in \cN$ and $s \in \cS$ such that $Q_i^\pi(s, \ba)$ is} a constant function of $\ba$ for any $\pi$, then the results in \Cref{prop:MPGcond} and hence after still hold. 
We defer a detailed   discussion on this alternative definition to  \Cref{appendix:MPGdef}.}} 
\end{remark}

\neurips{
\vspace{-7pt}
}

\begin{remark}[Implication of decomposable transition dynamics]\label{rem:implication}
For an MG to be an NMG, by \Cref{prop:MPGcond} the transition dynamics should be decomposable, i.e., $\PP(\cdot \given  s, \ba) = \sum_{j \in \cN_C} \FF_j ( \cdot \given  s, a_j)$ or $\PP(\cdot\given s, \ba) = \FF_o(s'\given s)$. {We first focus on the discussion of the former case.} Define $w_j(s,a_j):= \sum_{s' \in \cS} \FF_j(s'\given s, a_j)$.
If we fix the value of $s$ and $a_{-j}$, then  $w_j(s, a_j)$ has to be the same for different values of $a_j$ {due to the fact $\sum_{j \in \cN\kzedit{_C}} w_j(s, a_j) = 1$}. {Also note that  by definition, $w_j(s, a_j)$ does not depend on the choice of this fixed $a_{-j}$.} 
Therefore, {such a $w_j(s, a_j)$ can be  written as $w_j(s)$,} 
{where} $w_j(s) = \sum_{s' \in \cS} \FF_j(s'\given s, a_j) $ for all $a_j \in \cA_j$. {We can thus rewrite $(\FF_j)_{j\in\cN_C}$ using some actual probability distributions  $(\PP_j)_{j\in\cN_C}$, such that  if $w_j(s) \neq 0$, then we rewrite $\FF_j$ as $\FF_j(s' \given  s, a_j) = w_j (s)\frac{\FF_j (s' \given  s, a_j)}{w_j(s)}= w_j (s)\PP_j(s'\given s, a_j)$, and  if $w_j(s) = 0$,  we rewrite $\FF_j$ as $\FF_j(s' \given  s, a_j) = w_j(s) \PP_j(s'\given s, a_j)$ for an arbitrary {probability distribution}  $\PP_j(\cdot\given s, a_j)$. Notice that  $\sum_{s' \in \cS} \PP_j (s'\given s, a_j) = 1$ for any $j\in\cN_C$.} Then, the decomposable transition dynamics can be represented as $\PP(\cdot \mid s, \ba) = \sum_{j \in \cN_C} w_j(s) \PP_j (\cdot \given  s, a_j)$, i.e., an \emph{ensemble} of the transition dynamics that is only controlled by single controllers. The model's transition dynamics thus act according to the following two steps: (1) sampling the controller according to the  distribution {Multinomial$\big((w_i(s))_{i\in\cN_C}\big)$}, and (2) transitioning the state following the sampled controller's dynamics. 
{Such a model has also been investigated under the name of transition dynamics with \emph{additive structures} in \cite{flesch2007stochastic}.}
Note that our model is more general and thus  covers the single-controller MG setting \cite{filar2012competitive}, where there is only one agent controlling the transition dynamics at {\it all} states. {It also covers the setting of turn-based MGs \cite{filar2012competitive}, where in each round, depending on the current state $s$, the transition dynamics is  by turns affected by only one of the agents. This can be captured by the proper choice of $(w_i(s))_{i\in\cN_C})$ that takes value $1$ only for one agent at each state $s$ (while takes value $0$ for all other non-controller agents at  each state $s$). {Additionally, the second case where $\PP(\cdot\given s, \ba) = \FF_o(s'\given s)$ corresponds to the one with no ensemble of controller agents.}}
\end{remark}

\begin{proposition}[Decomposition of $(Q_i^V)_{i \in \cN}$]
\label{def:canonical}
\neurips{
For {an infinite-horizon $\gamma$-discounted} NMG with $\cG = (\cN, \cE_Q)$ such that $\cN_C \neq \emptyset$, {if we know that} $\PP(s'\mid s, \ba) = \sum_{i \in \cN_C} \FF_i(s'\mid s, a_i)$, and $r_{i}(s, \ba) = \sum_{j \in \cE_{Q,i}} r_{i,j}(s, a_i, a_j)$ for some $\{\FF_i\}_{i \in \cN_C}$ and $\{r_{i,j}\}_{(i,j) \in \cE_Q}$, {then} {{the $Q_{i,j}^V$  given in \Cref{sec:def_MZNG} can be represented as}

\neurips{
\vspace{-0.18in}
}

{
\fontsize{9}{9}\selectfont
\begin{align*}
Q_{i,j}^V(s, a_i, a_j) = r_{i,j}(s,a_i, a_j) + \sum_{s' \in \cS} \gamma \left(\pmb{1}(j \in \cN_C) \FF_j(s'\mid s, a_j) + \pmb{1}(i \in \cN_C)\lambda_{i,j}(s)\FF_i(s'\mid s,a_i) \right) V(s')
\end{align*}}}

\neurips{
\vspace{-0.18in}
}

for any {non-negative} {$(\lambda_{i,j}(s))_{(i,j) \in \cE_Q}$}  such that $\sum_{j \in \cE_{Q,i}} \lambda_{i,j}(s) = 1$, for all {$i\in\cN$ and $s\in \cS$}. We call it the \textit{canonical} decomposition of $\{Q_i^V\}_{i \in \cN}$ {when $Q_{i,j}^V$ can be represented as above with} $\lambda_{i,j}(s) = {1}/{|\cE_{Q,i}|}$ for $j \in \cE_{Q, i}$. {The case with $\cN_C=\emptyset$ is deferred to \Cref{appendix:MPGdef}}.} 
\arxiv{For {an infinite-horizon $\gamma$-discounted} NMG with $\cG = (\cN, \cE_Q)$ and $\cN_C \neq \emptyset$, if we know that $\PP(s'\given s, \ba) = \sum_{i \in \cN_C} \FF_i(s'\given s, a_i)$, and $r_{i}(s, \ba) = \sum_{j \in \cE_{Q,i}} r_{i,j}(s, a_i, a_j)$ for some $\{\FF_i\}_{i \in \cN_C}$ and $\{r_{i,j}\}_{(i,j) \in \cE_Q}$,  then {the $Q_{i,j}^V$  given in \Cref{sec:def_MZNG} can be represented as}
\numbering{\begin{align*}
Q_{i,j}^V(s, a_i, a_j) = r_{i,j}(s,a_i, a_j) + \sum_{s' \in \cS} \gamma \left(\pmb{1}(j \in \cN_C) \FF_j(s'\mid s, a_j) + \pmb{1}(i \in \cN_C)\lambda_{i,j}(s)\FF_i(s'\mid s,a_i) \right) V(s')
\end{align*}}
\tnumbering{\begin{align*}
Q_{i,j}^V(s, a_i, a_j) &= r_{i,j}(s,a_i, a_j) 
\\
&\qquad+ \sum_{s' \in \cS} \gamma \left(\pmb{1}(j \in \cN_C) \FF_j(s'\mid s, a_j) + \pmb{1}(i \in \cN_C)\lambda_{i,j}(s)\FF_i(s'\mid s,a_i) \right) V(s')
\end{align*}}
for any {non-negative} {$(\lambda_{i,j}(s))_{(i,j) \in \cE_Q}$}  such that $\sum_{j \in \cE_{Q,i}} \lambda_{i,j}(s) = 1$ for all {$i\in\cN$ and $s\in \cS$}. 
For {an infinite-horizon $\gamma$-discounted} NMG with $\cG = (\cN, \cE_Q)$ with  $\cN_C = \emptyset$, if we know that $\PP(s'\mid s, \ba) = \FF_o(s'\mid s)$, and $r_{i}(s, \ba) = \sum_{j \in \cE_{Q,i}} r_{i,j}(s, a_i, a_j)$ for some $\FF_o$ and $\{r_{i,j}\}_{(i,j) \in \cE_Q}$, then
{the $Q_{i,j}^V$  given in \Cref{sec:def_MZNG} can be represented as}
\numbering{\begin{align*}
Q_{i,j}^V(s, a_i, a_j) = r_{i,j}(s,a_i, a_j) + \sum_{s' \in \cS} \gamma \left(\lambda_{i,j}(s)\FF_o(s'\mid s) \right) V(s')
\end{align*}}
\tnumbering{\begin{align*}
Q_{i,j}^V(s, a_i, a_j) = &r_{i,j}(s,a_i, a_j) + \sum_{s' \in \cS} \gamma \left(\lambda_{i,j}(s)\FF_o(s'\mid s) \right) V(s')
\end{align*}}
\noindent for any {non-negative} {$(\lambda_{i,j}(s))_{(i,j) \in \cE_Q}$}  such that $\sum_{j \in \cE_{Q,i}} \lambda_{i,j}(s) = 1$ for all $i\in\cN$ and $s\in\cS$. We call it the \textit{canonical} decomposition of $\{Q_i^V\}_{i \in \cN}$  {when $Q_{i,j}^V$ can be represented as above with} $\lambda_{i,j}(s) = {1}/{|\cE_{Q,i}|}$ for $j \in \cE_{Q, i}$.}

\end{proposition}

{We introduce this canonical decomposition since the representation of $Q_{i,j}^V$ is in general not unique, and we may use this canonical form to simplify the algorithm design later.} 
\arxiv{\numbering{
\tikzset{every picture/.style={line width=0.9pt}} %
\begin{figure}

    \centering

\begin{tikzpicture}[x=0.66pt,y=0.66pt,yscale=-1.2,xscale=1.2]

\draw   (405,80) -- (446,80) -- (446,68) -- (475,88) -- (446,108) -- (446,96) -- (405,96) -- cycle ;

\draw   (530.39,82) .. controls (530.39,71.71) and (538.74,63.36) .. (549.04,63.36) .. controls (559.33,63.36) and (567.68,71.71) .. (567.68,82) .. controls (567.68,92.29) and (559.33,100.64) .. (549.04,100.64) .. controls (538.74,100.64) and (530.39,92.29) .. (530.39,82) -- cycle ;
\draw   (478.2,138.65) .. controls (478.2,128.36) and (486.55,120.02) .. (496.84,120.02) .. controls (507.14,120.02) and (515.48,128.36) .. (515.48,138.65) .. controls (515.48,148.95) and (507.14,157.29) .. (496.84,157.29) .. controls (486.55,157.29) and (478.2,148.95) .. (478.2,138.65) -- cycle ;
\draw   (558.73,145.36) .. controls (558.73,135.07) and (567.07,126.73) .. (577.37,126.73) .. controls (587.66,126.73) and (596.01,135.07) .. (596.01,145.36) .. controls (596.01,155.66) and (587.66,164) .. (577.37,164) .. controls (567.07,164) and (558.73,155.66) .. (558.73,145.36) -- cycle ;
\draw   (592.88,72.64) .. controls (592.88,62.34) and (601.22,54) .. (611.52,54) .. controls (621.81,54) and (630.16,62.34) .. (630.16,72.64) .. controls (630.16,82.93) and (621.81,91.27) .. (611.52,91.27) .. controls (601.22,91.27) and (592.88,82.93) .. (592.88,72.64) -- cycle ;
\draw   (500,26.55) .. controls (500,16.25) and (508.35,7.91) .. (518.64,7.91) .. controls (528.94,7.91) and (537.28,16.25) .. (537.28,26.55) .. controls (537.28,36.84) and (528.94,45.18) .. (518.64,45.18) .. controls (508.35,45.18) and (500,36.84) .. (500,26.55) -- cycle ;
\draw    (528,42) -- (539,65) ;
\draw    (495.84,120.02) -- (511,43) ;
\draw    (603,88) -- (583,127) ;
\draw    (593,78) -- (567.68,82) ;
\draw   (80.13,81.89) .. controls (80.13,71.6) and (88.48,63.25) .. (98.77,63.25) .. controls (109.07,63.25) and (117.41,71.6) .. (117.41,81.89) .. controls (117.41,92.18) and (109.07,100.53) .. (98.77,100.53) .. controls (88.48,100.53) and (80.13,92.18) .. (80.13,81.89) -- cycle ;
\draw   (27.94,138.55) .. controls (27.94,128.25) and (36.28,119.91) .. (46.58,119.91) .. controls (56.87,119.91) and (65.22,128.25) .. (65.22,138.55) .. controls (65.22,148.84) and (56.87,157.18) .. (46.58,157.18) .. controls (36.28,157.18) and (27.94,148.84) .. (27.94,138.55) -- cycle ;
\draw   (108.46,145.25) .. controls (108.46,134.96) and (116.81,126.62) .. (127.11,126.62) .. controls (137.4,126.62) and (145.75,134.96) .. (145.75,145.25) .. controls (145.75,155.55) and (137.4,163.89) .. (127.11,163.89) .. controls (116.81,163.89) and (108.46,155.55) .. (108.46,145.25) -- cycle ;
\draw   (142.61,72.53) .. controls (142.61,62.23) and (150.96,53.89) .. (161.25,53.89) .. controls (171.55,53.89) and (179.89,62.23) .. (179.89,72.53) .. controls (179.89,82.82) and (171.55,91.16) .. (161.25,91.16) .. controls (150.96,91.16) and (142.61,82.82) .. (142.61,72.53) -- cycle ;
\draw   (49.74,26.44) .. controls (49.74,16.14) and (58.08,7.8) .. (68.38,7.8) .. controls (78.67,7.8) and (87.02,16.14) .. (87.02,26.44) .. controls (87.02,36.73) and (78.67,45.07) .. (68.38,45.07) .. controls (58.08,45.07) and (49.74,36.73) .. (49.74,26.44) -- cycle ;
\draw    (75.74,43.89) -- (87.74,66.89) ;
\draw    (46.58,119.91) -- (61.74,42.89) ;
\draw    (152.74,87.89) -- (132.74,126.89) ;
\draw    (142.61,72.53) -- (117.41,81.89) ;
\draw  [color={rgb, 255:red, 0; green, 0; blue, 0 }  ,draw opacity=1 ] (287.13,81.45) .. controls (287.13,71.16) and (295.48,62.82) .. (305.77,62.82) .. controls (316.07,62.82) and (324.41,71.16) .. (324.41,81.45) .. controls (324.41,91.75) and (316.07,100.09) .. (305.77,100.09) .. controls (295.48,100.09) and (287.13,91.75) .. (287.13,81.45) -- cycle ;
\draw  [draw opacity=0][fill={rgb, 255:red, 140; green, 200; blue, 255 }  ,fill opacity=1 ] (234.94,138.11) .. controls (234.94,127.82) and (243.28,119.47) .. (253.58,119.47) .. controls (263.87,119.47) and (272.22,127.82) .. (272.22,138.11) .. controls (272.22,148.4) and (263.87,156.75) .. (253.58,156.75) .. controls (243.28,156.75) and (234.94,148.4) .. (234.94,138.11) -- cycle ;
\draw   (315.46,144.82) .. controls (315.46,134.53) and (323.81,126.18) .. (334.11,126.18) .. controls (344.4,126.18) and (352.75,134.53) .. (352.75,144.82) .. controls (352.75,155.11) and (344.4,163.45) .. (334.11,163.45) .. controls (323.81,163.45) and (315.46,155.11) .. (315.46,144.82) -- cycle ;
\draw   (349.61,72.09) .. controls (349.61,61.8) and (357.96,53.45) .. (368.25,53.45) .. controls (378.55,53.45) and (386.89,61.8) .. (386.89,72.09) .. controls (386.89,82.38) and (378.55,90.73) .. (368.25,90.73) .. controls (357.96,90.73) and (349.61,82.38) .. (349.61,72.09) -- cycle ;
\draw   (256.74,26) .. controls (256.74,15.71) and (265.08,7.36) .. (275.38,7.36) .. controls (285.67,7.36) and (294.02,15.71) .. (294.02,26) .. controls (294.02,36.29) and (285.67,44.64) .. (275.38,44.64) .. controls (265.08,44.64) and (256.74,36.29) .. (256.74,26) -- cycle ;
\draw [color={rgb, 255:red, 0; green, 0; blue, 255 }  ,draw opacity=1 ]   (253.58,119.47) -- (268.74,42.45) ;
\draw   (405,80) -- (446,80) -- (446,68) -- (475,88) -- (446,108) -- (446,96) -- (405,96) -- cycle ;
\draw [color={rgb, 255:red, 0; green, 0; blue, 255 }  ,draw opacity=1 ]   (269,128) -- (294,95) ;
\draw [color={rgb, 255:red, 0; green, 0; blue, 255 }  ,draw opacity=1 ]   (271,133) -- (355,84) ;
\draw [color={rgb, 255:red, 0; green, 0; blue, 255 }  ,draw opacity=1 ]   (272.22,138.11) -- (315.46,144.82) ;
\draw [color={rgb, 255:red, 0; green, 0; blue, 255 }  ,draw opacity=1 ]   (515.48,138.65) -- (558.73,145.36) ;
\draw [color={rgb, 255:red, 0; green, 0; blue, 255 }  ,draw opacity=1 ]   (512,127) -- (536,95) ;
\draw [color={rgb, 255:red, 0; green, 0; blue, 255 }  ,draw opacity=1 ]   (502,121) -- (516.64,45.18) ;
\draw [color={rgb, 255:red, 0; green, 0; blue, 255 }  ,draw opacity=1 ]   (514,133) -- (597,83) ;
\draw    (85,96) -- (58,124) ;
\draw    (107,99) -- (119,128) ;
\draw    (533,90) -- (509,123) ;
\draw    (555,100) -- (570,129) ;

\draw (542.78,76.15) node [anchor=north west][inner sep=0.75pt]  [font=\footnotesize]  {$1$};
\draw (512.39,20.95) node [anchor=north west][inner sep=0.75pt]  [font=\footnotesize]  {$2$};
\draw (605.76,67.04) node [anchor=north west][inner sep=0.75pt]  [font=\footnotesize]  {$3$};
\draw (490.59,133.25) node [anchor=north west][inner sep=0.75pt]  [font=\footnotesize]  {$5$};
\draw (571.07,139.51) node [anchor=north west][inner sep=0.75pt]  [font=\footnotesize]  {$4$};
\draw (92.52,76.04) node [anchor=north west][inner sep=0.75pt]  [font=\footnotesize]  {$1$};
\draw (62.12,20.84) node [anchor=north west][inner sep=0.75pt]  [font=\footnotesize]  {$2$};
\draw (155.5,66.93) node [anchor=north west][inner sep=0.75pt]  [font=\footnotesize]  {$3$};
\draw (40.32,133.15) node [anchor=north west][inner sep=0.75pt]  [font=\footnotesize]  {$5$};
\draw (120.81,139.4) node [anchor=north west][inner sep=0.75pt]  [font=\footnotesize]  {$4$};
\draw (300.52,75.6) node [anchor=north west][inner sep=0.75pt]  [font=\footnotesize]  {$1$};
\draw (269.12,20.4) node [anchor=north west][inner sep=0.75pt]  [font=\footnotesize]  {$2$};
\draw (362.5,66.49) node [anchor=north west][inner sep=0.75pt]  [font=\footnotesize]  {$3$};
\draw (247.32,132.71) node [anchor=north west][inner sep=0.75pt]  [font=\footnotesize]  {$5$};
\draw (327.81,138.96) node [anchor=north west][inner sep=0.75pt]  [font=\footnotesize]  {$4$};
\draw (88,172.4) node [anchor=north west][inner sep=0.75pt]    {$\cE_{r}$};
\draw (270,172.4) node [anchor=north west][inner sep=0.75pt]    {$\cE_{\cN_{P}, \text{connected}}$};
\draw (541,170.4) node [anchor=north west][inner sep=0.75pt]    {$\cE_Q$};
\draw [line width=3]  (190,88)  node [anchor=north west][inner sep=0.75pt]    {\huge $\cup$};
\end{tikzpicture}
\caption{Relationship between $\cE_{r}$, $\cE_{\cN_{P}, \text{connected}}$, and $\cE_Q$. Here, {we define $\cE_{\cN_P, \text{connected}}:= \{(i,j)\,\given \, i \in \cN_P \text{ or } j \in \cN_P, i \neq j\}$.}   $r_i(s, a_i, \cdot)$ is decomposable with respect to $\cE_{r,i}$ for all $i \in \cN$, and $\PP(s'\given s, \cdot)$ is decomposable with respect to $\cN_P$ (See \Cref{rmk:stronger}). The transition dynamics $\PP$ is expressed as the ensemble of controllers in the set  $\cN_P = \{5\}$ while the $\cN_C$ in this case is $\{1,5\}$.}  
\label{fig:prop1}
\vspace{-7pt}
\end{figure}
}}
\arxiv{}

\numbering{\subsection{Examples of multi-player (zero-sum) NMGs}}
\tnumbering{\section{Examples of multi-player (zero-sum) NMGs}}

\neurips{
\vspace{-0.07in}
}

\arxiv{We now provide several examples of (multi-player) MGs with networked separable interactions.}  
\neurips{We now provide several examples of (multi-player) MGs with networked separable interactions here and in \Cref{ssec:ex}.}

\neurips{
\vspace{-0.07in}
}

\paragraph{Example 1 (Markov fashion games).} Fashion games are an intriguing class of games \cite{cao2013fashion, cao2014fashion,zhang2018fashion} that plays a vital role not only in Economics theory but also in practice. A fashion game is a networked extension of the Matching Pennies game, in which each player has the action space $\cA_i = \{-1, +1\}$, which means \emph{light} and \emph{dark} color fashions, respectively, for example. There are two types of players: \emph{conformists} ($\mathfrak{c}$), who prefer to conform  to their neighbors' fashion (action), and \emph{rebels} ($\mathfrak{r}$), who prefer to oppose their neighbors' fashion (action). Such interactions with the neighbors are exactly captured by polymatrix games. {We denote the interaction {network} between players as $\cG = (\cN, \cE)$. }

{Such a game naturally involves the following state transition dynamics:} 
we {introduce the} state $s\in\cS = \ZZ$ by setting $s_0 = 0$ and $s_{t+1} \sim s_{t} + \text{Unif}((a_{t, c})_{c \in \cC})$, {which indicates the \emph{fashion trend} where {$\cC\subseteq \cN$ is the  set of influencers.}} {The fashion trend favors either light or dark colors if $s\geq 0$ or $s <0$, respectively. We can think of dynamics as {the  impact of the influencers on the fashion trend at time $t$}.} For each $(s,\ba)$, the reward {function for player $i$, depending on whether she is a} conformist or a rebel, are defined as  $r_{\mathfrak{c}, i}(s, \ba) = \sum_{j \in \cE_i}r_{\mathfrak{c}, i, j}(s, a_i, a_j) = \sum_{j \in \cE_i} (\frac{1}{|\cE_i|} \pmb{1}(\text{sgn}(s) = a_i)+ \pmb{1}(a_i = a_j))$ and $r_{\mathfrak{r}, i}(s, \ba) = \sum_{j \in \cE_i}r_{\mathfrak{r}, i, j}(s, a_i, a_j) = \sum_{j \in \cE_i} (\frac{1}{|\cE_i|} \pmb{1}(\text{sgn}(s) \neq a_i)+ \pmb{1}(a_i \neq a_j))$, respectively. This is  an NMG {as defined in \Cref{sec:def_MZNG}}. Moreover, if the conformists and rebels constitute a bipartite graph, {i.e., the neighbors of a conformist are all rebels and vice versa, it becomes a multi-player constant-sum MG with networked separable interactions, and we can subtract the constant offset to make it a zero-sum NMG. }
 
\arxiv{\paragraph{Example 2 (Markov security games). }  Security games as described in  \cite{grossklags2008secure,cai2016zero} is  a {primary} example of zero-sum NGs/polymatrix games,  which features two types of players: \emph{attackers} who work as a group ($\mathfrak{a}$), and \emph{users} ($\mathfrak{u}$). Let $\mathfrak{U}$ denote the set of all users. We construct a star-shaped network (c.f. \Cref{fig:PPAD Hardness}) {with the attacker group including   $n_\mathfrak{a}$ number of attackers} sitting at the center, connected to each user. There is an IP address set $[C]$. {We define the action spaces for each user $\mathfrak{u}_i$ and the attacker group as $\cA_{\mathfrak{u}_i} = [C]$ and $\cA_{\mathfrak{a}} = \{T \mid T \subseteq [C],\, |T| = n_{\mathfrak{a}}\}$, respectively.}  Each user selects one IP address, while the attacker group selects a subset $I \subseteq [C]$. For each user whose IP address is attacked, the attacker group {gains one unit of payoff, and the attacked user loses one unit. Conversely, if a user's IP address is not attacked, the user earns one unit of payoff, and the attacker loses one unit.}

{We naturally extend} the security games to \emph{Markov security games} as follows: we define state $s \in \cS = \RR^C $ by setting $s_0 = \pmb{0}$ and $s_{t+1} \sim s_{t} + \text{Unif}((e_{a_{\mathfrak{u}_i, t}})_{\mathfrak{u}_i \in \mathfrak{U}})$, representing the vector of \emph{security level} for the IP addresses. {Specifically}, a vaccine program can improve the security level of each IP address if it has been attacked previously. {We define $X\in\RR^C$ as a vector such that each of its components, $X_c$, corresponds to a unique user's IP address, indexed by $c$. Each $X_c$ is defined by the random variable as $X_c \sim 2\text{Bern}(1- 1/(s_{t, c} + 1))-1$, indicating the outcome of a potential attack on IP address $c \in [C]$. Here, $s_{t, c}$ denotes the security level of each IP address $c$ at a given time $t$, i.e., the $c$-th component of $s_t$. The success probability of an attack on an IP address is inversely proportional to its security level, represented by $1/(s_{t, c} + 1)$. Therefore, higher security levels make an attack less likely to succeed. The term $2\text{Bern}({1}- 1/(s_{t, c} + 1))-{1}$ describes a Bernoulli distribution, typically taking values $0$ or $1$, that has been scaled and shifted to take values $-1$ or $1$ instead. Here, $-1$ represents an unsuccessful attack, while $1$ denotes  a successful attack on the IP address $c$. Therefore, each $X_c$ provides a probabilistic view of the failure of an attack on each IP address, given its security level.}
For each $(s,\ba)$, the reward functions for the users and the attacker group are defined as $r_{\mathfrak{u}_i} (s, \ba, I) = r_{\mathfrak{u}_i, \mathfrak{a}}(s, a_{\mathfrak{u}_i}, I) =  \pmb{1} (a_{\mathfrak{u}_i} \in I) X_{a_{\mathfrak{u}_i}} + \pmb{1} (a_{\mathfrak{u}_i} \notin I)$ and $r_{\mathfrak{a}}(s, \ba, I)= \sum_{\mathfrak{u}_i \in \mathfrak{U}}r_{\mathfrak{a}, \mathfrak{u}_i}(s, a_{\mathfrak{u}_i}, I) - \pmb{1} (a_{\mathfrak{u}_i} \notin I)$ {where} $r_{\mathfrak{a}, \mathfrak{u}_i}(s, a_{\mathfrak{u}_i}, I) = -\pmb{1} (a_{\mathfrak{u}_i} \in I)X_{a_{\mathfrak{u}_i}}$. The reward function of users can be interpreted as follows: if the user's action $a_{\mathfrak{u}_i}$ is in the set of attacked IP addresses $I$ and the attack failed (i.e., $X_{a_{\mathfrak{u}_i}} = 1$), then the user receives a reward equal to $1$. Otherwise, if the user's action is not in $I$, the user also receives a reward of $1$, likely representing a successful defense or evasion of an attack.
Since the reward is always zero-sum, this game is a zero-sum NMG with networked separable interactions.

\paragraph{Example 3 (Global {economy})}.
Macroeconomic dynamics may also be modeled through either zero-sum NMGs or NMGs. Trading between nations has been analyzed in game theory \cite{wilson1985game,wilson1989game}. 
We consider nations as players, each nation has an action space, $\cA_i = \RR$, and the actions decide their expenditure levels. We define the {\it state} of the global economy,  $s \in \RR$, such that $s_0 = 0$ and $s_{t+1} \sim s_{t} + \text{Unif}((a_{c, t})_{c \in \cC}) + Z_t$. Here, $Z_t$ is a random variable representing the unpredictable nature of global events (e.g., COVID-19), and $\cC$ represents the set of {\it powerful} nations, which models the fact that powerful nations' politics or military spending have a relatively significant impact on global economy  \cite{zhang2019economic,beckley2018power}. The aggregated (or ensemble) effect of the powerful nations on the economy is modeled by the term $\text{Unif}((a_{c, t})_{c \in \cC})$. 

During the global financial crisis in 2008-2009, many nations implemented significant fiscal stimulus measures to counteract the downturn \cite{freedman2010global, armingeon2012politics}. Conversely, in good economic conditions, the estimated government spending multipliers were less than one, suggesting that the increased government spending in such situations  might not have the intended positive effects on the economy \cite{gomes2022government}. {Such a state-dependence on reward functions may be modeled as follows.} 
First, we consider the reward being decomposable with respect to nations, as it can be interpreted as (1) the expenditure of each nation is related to the amount of {payment spent} on trading, and (2) we focus on the case with {\it bilateral}  trading, where the surplus from trading can be decomposed by the surplus from the {\it pairwise} trading with other nations. Second, as mentioned above, the relationship between government spending and the global economy can be seen {as {\it countercyclical} \cite{gomes2022government}, which we use the formula $s(a_j - a_i)$ to model explicitly, for nation $i$. Specifically, $s>0$ denotes a good economic condition, in which all the nations may choose to decrease the expenditure level (the $-a_i$ term). 
Hence, the reward function for nation $i$ can be written  as $r_{i}(s, \ba) = \sum_{j \in \cN} r_{i,j} (s, a_i, a_j)= \texttt{Const} + \sum_{j \in \cN}s (a_j - a_i)$, where the positive constant \texttt{Const}  represents  the net benefit out of the tradings. 
Hence, the game shares the characteristics of being a constant-sum NMG. Moreover, other alternative forms of the reward functions may exist to reflect the countercyclical phenomenon, and may not necessarily satisfy the zero-sum {(constant-sum)}  property, but the game would still qualify as an NMG.}}

\neurips{
\vspace{-0.06in}
}

\tnumbering{\section{Relationship between  CCE and NE in zero-sum NMGs}}
\numbering{\subsection{Relationship between  CCE and NE in zero-sum NMGs}}
\label{sec:CCENE}

\neurips{
\vspace{-0.06in}
}

A well-known property for zero-sum NGs is that marginalizing a CCE leads to a NE, which makes it computationally tractable to find the NE  \cite{cai2011minmax,cai2016zero}. 
We now provide below a counterpart in the Markov game setting, and provide a more detailed statement of the result  in \Cref{appendix:MPGdef}. 

\begin{restatable}{proposition}{CCENE}
\label{prop:MZNMG}
{Given {an $\epsilon$-approximate} Markov CCE of an infinite-horizon $\gamma$-discounted zero-sum NMG}, marginalizing it 
{at each state} results in an {$\frac{(n+1)}{(1-\gamma)}\epsilon$-approximate} Markov NE of the zero-sum NMG. The same argument also holds for the finite-horizon episodic setting {with $(1-\gamma)^{-1}$ being replaced by $H$}. 
\end{restatable} 

This result holds for both stationary and non-stationary {$\epsilon$-approximate} Markov CCEs. We defer the proof {of \Cref{prop:MZNMG}} {to}  \Cref{appendix:MPGdef}. 
This proposition suggests that if we can have some algorithms to find {an approximate} Markov CCE for a zero-sum NMG, we can obtain {an approximate} Markov NE by marginalizing the {approximate} CCE {at each state}. 
We also emphasize that the \emph{Markovian} property of the equilibrium policies is important for the result to hold. 
As a result, the learning algorithm in \cite{daskalakis2022complexity}, which learns an approximate Markov \emph{non-stationary}  CCE with polynomial time and samples, may thus be used to find an approximate Markov \emph{non-stationary} NE in zero-sum NMGs. However, as the focus of \cite{daskalakis2022complexity} was the {more challenging setting of \emph{model-free learning}, the complexity therein has a high dependence on the problem parameters, and the algorithm can only find non-perfect equilibria.} When it comes to (perfect) equilibrium computation, one may exploit the multi-player zero-sum structure of zero-sum NMGs, and develop more natural and faster algorithms to find a Markov {non-stationary}  NE. Moreover, when it comes to \emph{stationary} equilibrium computation, even Markov CCE is \emph{not} tractable in general-sum cases \cite{daskalakis2022complexity,jin2022complexity}. 
Hereafter, we will focus on approaching zero-sum NMGs from these perspectives.

\neurips{
\vspace{-0.10in}
}
\numbering{\section{Hardness for Stationary CCE Computation}}
\tnumbering{\chapter{Hardness for Stationary CCE Computation}}
\label{subsection:PPAD}

\neurips{
\vspace{-0.10in}
}

Given the results in \Cref{sec:CCENE}, it seems tempting and sufficient to compute the Markov CCE of the zero-sum NMG. 
Indeed, computing CCE (and thus NE) in zero-sum polymatrix games is known to be tractable \cite{cai2011minmax,cai2016zero}. It is thus natural to ask: \emph{Is finding Markov CCE computationally tractable?}  
Next, we answer the question with different answers  for finding stationary CCE (in infinite-horizon $\gamma$-discounted setting) and non-stationary CCE (in finite-horizon episodic setting), respectively.  

For \emph{two-player} infinite-horizon $\gamma$-discounted \emph{zero-sum} MGs, significant progress  in computing/learning the (Markov) stationary NE has been made recently \cite{daskalakis2020independent, wei2021last, zhao2021provably, chen2021sample, cen2021fast, alacaoglu2022natural, zeng2022regularized, cen2022faster}. On the other hand, for \emph{multi-player general-sum} MGs, recent results in \cite{daskalakis2022complexity,jin2022complexity} showed that computing (Markov) stationary CCE can be {\tt PPAD}-hard and thus believed to be computationally intractable. We next show that this hardness persists in most non-degenerate cases even if one enforces the zero-sum and networked interaction structures in the multi-player case. We state the formal result as follows, whose detailed proof  is available in \Cref{appendix:PPADhard-proof}.  
{
\begin{restatable}{theorem}{PPADHardmain}
\label{thm:PPAD-hard-main}
There is a constant $\epsilon > 0$ for which computing an $\epsilon$-approximate  Markov {perfect} stationary CCE in {infinite-horizon}  $\frac{1}{2}$-discounted zero-sum NMGs, whose underlying network structure contains either a triangle or a 3-path subgraph, is {\tt PPAD}-hard. Moreover, given the PCP for {\tt PPAD} conjecture \cite{babichenko2015can}, there is a constant $\epsilon > 0$ such that computing even an $\epsilon$-approximate Markov {non-perfect} stationary CCE  in such zero-sum NMGs is {\tt PPAD}-hard. 
\end{restatable}} 

\neurips{
\vspace{-0.10in}
}

\begin{proof}[Proof Sketch of \Cref{thm:PPAD-hard-main}.] {Due to space constraints, we focus on the case with three players, and {the   underlying network structure has a triangle subgraph}. Proof  for the $3$-path case is similar and can be found in \Cref{appendix:PPADhard-proof}.}   
We will show  that for \emph{any} general-sum two-player \emph{turn-based} MG \textbf{(A)}, the problem of computing its Markov stationary CCE,  which is inherently a {\tt PPAD}-hard problem \cite{daskalakis2022complexity}, can be reduced to computing the Markov stationary CCE of a three-player zero-sum MG {with a triangle  structure} networked separable interactions  \textbf{(B)}. Consider an MG \textbf{(A)} with two players, players 1 and 2, and reward functions $r_1(s, a_1, a_2)$ and $r_2(s, a_2, a_1)$, where $a_i$ is the action of the $i$-th player and $r_i$ is the reward function of the $i$-th player. The transition dynamics is given by {$\PP(s'\,|\,s, a_1, a_2)$}. In even {rounds}, player 2's action space is limited to \textsf{Noop2}, and in odd {rounds}, player 1's action space is limited to \textsf{Noop1}, where \textsf{Noop} is an abbreviation of {``no-operation'', i.e., the player does not affect the transition dynamics {or the reward} in that round}. We denote player 1's action space in even rounds as $\cA_{1, \text{even}}$ and player 2's action space in odd {rounds} as $\cA_{2, \text{odd}}$, respectively.    

Now, we construct a {three-player}  zero-sum NMG. {with a triangle network structure.} We set the reward function as $\tilde{r}_i(s, \ba) = \sum_{j \neq i}\tilde{r}_{i,j}(s, a_i, a_j) $ and $\tilde{r}_{i,j}(s,a_i,a_j)= -\tilde{r}_{j,i}(s,a_j,a_i)$. The reward functions   are designed so that $\tilde{r}_{i,j} = -\tilde{r}_{j,i}$ for all $i, j$, $\tilde{r}_{1,2} + \tilde{r}_{1,3} = r_1$, and $\tilde{r}_{2,1} + \tilde{r}_{2,3} = r_2$, {where $r_1,r_2$ are the reward functions in game \textbf{(A)},} {by introducing} a dummy player{, player 3}. In even rounds, player 2's action space is limited to \textsf{Noop2}, and in odd rounds, player 1's action space is limited to \textsf{Noop1}. Player 3's action space is always limited to \textsf{Noop3} in all  rounds. The transition dynamics is defined as {$\tilde{\PP}(s'\,|\,s, a_1, a_2, a_3) = {\PP}(s'\,|\,s, a_1, a_2)$}, since $a_3$ is always chosen from \textsf{Noop3}. {In other words, player 3's action does not affect the rewards of the other two players, nor the transition dynamics, and players 1 and 2 will receive the reward as in the two-player turn-based MG. Also, }note that due to the turn-based structure of the game \textbf{(A)}, the transition dynamics satisfy the decomposable condition in our \Cref{prop:MPGcond}, and it is thus a zero-sum NMG. {In fact,  turn-based dynamics can be represented as an ensemble of single controller dynamics, as we have discussed in \Cref{sec:definition}.}   

Note that the new game \textbf{(B)} is still a turn-based game, and thus the Markov stationary CCE is the same as the Markov stationary NE. Also, note that by construction, the equilibrium policies of players $1$ and $2$ at the   Markov stationary CCE of the game \textbf{(B)} constitute a Markov stationary CCE of the game \textbf{(A)}. 
{If {the underlying network} is more general than a triangle, but contains a triangle subgraph, we can specify the reward and transition dynamics of these three players as above, and specify all other players to be dummy players, whose reward functions are all zero, and do not affect the reward functions of these three players, nor the transition dynamics. 
This completes the proof.} 
\end{proof}

\neurips{
\vspace{-0.10in}
}

\Cref{fig:PPAD Hardness} briefly explains how we may  reduce {the equilibrium computation problem of } \textbf{(A)} to {that of} \textbf{(B)}. In fact, a {connected} graph that does not {contain a subgraph of} a triangle or a 3-path {has to be} a \emph{star-shaped}  network   (\Cref{prop:star-shaped-only}), {which is proved in \Cref{appendix:PPADhard-proof}.} {Hence, by \Cref{thm:PPAD-hard-main}, we know that} in the infinite-horizon discounted setting, finding Markov stationary NE/CE/CCE is a computationally hard problem unless the underlying network is star-shaped. 
{This may also imply that \emph{learning} Markov stationary NE in zero-sum NMGs, e.g., using natural dynamics like fictitious play to reach the NE, can be challenging, unless in the star-shaped case. In turn, one may hope fictitious-play dynamics to converge for star-shaped zero-sum NMGs.  
We instantiate this idea next in  \Cref{sec:fictitious-play}. 
Furthermore, in light of \Cref{thm:PPAD-hard-main},  we will shift gear to computing Markov \emph{non-stationary} NE by utilizing the structure of  networked separable interactions, as to be detailed in  \Cref{section:OMWU}.}
\arxiv{
\tikzset{every picture/.style={line width=0.9pt}} %

\begin{figure}
    \centering 

\begin{tikzpicture}[x=0.72pt,y=0.72pt,yscale=-1,xscale=1]

\draw    (116,43) -- (167,128) ;
\draw    (54.64,146) -- (156.36,146) ;
\draw    (45,129) -- (89,44) ;
\draw    (259,49.64) -- (260,127.36) ;
\draw    (380.36,146) -- (278.64,146) ;
\draw    (399,127.36) -- (398,48.64) ;
\draw   (530.39,82) .. controls (530.39,71.71) and (538.74,63.36) .. (549.04,63.36) .. controls (559.33,63.36) and (567.68,71.71) .. (567.68,82) .. controls (567.68,92.29) and (559.33,100.64) .. (549.04,100.64) .. controls (538.74,100.64) and (530.39,92.29) .. (530.39,82) -- cycle ;
\draw   (478.2,138.65) .. controls (478.2,128.36) and (486.55,120.02) .. (496.84,120.02) .. controls (507.14,120.02) and (515.48,128.36) .. (515.48,138.65) .. controls (515.48,148.95) and (507.14,157.29) .. (496.84,157.29) .. controls (486.55,157.29) and (478.2,148.95) .. (478.2,138.65) -- cycle ;
\draw   (558.73,145.36) .. controls (558.73,135.07) and (567.07,126.73) .. (577.37,126.73) .. controls (587.66,126.73) and (596.01,135.07) .. (596.01,145.36) .. controls (596.01,155.66) and (587.66,164) .. (577.37,164) .. controls (567.07,164) and (558.73,155.66) .. (558.73,145.36) -- cycle ;
\draw   (592.88,72.64) .. controls (592.88,62.34) and (601.22,54) .. (611.52,54) .. controls (621.81,54) and (630.16,62.34) .. (630.16,72.64) .. controls (630.16,82.93) and (621.81,91.27) .. (611.52,91.27) .. controls (601.22,91.27) and (592.88,82.93) .. (592.88,72.64) -- cycle ;
\draw   (470,33.55) .. controls (470,23.25) and (478.35,14.91) .. (488.64,14.91) .. controls (498.94,14.91) and (507.28,23.25) .. (507.28,33.55) .. controls (507.28,43.84) and (498.94,52.18) .. (488.64,52.18) .. controls (478.35,52.18) and (470,43.84) .. (470,33.55) -- cycle ;
\draw    (501.32,45.85) -- (532.63,72.68) ;
\draw    (508.03,123.37) -- (534.87,94.3) ;
\draw    (555.75,98.77) -- (569.17,127.85) ;
\draw    (592.88,72.64) -- (567.68,82) ;
\draw   (17.36,146) .. controls (17.36,135.71) and (25.71,127.36) .. (36,127.36) .. controls (46.29,127.36) and (54.64,135.71) .. (54.64,146) .. controls (54.64,156.29) and (46.29,164.64) .. (36,164.64) .. controls (25.71,164.64) and (17.36,156.29) .. (17.36,146) -- cycle ;
\draw   (156.36,146) .. controls (156.36,135.71) and (164.71,127.36) .. (175,127.36) .. controls (185.29,127.36) and (193.64,135.71) .. (193.64,146) .. controls (193.64,156.29) and (185.29,164.64) .. (175,164.64) .. controls (164.71,164.64) and (156.36,156.29) .. (156.36,146) -- cycle ;
\draw  [color={rgb, 255:red, 0; green, 0; blue, 0 }  ,draw opacity=0.1 ][fill={rgb, 255:red, 0; green, 0; blue, 0 }  ,fill opacity=0.1 ] (84.36,31) .. controls (84.36,20.71) and (92.71,12.36) .. (103,12.36) .. controls (113.29,12.36) and (121.64,20.71) .. (121.64,31) .. controls (121.64,41.29) and (113.29,49.64) .. (103,49.64) .. controls (92.71,49.64) and (84.36,41.29) .. (84.36,31) -- cycle ;
\draw   (241.36,146) .. controls (241.36,135.71) and (249.71,127.36) .. (260,127.36) .. controls (270.29,127.36) and (278.64,135.71) .. (278.64,146) .. controls (278.64,156.29) and (270.29,164.64) .. (260,164.64) .. controls (249.71,164.64) and (241.36,156.29) .. (241.36,146) -- cycle ;
\draw   (380.36,146) .. controls (380.36,135.71) and (388.71,127.36) .. (399,127.36) .. controls (409.29,127.36) and (417.64,135.71) .. (417.64,146) .. controls (417.64,156.29) and (409.29,164.64) .. (399,164.64) .. controls (388.71,164.64) and (380.36,156.29) .. (380.36,146) -- cycle ;
\draw  [color={rgb, 255:red, 0; green, 0; blue, 0 }  ,draw opacity=0.1 ][fill={rgb, 255:red, 0; green, 0; blue, 0 }  ,fill opacity=0.1 ] (240.36,31) .. controls (240.36,20.71) and (248.71,12.36) .. (259,12.36) .. controls (269.29,12.36) and (277.64,20.71) .. (277.64,31) .. controls (277.64,41.29) and (269.29,49.64) .. (259,49.64) .. controls (248.71,49.64) and (240.36,41.29) .. (240.36,31) -- cycle ;
\draw  [color={rgb, 255:red, 0; green, 0; blue, 0 }  ,draw opacity=0.1 ][fill={rgb, 255:red, 0; green, 0; blue, 0 }  ,fill opacity=0.1 ] (379.36,30) .. controls (379.36,19.71) and (387.71,11.36) .. (398,11.36) .. controls (408.29,11.36) and (416.64,19.71) .. (416.64,30) .. controls (416.64,40.29) and (408.29,48.64) .. (398,48.64) .. controls (387.71,48.64) and (379.36,40.29) .. (379.36,30) -- cycle ;

\draw (31,139.4) node [anchor=north west][inner sep=0.75pt]  [font=\footnotesize]  {$1$};
\draw (170,139.4) node [anchor=north west][inner sep=0.75pt]  [font=\footnotesize]  {$2$};
\draw (97,24) node [anchor=north west][inner sep=0.75pt]  [font=\footnotesize]  {$3$};
\draw (255,139.4) node [anchor=north west][inner sep=0.75pt]  [font=\footnotesize]  {$1$};
\draw (394,139.4) node [anchor=north west][inner sep=0.75pt]  [font=\footnotesize]  {$2$};
\draw (253.5,24) node [anchor=north west][inner sep=0.75pt]  [font=\footnotesize]  {$3$};
\draw (392,24) node [anchor=north west][inner sep=0.75pt]  [font=\footnotesize]  {$4$};
\draw (58,148.4) node [anchor=north west][inner sep=0.75pt]    {$\tilde{r}_{1,2}$};
\draw (126,148.4) node [anchor=north west][inner sep=0.75pt]    {$\tilde{r}_{2,1}$};
\draw (19,103.4) node [anchor=north west][inner sep=0.75pt]    {$\tilde{r}_{1,3}$};
\draw (51,43.4) node [anchor=north west][inner sep=0.75pt]    {$\tilde{r}_{3,1}$};
\draw (126,43.4) node [anchor=north west][inner sep=0.75pt]    {$\tilde{r}_{3,2}$};
\draw (161,103.4) node [anchor=north west][inner sep=0.75pt]    {$\tilde{r}_{2,3}$};
\draw (75.5,-6) node [anchor=north west][inner sep=0.75pt]  [font=\footnotesize] [align=left] {Dummy};
\draw (231,-6) node [anchor=north west][inner sep=0.75pt]  [font=\footnotesize] [align=left] {Dummy};
\draw (371,-6) node [anchor=north west][inner sep=0.75pt]  [font=\footnotesize] [align=left] {Dummy};
\draw (283,148.4) node [anchor=north west][inner sep=0.75pt]    {$\tilde{r}_{1,2}$};
\draw (351,148.4) node [anchor=north west][inner sep=0.75pt]    {$\tilde{r}_{2,1}$};
\draw (230,103.4) node [anchor=north west][inner sep=0.75pt]    {$\tilde{r}_{1,3}$};
\draw (230,53) node [anchor=north west][inner sep=0.75pt]    {$\tilde{r}_{3,1}$};
\draw (400,103.4) node [anchor=north west][inner sep=0.75pt]    {$\tilde{r}_{2,4}$};
\draw (400,53) node [anchor=north west][inner sep=0.75pt]    {$\tilde{r}_{4,2}$};
\draw (543.78,76.15) node [anchor=north west][inner sep=0.75pt]  [font=\footnotesize]  {$1$};
\draw (516,39) node [anchor=north west][inner sep=0.75pt]  [font=\footnotesize] [align=left] {center-player};
\draw (483.39,27.95) node [anchor=north west][inner sep=0.75pt]  [font=\footnotesize]  {$2$};
\draw (607.26,67.04) node [anchor=north west][inner sep=0.75pt]  [font=\footnotesize]  {$3$};
\draw (491.59,134.05) node [anchor=north west][inner sep=0.75pt]  [font=\footnotesize]  {$5$};
\draw (572.37,140.51) node [anchor=north west][inner sep=0.75pt]  [font=\footnotesize]  {$4$};

\end{tikzpicture}

  \caption{(Left, Middle): \texttt{PPAD}-hardness reduction visualization {of $\cE_Q$}. (Right): A star-shaped zero-sum NMG.
  } 
  \label{fig:PPAD Hardness}
\end{figure}
}
\numbering{\section{Fictitious-Play Property}}
\tnumbering{\chapter{Fictitious-Play Property}}
\label{sec:fictitious-play}
 
\neurips{
\vspace{-0.10in}
}

In this section, we study the fictitious-play property of multi-player zero-sum games with networked separable interactions, for both the matrix and  
Markov game settings. Following the convention in \cite{monderer1996fictitious}, we refer to the games in which fictitious-play dynamics converge to the NE as the games that have the \emph{fictitious-play property}.  
We defer the matrix game case results to \Cref{appendix:fictitious}, where we have also established convergence of the well-known variant of FP, smooth FP \cite{ref:Fudenberg93}, in zero-sum NGs.  

\neurips{
\vspace{-0.03in}
}

Echoing the computational intractability of computing CCE  of zero-sum NMG unless the underlying network  structure is 
star-shaped  
in the infinite-horizon discounted setting  
(c.f. \Cref{thm:PPAD-hard-main}), we now consider the FP property in such games. 
{Note that  by \Cref{prop:MPGcond}, $\cE_Q$ is a star-shape if and only if the reward structure is a star shape and $\cN_C = \{1\}$, where player 1 is the center of the star (\Cref{fig:PPAD Hardness}), or there are only two players in zero-sum NMG. }There is already existing literature for the latter case \cite{sayin2021decentralized, sayin2022fictitious}, so we focus on the former case, which is a single-controller case where player $1$ controls the transition dynamics, i.e., $\PP(s'\given s, \ba)= \PP_1 (s'\given s, a_1)$ {for some $\PP_1$}. 
We now introduce the fictitious-play dynamics for such zero-sum NMGs. 

\neurips{
\vspace{-0.05in}
}

Each player  $i$ first initializes her beliefs of other players'  policies  as uniform distributions, and also  initializes her belief of the $Q$-value estimates  with arbitrary values. Then,  at iteration $k$, player $i$ takes the  \emph{best-response} action based on her belief of other {players' policies}  $(\hat{\pi}_{-i}^{(k)}(s^{(k)}))$, and their $Q$ beliefs $\hat{Q}_i^{(k)}(s^{(k)}, \ba)$: 
$$
a_{i}^{(k)} \in \argmax_{a_i\in\cA_i}~\hat{Q}_i^{(k)}(s^{(k)}, e_{a_i}, \hat{\pi}_{-i}^{(k)}(s^{(k)})).$$   
{Then, player $i$ implements the action $a_{i}^{(k)}$, observes other players' actions $a_{-i}^{(k)}$, and updates her beliefs as follows: for each  player  $i\in\cN$, she updates her belief of   the opponents' policies  as 
$$\hat{\pi}_{-i}^{(k+1)}(s) = \hat{\pi}_{-i}^{(k)}(s)  +  \pmb{1}(s =s^{(k)}) {\alpha^{N(s)}} (e_{a_{-i}^{(k)}} - \hat{\pi}_{-i}^{(k)}(s))$$ 
for all $s \in \cS$, with stepsize $\alpha^{N(s)}\geq 0$ where $N(s)$ is the visitation count for the state $s$}; then if $i=1$, this player $1$ 
updates the belief of $Q_{1,j}$ for all $j \in \cN/\{1\}$ and  her own $\hat{Q}_{1}(s, \ba)$ for all $s\in \cS$ as   

\neurips{
\vspace{-0.07in}
}

\arxiv{\numbering{$$
  \hat{Q}_{1,j}^{(k+1)}(s, a_1, a_j) = \hat{Q}_{1,j}^{(k)}(s,a_1, a_j)  +  \pmb{1}(s =s^{(k)}) {\beta^{N(s)}}\Bigl( r_{1,j}(s, a_1, a_j) + \gamma \sum_{{s}' \in \cS} \frac{\PP_1(s' \mid s, a_1) }{n-1} \cdot\hat{V}_{1}^{(k)}({s}') - \hat{Q}_{1,j}^{(k)}(s,a_1, a_i) \Bigr),
$$}}
\tnumbering{
\begin{align*}
      \hat{Q}_{1,j}^{(k+1)}&(s, a_1, a_j) = \hat{Q}_{1,j}^{(k)}(s,a_1, a_j)  
      \\
      &+  \pmb{1}(s =s^{(k)}) {\beta^{N(s)}}\Bigl( r_{1,j}(s, a_1, a_j) + \gamma \sum_{{s}' \in \cS} \frac{\PP_1(s' \mid s, a_1) }{n-1} \cdot\hat{V}_{1}^{(k)}({s}') - \hat{Q}_{1,j}^{(k)}(s,a_1, a_i) \Bigr),
\end{align*}
  }
\neurips{
\begin{align*}
      \hat{Q}_{1,j}^{(k+1)}&(s, a_1, a_j) = \hat{Q}_{1,j}^{(k)}(s,a_1, a_j)  
      \\
      &+  \pmb{1}(s =s^{(k)}) {\beta^{N(s)}}\Bigl( r_{1,j}(s, a_1, a_j) + \gamma \sum_{{s}' \in \cS} \frac{\PP_1(s' \mid s, a_1) }{n-1} \cdot\hat{V}_{1}^{(k)}({s}') - \hat{Q}_{1,j}^{(k)}(s,a_1, a_i) \Bigr),
\end{align*}
  }

\neurips{
\vspace{-0.12in}
}

which is based on the canonical decomposition given in \Cref{def:canonical}, where $\hat{V}_{1}^{(k)}(s) = \max_{a_1\in\cA_1} \hat{Q}_1^{(k)}(s, e_{a_1}, \hat{\pi}_{-1}^{(k)}(s))$, and {$\beta^{N(s)}\geq 0$} is the stepsize. The agent then updates  $\hat{Q}_{1}^{(k+1)}(s, \ba) = \sum_{j \in \cN/\{1\} } \hat{Q}_{1,j}^{(k+1)}(s,a_1, a_j)$, for all $s \in \cS, \ba\in\cA$.  Otherwise, if $i\neq 1$, then player $i$ updates the belief of her $\hat{Q}_{i,1}(s,\ba)$ for all $s \in \cS, \ba\in\cA$ as 

\arxiv{ \numbering{
$$
  \hat{Q}_{i,1}^{(k+1)}(s, a_i, a_1) = \hat{Q}_{i,1}^{(k)}(s,a_i, a_1)  +  \pmb{1}(s =s^{(k)}) {\beta^{N(s)}}\Bigl( r_{i,1}(s, a_i, a_1) + \gamma \sum_{{s}' \in \cS} {\PP_1(s' \mid s, a_1) } \cdot\hat{V}_{i}^{(k)}({s}') - \hat{Q}_{i,1}^{(k)}(s,a_i, a_1) \Bigr),
$$ 
}}

\neurips{
\vspace{-0.12in}
}

\tnumbering{
\begin{align*}
   \hat{Q}_{i,1}^{(k+1)}&(s, a_i, a_1) = \hat{Q}_{i,1}^{(k)}(s,a_i, a_1)  
   \\
   &+  \pmb{1}(s =s^{(k)}) {\beta^{N(s)}}\Bigl( r_{i,1}(s, a_i, a_1) + \gamma \sum_{{s}' \in \cS} {\PP_1(s' \mid s, a_1) } \cdot\hat{V}_{i}^{(k)}({s}') - \hat{Q}_{i,1}^{(k)}(s,a_i, a_1) \Bigr),   
\end{align*}
}

\neurips{
\begin{align*}
   \hat{Q}_{i,1}^{(k+1)}&(s, a_i, a_1) = \hat{Q}_{i,1}^{(k)}(s,a_i, a_1)  
   \\
   &+  \pmb{1}(s =s^{(k)}) {\beta^{N(s)}}\Bigl( r_{i,1}(s, a_i, a_1) + \gamma \sum_{{s}' \in \cS} {\PP_1(s' \mid s, a_1) } \cdot\hat{V}_{i}^{(k)}({s}') - \hat{Q}_{i,1}^{(k)}(s,a_i, a_1) \Bigr),   
\end{align*}
}

\neurips{
\vspace{-0.12in}
}

where $\hat{V}_{i}^{(k)}(s) = \max_{a_i\in\cA_i} \hat{Q}_i^{(k)}(s, e_{a_i}, \hat{\pi}_{-i}^{(k)}(s))$, 
and we let $ \hat{Q}_i^{(k+1)}(s, \ba)= \hat{Q}_{i,1}^{(k+1)}(s, a_i, a_1)$ for these $i\neq 1$. 
The overall dynamics are summarized in \Cref{alg:markov-polymatrix-game-fictitious}{, which resembles the FP dynamics for two-player zero-sum \cite{sayin2022fictitious} and identical-interest  \cite{sayin2022fictitiousb} MGs. Now we are ready to present the convergence guarantees.}  

\begin{restatable}{assumption}{STEPSIZEMZNMG}
\label{assumption:stepsize2}
The sequences of step sizes $\left\{\alpha^k \in(0,1]\right\}_{k \geq 0}$ and $\left\{\beta^k \in(0,1]\right\}_{k \geq 0}$ satisfy the following conditions: (1) $\sum_{k=0}^{\infty} \alpha^k=\infty$, $\sum_{k=0}^{\infty} \beta^k=\infty$, and $\lim _{k \rightarrow \infty} \alpha^k=\lim _{k \rightarrow \infty} \beta^k=0$; (2) $\lim _{k \rightarrow \infty} \frac{\beta^k}{\alpha^k}=0$, indicating that the rate at which the beliefs about $Q$-functions are updated is slower than the rate at which the beliefs about policies are updated. 
\end{restatable}

\neurips{
\vspace{-0.07in}
}

\arxiv{
\begin{algorithm}[!t]
\caption{Fictitious play in zero-sum NMGs {of a star-shape} ($i$-th player)}
\label{alg:markov-polymatrix-game-fictitious}
\begin{algorithmic}
\STATE{{Choose} $\hat{\pi}_{j}^{(0)}(s)$ to be a uniform distribution for all $j \in \cN/\{i\}$ and  $s \in \cS$}
\STATE{{Choose} $\hat{Q}_{i}^{(0)}(s, \ba)$ to be an arbitrary value for all $s \in \cS$ and $\ba \in \cA$}
\STATE{{Choose} $N(s) = 0$ for all $s \in \cS$}
\FOR{each timestep $k = 0, 1, \dots$}
\STATE {Observe the current state $s^{(k)}$ and update the visitation number as $N(s^{(k)}) = N(s^{(k)})+1$}
\STATE {Take action $a_{i}^{(k)} \in \argmax_{a_i\in\cA_i}\hat{Q}_i^{(k)}(s^{(k)}, e_{a_i}, \hat{\pi}_{-i}^{(k)}(s^{(k)}))$}
\STATE {Update ${V}_{i}$-belief for all $i \in \cN$ and $s \in \cS$ as 
\begin{align}
    \hat{V}_{i}^{(k)}(s) = \max_{a_i\in\cA_i} \hat{Q}_i^{(k)}(s, e_{a_i}, \hat{\pi}_{-i}^{(k)}(s)) \label{eqref:V-def}
\end{align} } 
\STATE {Observe other players'  action $a_{-i}^{(k)}$} 
\STATE {Update the belief as $\hat{\pi}_{-i}^{(k+1)}(s) = \hat{\pi}_{-i}^{(k)}(s)  +  \pmb{1}(s = s^{(k)}) \alpha^{N(s)} (e_{a_{-i}^{(k)}} - \hat{\pi}_{-i}^{(k)}(s)) $ for all $s \in \cS$}
\IF {player $i$ = 1}
\STATE {Update the $Q_{1,j}$-belief for all $j \in \cN/\{1\}$, $s \in \cS$, and $\ba\in\cA$ as 
\numbering{$$
{\small\begin{aligned}
  \hat{Q}_{1,j}^{(k+1)}(s, a_1, a_j) &= \hat{Q}_{1,j}^{(k)}(s,a_1, a_j)  + \pmb{1}(s = s^{(k)}) \beta^{N(s)}\Bigl( r_{1,j}(s, a_1, a_j) 
  + \gamma \sum_{{s}' \in \cS} \frac{1}{n-1} \PP_1(s' \mid s, a_1) \hat{V}_{1}^{(k)}({s}') - \hat{Q}_{1,j}^{(k)}(s,a_1, a_j) \Bigr)   
\end{aligned}}
$$}
\tnumbering{$$
{\begin{aligned}
  \hat{Q}_{1,j}^{(k+1)}(s, a_1, a_j) &= \hat{Q}_{1,j}^{(k)}(s,a_1, a_j)  + \pmb{1}(s = s^{(k)}) \beta^{N(s)}\Bigl( r_{1,j}(s, a_1, a_j) 
  \\
  &\qquad+ \gamma \sum_{{s}' \in \cS} \frac{1}{n-1} \PP_1(s' \mid s, a_1) \hat{V}_{1}^{(k)}({s}') - \hat{Q}_{1,j}^{(k)}(s,a_1, a_j) \Bigr)   
\end{aligned}}
$$}
}
\STATE { Update the $Q_{1}$-belief for all $s \in \cS$ and $\ba\in\cA$ as 
\numbering{$$
\begin{aligned}
  \hat{Q}_{1}^{(k+1)}(s, \ba) &= \sum_{j \in \cN/\{1\} } \hat{Q}_{1,j}^{(k+1)}(s,a_1, a_j)
\end{aligned}
$$
}
\tnumbering{
\begin{align*}
  \hat{Q}_{1}^{(k+1)}(s, \ba) &= \sum_{j \in \cN/\{1\} } \hat{Q}_{1,j}^{(k+1)}(s,a_1, a_j)
\end{align*}
}
}
\ELSE
\STATE {Update the $Q_{i}$-belief for all $s \in \cS$ and $\ba\in\cA$ as $$
\begin{aligned}
  \hat{Q}_i^{(k+1)}(s, \ba) = &\hat{Q}_{i,1}^{(k+1)}(s, a_i, a_1) = \hat{Q}_{i,1}^{(k)}(s,a_i, a_1)  +  \pmb{1}(s = s^{(k)}) \beta^{N(s)}\Bigl( r_{i, 1}(s, a_i, a_1) 
\\
 &\qquad\qquad\qquad\qquad\qquad\qquad \qquad + \gamma \sum_{{s}' \in \cS} \PP_1({s}' \mid s, a_1) \hat{V}_{i}^{(k)}({s}') - \hat{Q}_{i,1}^{(k)}(s,a_i, a_1) \Bigr)   
\end{aligned}
$$
}
\ENDIF
\STATE{State transitions~~{$s^{(k+1)}\sim \PP_1(\cdot\given s^{(k)},a^{(k)}_1)$}}
\ENDFOR
\end{algorithmic}
\end{algorithm}
 }

\begin{restatable}{theorem}{fppMZNMG}
\label{thm:fppMZNMG_main} 
{Suppose  \Cref{assumption:stepsize2} holds  and \Cref{alg:markov-polymatrix-game-fictitious}  visits every state infinitely often with probability $1$.  Then, for a star-shaped multi-player zero-sum NMG, the belief $(\hat{\pi}^{(k)})_{k\geq 0}$ converges to a Markov stationary NE and the belief  $(\hat{Q}^{(k)})_{k\geq 0}$ converges to {the corresponding} NE value of the zero-sum NMG with probability 1, as $k \to \infty$.}
\end{restatable}

\neurips{
\vspace{-0.07in}
}
{We defer the proof to \Cref{sec:FP_inf_case} due to space constraints. Note that to illustrate the idea,  we only present the result for the \emph{model-based} case, 
i.e., when the transition dynamics  $\PP$ is known.
With this result, it is direct to extend to the model-free and learning case, where $\PP$ is not known \cite{sayin2022fictitious,sayin2022fictitiousb,baudin2022fictitious}, still  using the tool of stochastic approximation \cite{benaim2005stochastic}. See \Cref{appendix:fictitious} for more details.  
\begin{remark}[Challenges for analyzing  general cases]\label{rem:challenges_general}
One might ask why we had to focus on a star-shaped structure. First, for general networked  structures, even in the matrix-game case, it is known that the NE \emph{values}  of a zero-sum NG may not be unique \cite{cai2016zero}. Hence, suppose one performs \emph{Nash-value iteration}, i.e., solving for the NE of the stage game and conducting backward induction, this value iteration process does  not converge  in general as the number of backward steps increases, since the solution at each stage is not even unique, and there may not exist a unique fixed point. This is in stark contrast to the $\max$ and $\max\min$ operators in the value iteration updates for single-player and two-player zero-sum cases, respectively. By exploiting a star-shaped structure, we managed to reformulate a {\it minimax}  optimization problem when solving each stage game, which makes the corresponding value iteration operator \emph{contracting},   and thus iterating it infinitely converges to the unique  fixed point. Second, suppose there exists some other network structure (other than star-shaped ones) that also leads to a contracting value iteration operator, then for a fixed constant $\gamma$, the fixed point (which corresponds to the Markov stationary CCE/NE of the zero-sum NMG) becomes unique and can be computed efficiently, which contradicts our hardness result in  \Cref{thm:PPAD-hard-main}. Indeed, it was the exclusion of a star-shaped structure in \Cref{thm:PPAD-hard-main} that inspired us to consider this structure in proving the convergence of FP dynamics. That being said, we note that having a contracting value iteration operator is only a \emph{sufficient} condition for the FP dynamics to converge. It would be interesting to explore other structures that enjoy the FP property for reasons beyond this contraction property. We leave this as an immediate future work.     
\end{remark}
}

\arxiv{
\begin{remark}[Stationary equilibrium computation via {\it value iteration}]
 Following up on \Cref{rem:challenges_general}, we know that with a star-shape topology, one can   formulate 
 a contracting value iteration operator, and develop 
 Nash-value iteration algorithm accordingly, {to 
 find the stationary NE  in this star-shaped case efficiently.} This folklore result supplements the hardness results in \Cref{{thm:PPAD-hard-main}}, where stationary equilibria computation in cases other than the star-shaped ones are computationally intractable. This thus completes the landscape of stationary equilibria computation in zero-sum NMGs. We provide the value-iteration process in \Cref{alg:Value-iteration} and a more detailed discussion in \Cref{sec:FP_inf_case}.   
\end{remark}
}

Next, we present another positive result in light of the hardness in \Cref{thm:PPAD-hard-main}, regarding the computation of \emph{non-stationary} equilibria in multi-player zero-sum NMGs.

\numbering{\section{Non-Stationary NE Computation}}
\tnumbering{\chapter{Non-Stationary NE Computation}}
\label{section:OMWU}

We now focus on computing an (approximate)  Markov \emph{non-stationary} equilibrium in zero-sum NMGs.   
In particular, {we show that when relaxing the stationarity requirement, not only CCE, but NE, can be computed efficiently.}  Before introducing our algorithm, we first recall the folklore result that approximating Markov non-stationary NE in \emph{infinite-horizon} discounted settings can be achieved by finding approximate Markov NE in \emph{finite-horizon} settings, with a large enough horizon length (c.f. \Cref{prop:truncation}). Hence, we will focus on the finite-horizon setting from now on. 

{Before delving into the details of our algorithm, we introduce the notation  $\bQ_{h,i}(s)$ and $\bQ_h(s)$ for $h \in [H], i \in \cN, s \in \cS$ as follows: 
 \begin{align*} 
     &\bQ_{ h, i}(s) := (Q_{ h, i,1}(s), \dots, Q_{h, i,i-1 }(s), \bm{0}, Q_{ h, i,i+1}(s) \dots, Q_{ h, i, n}(s)) \in  \RR^{|\cA_i| \times \sum_{i \in \cN} |\cA_i| }   
     \\     
     &\bQ_h(s) := ((\bQ_{ h,1}(s))^\intercal, (\bQ_{h,2}(s))^\intercal, \dots, (\bQ_{h,n}(s))^\intercal)^\intercal \in \RR^{\sum_{i \in \cN} |\cA_i| \times \sum_{i \in \cN} |\cA_i| }.
\end{align*}
Here, $Q_{h,i,j}$ represents an estimate of the equilibrium value function with canonical decomposition (\Cref{def:canonical}).}  
{Hereafter, we similarly define the notation of $\bQ_{h, i}^\pi$ and $\bQ_h^\pi$.  Our algorithm is based on value iteration, and  iterates three main steps from $h = H$ to 1 as follows: (1) $Q$-value computation: compute $Q_{h,i,j}$, which estimates the equilibrium $Q$-value \arxiv{function }with a canonical decomposition form;  {in particular, when $\cN_C \neq \emptyset$,  $Q_{h,i,j}$ is {updated \arxiv{as follows }for all $s\in\cS, (i, j) \in \cE_Q, a_i\in\cA_i$, and $a_j\in\cA_j$}:}
\numbering{{\neurips{\fontsize{7.5}{6}\selectfont}
\begin{align*}
    &Q_{h,i,j}(s, a_i, a_j) = r_{h,i,j}(s, a_i, a_j) + \sum_{s' \in \cS} \biggl(\frac{1}{|\cE_{Q, i}|} \pmb{1}(i \in \cN_C) \FF_{h, i} (s' \given s, a_i)  +  \pmb{1}(j \in \cN_C)\FF_{h, j}(s' \given s,a_j)\biggr)  V_{h+1, i}(s'),       
\end{align*}}
\neurips{
\vspace{-2pt}
}
}
\tnumbering{
\begin{align*}
    Q_{h,i,j}(s, a_i, a_j) = r_{h,i,j}(s, a_i, a_j) + \sum_{s' \in \cS}& \biggl(\frac{1}{|\cE_{Q, i}|} \pmb{1}(i \in \cN_C) \FF_{h, i} (s' \given s, a_i)  
    \\
    &+  \pmb{1}(j \in \cN_C)\FF_{h, j}(s' \given s,a_j)\biggr)  V_{h+1, i}(s'),       
\end{align*}
}
(2) Policy update: update $\pi_{h}(s)$ with an \textsf{NE-ORACLE}: finding (approximate)-NE of some zero-sum NG $(\cG, \cA, (Q_{h,i,j}(s))_{(i, j) \in \cE_{Q}})$ for all $s \in \cS$, and (3) Value function update: compute $V_{h, i}$, which estimates the equilibrium value function as follows for all $s\in\cS, i \in \cN$: $V_{h, i}(s) = \pi_{h, i}^{\intercal}(s) \bQ_{h,i}(s) \pi_{h}(s)$. \arxiv{When $\cN_C = \emptyset$, we can  similarly calculate $Q_{h,i}$ in the first step.} {The overall procedure is summarized in \Cref{alg:meta-algorithm}.} }
\arxiv{
\begin{algorithm}
\caption{A value-iteration-based  algorithm for finding NE in zero-sum NMGs}
\label{alg:meta-algorithm}
\begin{algorithmic}
\STATE{Update $V_{H+1,i}(s) = 0$ for all $s\in\cS$ and $i\in\cN$}
\FOR{step $h = H, H-1, \dots, 1$}
\IF{$\cN_C \neq \emptyset$} 
\STATE{Update $Q_{h,i, j}(s, a_i, a_j)$ for all $(i, j) \in \cE_Q, s \in \cS, a_i \in \cA_i, a_j \in \cA_j$  as 
\begin{equation}
\numbering{
\begin{aligned}
    &Q_{h,i,j}(s, a_i, a_j)= r_{h,i,j}(s, a_i, a_j) + \sum_{s' \in \cS} \biggl(\frac{1}{|\cE_{Q, i}|} \pmb{1}(i \in \cN_C) \FF_{h, i} (s' \given  s, a_i) 
     +  \pmb{1}(j \in \cN_C)\FF_{h, j}(s'\given s,a_j)\biggr) \cdot V_{h+1, i}(s')        
\end{aligned}
}
\tnumbering{
\begin{aligned}
    Q_{h,i,j}(s, a_i, a_j)&= r_{h,i,j}(s, a_i, a_j) + \sum_{s' \in \cS} \biggl(\frac{1}{|\cE_{Q, i}|} \pmb{1}(i \in \cN_C) \FF_{h, i} (s' \given  s, a_i)
    \\
    &\qquad  +  \pmb{1}(j \in \cN_C)\FF_{h, j}(s'\given s,a_j)\biggr) \cdot V_{h+1, i}(s')        
\end{aligned}
}
 \label{eqn:Qupdates}
\end{equation}}
\ELSIF{$\cN_C = \emptyset$} 
\STATE{ Update $Q_{h,i, j}(s, a_i, a_j)$ for all $(i, j) \in \cE_Q, s \in \cS, a_i \in \cA_i, a_j \in \cA_j$ as
\begin{equation}
    \numbering{\begin{aligned}
    &Q_{h,i,j}(s, a_i, a_j)= r_{h,i,j}(s, a_i, a_j) + \sum_{s' \in \cS} \left(\frac{1}{|\cE_{Q, i}|} \pmb{1}(j \in \cE_{Q, i}) \FF_{h, o} (s'\given s)\cdot V_{h+1, i}(s')\right)    
    \end{aligned}}
    \tnumbering{\begin{aligned}
    Q_{h,i,j}(s, a_i, a_j)&= r_{h,i,j}(s, a_i, a_j) 
    \\
    &\qquad+ \sum_{s' \in \cS} \left(\frac{1}{|\cE_{Q, i}|} \pmb{1}(j \in \cE_{Q, i}) \FF_{h, o} (s'\given s)\cdot V_{h+1, i}(s')\right)    
    \end{aligned}}
    \label{eqn:Qupdates2}
\end{equation}
}
\ENDIF
\STATE{Update  $\pi_{h}(s) = \textsf{NE-ORACLE}(\cG, \cA, (Q_{h,i,j}(s))_{(i, j) \in \cE_{Q}})$ for all $ s\in \cS $}
\STATE{Update  $V_{h,i}(s)$ for all $i \in \cN, s \in \cS$ as 
{\begin{equation}
\begin{aligned}
    &V_{h, i}(s)=  \sum_{j \in \cE_{Q, i}}  {\pi}_{h, i}^{\intercal}(s) Q_{h,i,j}(s) {\pi}_{h,j}(s)
\end{aligned}
\label{eqn:Vupdates}
\end{equation}}
}
\ENDFOR
\end{algorithmic} 
\end{algorithm} 
}
\arxiv{\numbering{
\begin{table}
\centering
\begin{tabular}{c|c|c} 
\hline
              & Regularization & Regularization-free  \\ 
\hline
Optimism    &\Centerstack{OMWU  \cite{ao2022asynchronous}: \\ $\tilde{\cO} (1/{\epsilon})$ last-iterate} & \Centerstack{OMD  \cite{anagnostides2022last}:\\ $\tilde{\cO}({1}/{\epsilon^2})$ best-iterate \\ Asymptotic last-iterate\\\\
\cite{daskalakis2021near,anagnostides2022uncoupled,farina2022near,anagnostides2022last}: \\
$\tilde{\cO} ({1}/{\epsilon})$ average-iterate + Marginalization}                 \\ 
\hline
Optimism-free & \Centerstack{ \Cref{alg:polymatrix-game-QRE}:\\ $\tilde{\cO}(1 / \epsilon^4 )$ last-iterate  \\\\ \Cref{alg:polymatrix-game-decreasing-regular}: \\$\tilde{\cO}(1/\epsilon^6)$ last-iterate  }      &    \Centerstack{Any no-regret learning algorithm with\\   {$\tilde{\cO}(1/\epsilon^2)$} average-iterate + Marginalization} 
\\
\hline
\end{tabular}
\caption{Iteration complexities for finding an $\epsilon$-NE for a zero-sum NG with $(\cG= (\cN, \cE), \cA, (r_{i,j})_{(i, j) \in \cE})$ with different \textsf{NE-ORACLE} subroutines. $\tilde{\cO}(\cdot)$ omits  polylog terms and polynomial dependencies on $n, \|\br\|_{\max}, R$.}
\label{table:MZNG-epsilon-NE}
\end{table}
}}
\tnumbering{
\begin{table}
\centering
\renewcommand{\arraystretch}{1.5} %
\newcolumntype{C}[1]{>{\centering\arraybackslash}m{#1}} %
\begin{tabular}{c|C{3.5cm}|C{5.5cm}} 
\hline
              & \textbf{Regularization} & \textbf{Regularization-free}  \\ 
\hline
\textbf{Optimism}    & OMWU \cite{ao2022asynchronous}: $\tilde{\cO}(1/{\epsilon})$ last-iterate & OMD \cite{anagnostides2022last}: $\tilde{\cO}(1/{\epsilon^2})$ best-iterate \\ 
& & Asymptotic last-iterate \cite{daskalakis2021near,anagnostides2022uncoupled,farina2022near,anagnostides2022last}: $\tilde{\cO}(1/{\epsilon})$ average-iterate + Marginalization \\ 
\hline
\textbf{Optimism-free} & \Cref{alg:polymatrix-game-QRE}: $\tilde{\cO}(1/{\epsilon^4})$ last-iterate  & \Cref{alg:polymatrix-game-decreasing-regular}: $\tilde{\cO}(1/{\epsilon^6})$ last-iterate \\
& & Any no-regret learning algorithm with $\tilde{\cO}(1/{\epsilon^2})$ average-iterate + Marginalization \\
\hline
\end{tabular}
\caption{Iteration complexities for finding an $\epsilon$-NE for a zero-sum NG with $(\cG = (\cN, \cE), \cA, (r_{i,j})_{(i, j) \in \cE})$ with different \textsf{NE-ORACLE} subroutines. $\tilde{\cO}(\cdot)$ omits polylog terms and polynomial dependencies on $n, \|\br\|_{\max}, R$.}
\label{table:MZNG-epsilon-NE}
\end{table}
}

\paragraph{\textsf{NE-ORACLE} and iteration complexity.} {The \textsf{NE-ORACLE} in \Cref{alg:meta-algorithm} can be instantiated by several different algorithms that can find an NE in a zero-sum NG. Depending on the algorithms, the convergence guarantees can be either in terms of average-iterate, best-iterate, or last-iterate. {Note that for algorithms with average-iterate convergence,} one may additionally need to {\it marginalize} the output joint policy, i.e., the approximate CCE, and combine them as a {\it product} policy that is an approximate NE (\Cref{prop:polymatrix}). For those with best-/last-iterate convergence, by contrast, the best-/last-iterate is already in product form, and one can directly output it as an approximate NE. Moreover, last-iterate convergence is known to be a more favorable metric than the average-iterate one in learning in games \cite{mertikopoulos2018cycles,daskalakis2018training,daskalakis2018last,bailey2018multiplicative,mertikopoulos2018optimistic}, which is able to characterize the {\it day-to-day} behavior of the iterates and implies the stability of the update rule.  
Hence, one may prefer to have last-iterate convergence for solving zero-sum N(M)Gs. To this end, two algorithmic ideas may be useful: adding regularization to the payoff matrix \cite{cen2021fast,cen2022faster,cen2022independent,ao2022asynchronous,pattathil2022symmetric}, and/or using the idea of optimism \cite{daskalakis2018last,wei2021last,weilinear}. {Recent results \cite{anagnostides2022last,ao2022asynchronous} have instantiated the ideas of {\it optimism-only} and {\it optimism + regularization}, respectively, for best-/last-iterate convergence in zero-sum polymatrix games.} We additionally established results for the idea of {\it regularization-only} in obtaining last-iterate convergence in these games. Specifically, we propose to study the vanilla Multiplicative Weight Update (MWU) algorithm \cite{arora2012multiplicative} in the regularized zero-sum NG, as tabulated in \Cref{alg:polymatrix-game-QRE}. We have also introduced a variant with diminishing regularization, and summarize the update rule in \Cref{alg:polymatrix-game-decreasing-regular}. 

\arxiv{
{In terms of iteration complexity, if neither optimism nor regularization is used, then one can resort to any no-regret learning algorithm, together with a marginalization step, to obtain the standard result of $\tilde\cO(1/\epsilon^2)$  \cite{cesa2006prediction}, thanks to the equilibrium-collapse result in \cite{cai2016zero} (and our generalized version that accommodates approximation in \Cref{prop:polymatrix}). 
If both regularization and optimism are used, then the optimistic MWU (OMWU) algorithm in \cite{ao2022asynchronous} leads to the fast rate of $\tilde\cO(1/\epsilon)$, in terms of last-iterate convergence. If only optimism is used, the  optimistic mirror descent (OMD) algorithm in \cite{anagnostides2022last} gives $\tilde\cO(1/\epsilon^2)$ best-iterate and asymptotic last-iterate for finding NE in zero-sum polymatrix games. Moreover, optimistic algorithms, e.g., those in \cite{daskalakis2021near,anagnostides2022uncoupled,farina2022near,anagnostides2022last}, can also achieve a fast rate of $\tilde\cO(1/\epsilon)$ in terms of average-iterate convergence (due to $\tilde\cO(1)$ regret guarantees). Finally, if only regularization is used, our algorithms (\Cref{alg:polymatrix-game-QRE,alg:polymatrix-game-decreasing-regular}) can achieve  $\tilde\cO(1/\epsilon^4)$ and $\tilde\cO(1/\epsilon^6)$ last-iterate convergence to NE, respectively. Note that for best-iterate and last-iterate convergence, no marginalization is needed when outputting the approximate equilibrium policies. We summarize the results in \Cref{table:MZNG-epsilon-NE}.}}

{Given the results above, aggregating $\epsilon$-approximate NE for the zero-sum NGs    $(\cG, \cA, (Q_{h,i,j}(s))_{(i, j) \in \cE_{Q}})$ for all $h \in [H], i \in \cN, s \in \cS$ provides an $H\epsilon$-approximate NE for the corresponding zero-sum NMG. We have the following formal result. 

\begin{restatable}{proposition}{NEORACLETONE}
\label{prop:NEoracle-ne}
    Suppose that for all $h \in [H], i \in \cN, s \in \cS$, $\textsf{NE-ORACLE}(\cG, \cA, (Q_{h,i,j}(s))_{(i, j) \in \cE_{Q}})$ provides an $\epsilon_{h, s}$-approximate  NE for the zero-sum NG $(\cG, \cA, (Q_{h,i,j}(s))_{(i, j) \in \cE_{Q}})$ in \Cref{alg:meta-algorithm}. Then, the output policy $\pi$ in \Cref{alg:meta-algorithm} is an $(\sum_{h \in [H]}$ $\max_{s \in\cS}\epsilon_{h, s})$-approximate NE for the corresponding zero-sum NMG $(\cG = (\cN, \cE_Q), \cS, \cA, H, (\PP_h)_{h \in [H]},$ $(r_{h,i,j}(s))_{(i, j) \in \cE_Q, s \in \cS})$.
\end{restatable}}

The proof of \Cref{prop:NEoracle-ne} is deferred to \Cref{appendix:OMWU}. In light of \Cref{prop:NEoracle-ne} and \Cref{table:MZNG-epsilon-NE}, we obtain \Cref{table:MZNMG-epsilon-NE}, which summarizes the iteration complexities required to find an $\epsilon$-NE for zero-sum NMGs, with different \textsf{NE-ORACLE} subroutines. {Note that the iteration complexities are all polynomial in $H,n,|\cS|$, and inherit the order of dependencies on $\epsilon$ from \Cref{table:MZNG-epsilon-NE} for the matrix-game case. In particular,  \Cref{alg:meta-algorithm} \arxiv{in conjunction} with the OMWU in \cite{ao2022asynchronous} yields the fast rate of $\tilde{\mathcal{O}}(1/\epsilon)$ for the last iterate.}  

\arxiv{\begin{theorem}
    \Cref{alg:meta-algorithm} with the \textsf{NE-ORACLE} subroutine being  \Cref{alg:polymatrix-game-QRE} or \Cref{alg:polymatrix-game-decreasing-regular} requires no more than $\tilde{\cO}({H^9n |\cS|}/{\epsilon^4})$ or $\tilde{\cO}(H^{19}  n^{3}|\cS|/\epsilon^6)$ iterations 
    to achieve an $\epsilon$-NE at the last iterate, respectively. 
\end{theorem}}

\arxiv{ \numbering{
\begin{table}
\centering

\begin{tabular}{c|c|c} 
\hline  
& Regularization & Regularization-free   
\\
\hline
Optimism  & \Centerstack{\Cref{alg:meta-algorithm} + OMWU \cite{ao2022asynchronous}: \\ $\tilde{\cO}({1}/{\epsilon})$ last-iterate} & \Centerstack{\Cref{alg:meta-algorithm} + OMD \cite{anagnostides2022last}: \\ $\tilde{\cO}({1}/{\epsilon^2})$ best-iterate\\\\
\Cref{alg:meta-algorithm} + \cite{daskalakis2021near,anagnostides2022uncoupled,farina2022near,anagnostides2022last}:\\
$\tilde{\cO}({1}/{\epsilon})$ average-iterate + Marginalization}                 \\ 
\hline
Optimism-free & \Centerstack{ \Cref{alg:meta-algorithm} + \Cref{alg:polymatrix-game-QRE}: \\
$\tilde{\cO}({1}/{\epsilon^4})$ last-iterate \\\\ \Cref{alg:meta-algorithm} + \Cref{alg:polymatrix-game-decreasing-regular}: \\
$\tilde{\cO}({1}/{\epsilon^6})$ last-iterate}      &    \Centerstack{\Cref{alg:meta-algorithm} + Any no-regret 
 learning algorithm \\with {$\tilde{\cO}(1/\epsilon^2)$} average-iterate +  Marginalization}
\\
\hline
\end{tabular}
\caption{Iteration complexities for finding an $\epsilon$-NE for a zero-sum NMG with different \textsf{NE-ORACLE} subroutines. $\tilde{\cO}(\cdot)$ omits  polylog terms and polynomial dependencies on $n,H,|\cS|, R, \norm{\br}_{\max}$.} 
\label{table:MZNMG-epsilon-NE} 
\end{table}
}}

\numbering{\section{Experimental Results}}
\tnumbering{\chapter{Experimental Results}}
\label{appendix:experiment}
We now present experimental results for the learning dynamics/algorithms investigated before.

\numbering{\subsection{Fictitious-play property of zero-sum NMGs}}
\tnumbering{\section{Fictitious-play property of zero-sum NMGs}}
We present an experiment for the fictitious-play property in \Cref{sec:fictitious-play}. We experimented  with an infinite-horizon $\gamma$-discounted  zero-sum NMG $(\cG = (\cN, \cE_Q), \cS, \cA, \PP, (r_{i})_{i \in \cN}, \gamma)$, where $\cN = \{0,1,2\}$, $\cS = \{0, 1\}$, $\cA_i = \{0, 1\}$ for all $i \in \cN$, $\PP(s' \given s, \ba) := \FF_{0}(s'\given s, a_0)$, $\FF_{0}(0\given 0, 1) = \FF_0(0\given 1, 0) = 0.8$, $\FF_0(0\given 1,1) = \FF_0(0\given 0,0) = 0.2$, $\br_{0,1}(0) = -\br_{1,0}(0)^\intercal =   \begin{bmatrix}
1 & 2 \\
4 & 3 
\end{bmatrix} $, $\br_{0,2}(0) = -\br_{2,0}(0)^\intercal =   \begin{bmatrix}
4 & 3 \\
2 & 1 
\end{bmatrix} $, $\br_{0,1}(1) = -\br_{1,0}(1)^\intercal =   \begin{bmatrix}
4 & 3 \\
2 & 1 
\end{bmatrix} $, $\br_{0,2}(1) = -\br_{2,0}(1)^\intercal =   \begin{bmatrix}
1 & 2 \\
4 & 3 
\end{bmatrix} $, $\alpha_t= \frac{1}{t^{0.55}}$, $\beta_t = \frac{1}{t^{0.75}}$, and $\gamma = 0.99$. We iterated $2^{28}$ times for the experiments.  
The result is demonstrated in  \Cref{fig:experiment} (a). Note that the gray and black lines indicate the sum of the values for states 0 and 1, which asymptotically go to 0. 
\arxiv{\numbering{\begin{figure}[h!] 
\centering
\begin{subfigure}{.5\textwidth}
  \centering
  \includegraphics[scale = 0.45]{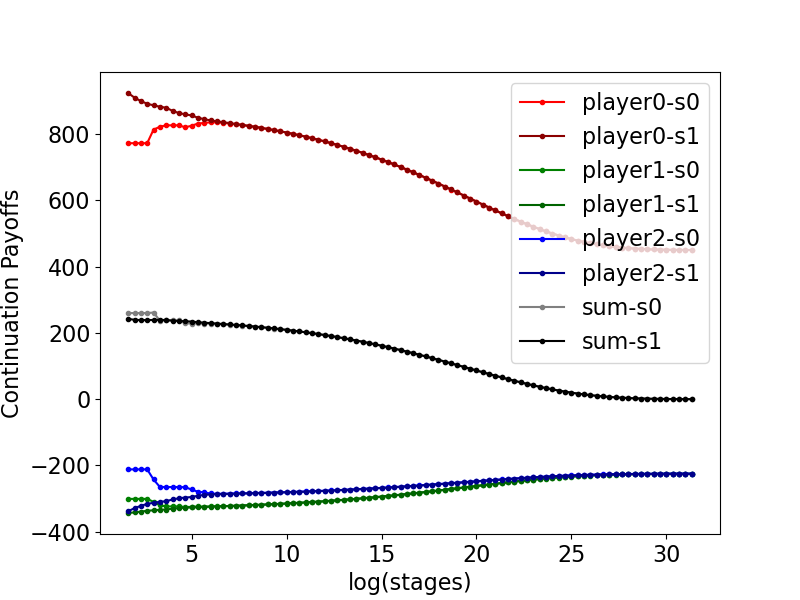}
  \caption{Value function estimates plot}
\end{subfigure}%
\begin{subfigure}{.5\textwidth}
  \centering
  \includegraphics[scale = 0.45]{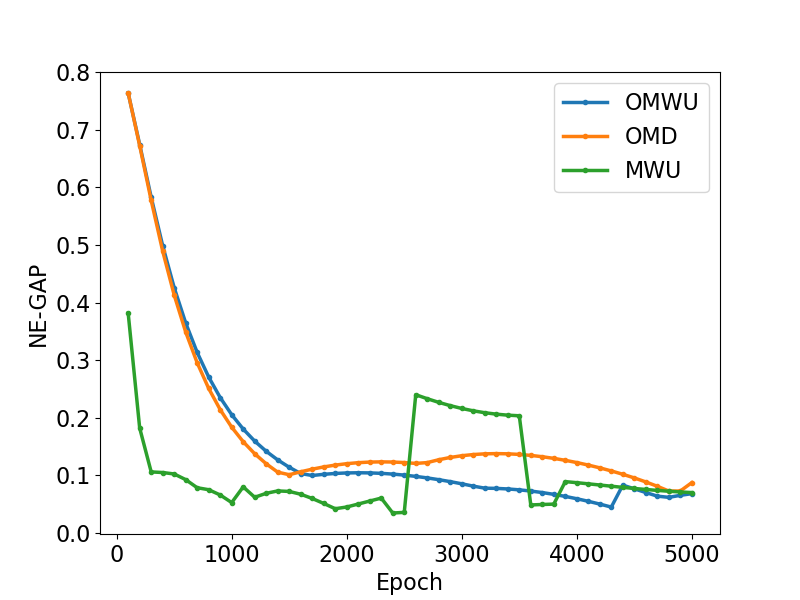}
  \caption{$\max_{i\in\cN, s \in \cS}\left(\max_{\pi_i'\in\Delta(\cA_i)} V_{1,i}^{\pi_i', \pi_{-i}}(s) - V_{1,i}^{\pi}(s)\right)$ plot}
\end{subfigure}%
\caption{(a) Fictitious play experiment. The red and dark red lines indicate player 0's value function estimates for states 0 and 1, respectively. The green and dark green lines indicate player 1's value function estimates for states 0 and 1, respectively. The blue and dark blue lines indicate player 2's value function estimates for states 0 and 1, respectively. The gray and black lines indicate the sum of each player's value function estimates for states 0 and 1, respectively. $x$-axis denotes the logarithm with base 2 of the number of iterates (stages) and $y$-axis denotes the value function estimates. (b) Value-iteration-based algorithms with (OMWU, OMD, MWU) \textsf{NE-ORACLE} subroutines. The blue, orange, and green lines indicate $\max_{i\in\cN, s \in \cS}\left(\max_{\pi_i'\in\Delta(\cA_i)} V_{1,i}^{\pi_i', \pi_{-i}}(s) - V_{1,i}^{\pi}(s)\right)$ value of the OMWU, OMD, and MWU \textsf{NE-ORACLE}, respectively. $x$-axis denotes  the number of iteration of \textsf{NE-ORACLE} subroutine, and $y$-axis denotes  the \textsf{NE-Gap}.}
\label{fig:experiment}
\end{figure}
}}
{\tnumbering{\begin{figure}[h!] 
\centering
\begin{subfigure}{.5\textwidth}
  \centering
  \includegraphics[scale = 0.35]{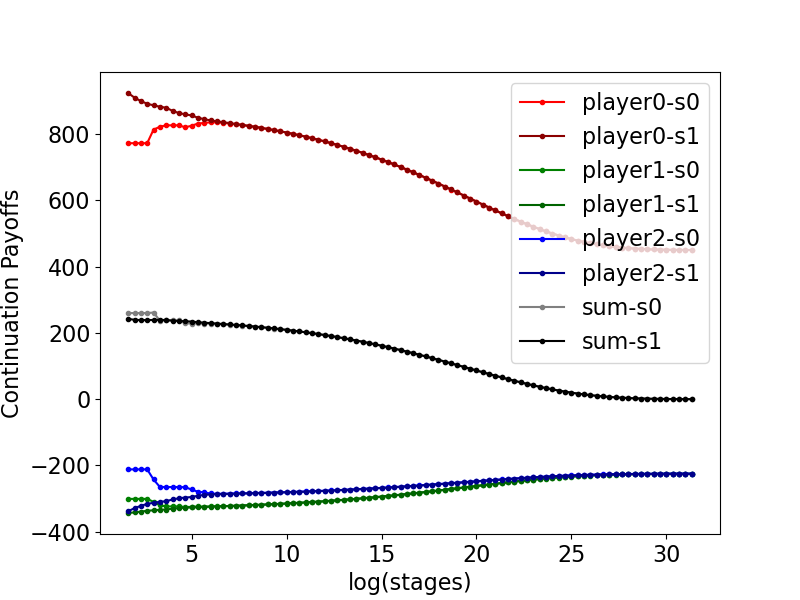}
  \caption{Value function estimates plot}
\end{subfigure}%
\begin{subfigure}{.5\textwidth}
  \centering
  \includegraphics[scale = 0.35]{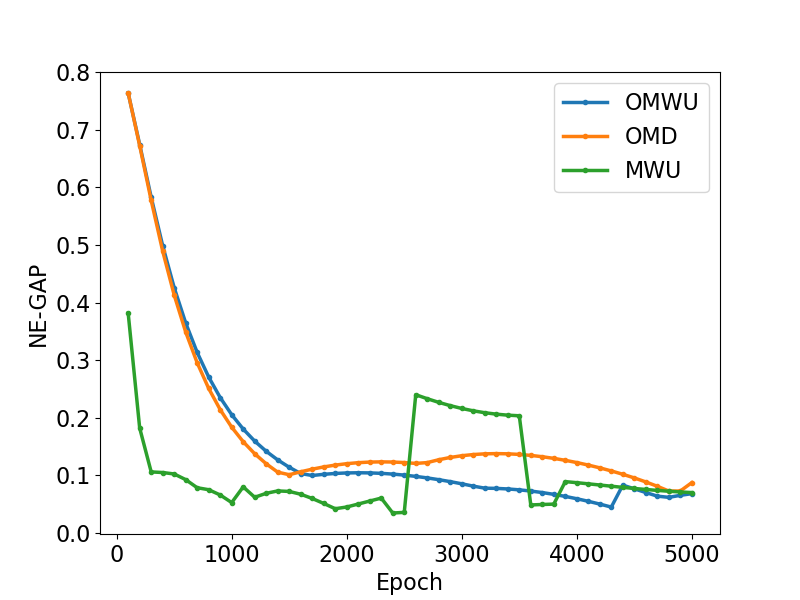}
  \caption{$\max_{i\in\cN, s \in \cS}\left(\max_{\pi_i'\in\Delta(\cA_i)} V_{1,i}^{\pi_i', \pi_{-i}}(s) - V_{1,i}^{\pi}(s)\right)$}
\end{subfigure}%
\caption{(a) Fictitious play experiment. The red and dark red lines indicate player 0's value function estimates for states 0 and 1, respectively. The green and dark green lines indicate player 1's value function estimates for states 0 and 1, respectively. The blue and dark blue lines indicate player 2's value function estimates for states 0 and 1, respectively. The gray and black lines indicate the sum of each player's value function estimates for states 0 and 1, respectively. $x$-axis denotes the logarithm with base 2 of the number of iterates (stages) and $y$-axis denotes the value function estimates. (b) Value-iteration-based algorithms with (OMWU, OMD, MWU) \textsf{NE-ORACLE} subroutines. The blue, orange, and green lines indicate $\max_{i\in\cN, s \in \cS}\left(\max_{\pi_i'\in\Delta(\cA_i)} V_{1,i}^{\pi_i', \pi_{-i}}(s) - V_{1,i}^{\pi}(s)\right)$ value of the OMWU, OMD, and MWU \textsf{NE-ORACLE}, respectively. $x$-axis denotes  the number of iteration of \textsf{NE-ORACLE} subroutine, and $y$-axis denotes  the \textsf{NE-Gap}.}
\label{fig:experiment}
\end{figure}
}}
\neurips{\begin{figure}[h!] 
\centering
\begin{subfigure}{.5\textwidth}
  \centering
  \includegraphics[scale = 0.35]{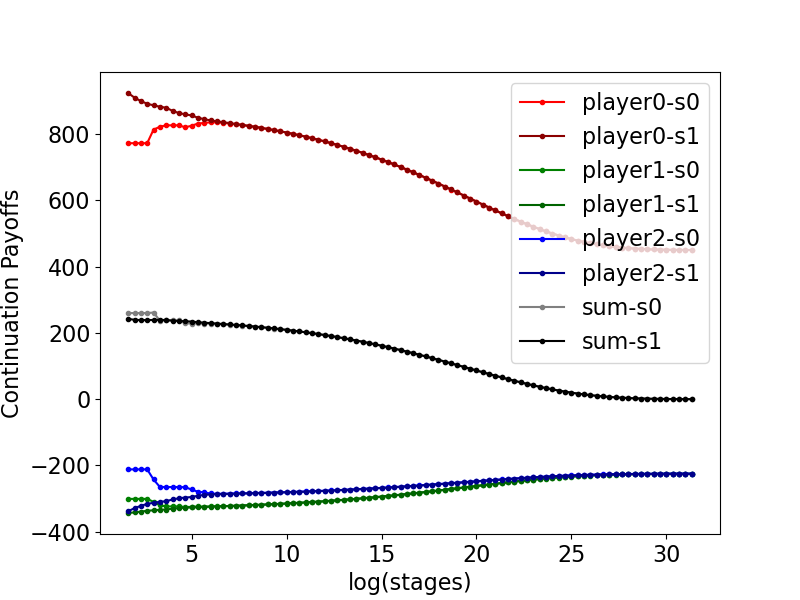}
  \caption{Value function estimates plot}
\end{subfigure}%
\begin{subfigure}{.5\textwidth}
  \centering
  \includegraphics[scale = 0.35]{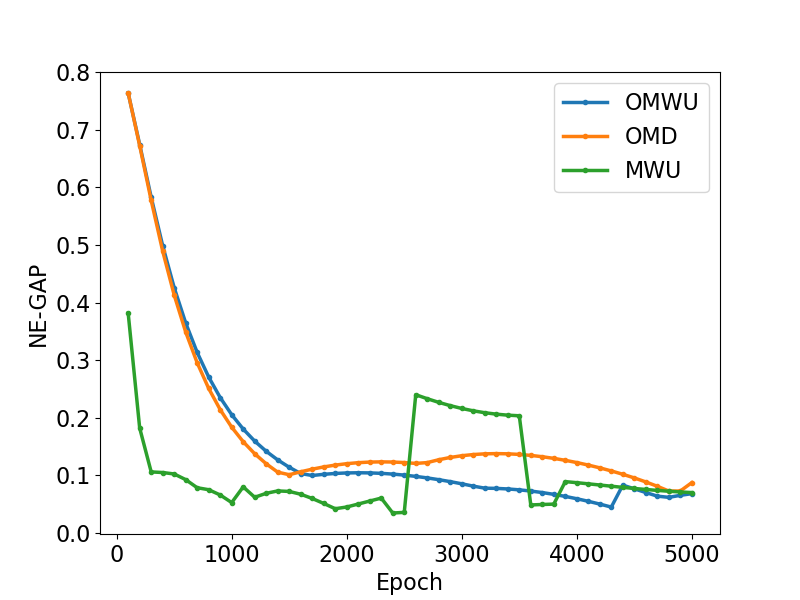}
  \caption{$\max_{i\in\cN, s \in \cS}\left( V_{1,i}^{\dagger, \pi_{-i}}(s) - V_{1,i}^{\pi}(s)\right)$ plot}
\end{subfigure}%
\caption{(a) Fictitious play experiment. The red and dark red lines indicate player 0's value function estimates for states 0 and 1, respectively. The green and dark green lines indicate player 1's value function estimates for states 0 and 1, respectively. The blue and dark blue lines indicate player 2's value function estimates for states 0 and 1, respectively. The gray and black lines indicate the sum of each player's value function estimates for states 0 and 1, respectively. $x$-axis denotes the logarithm with base 2 of the number of iterates (stages) and $y$-axis denotes the value function estimates. (b) Value-iteration-based algorithms with (OMWU, OMD, MWU) \textsf{NE-ORACLE} subroutines. The blue, orange, and green lines indicate $\max_{i\in\cN, s \in \cS}\left(\max_{\pi_i'\in\Delta(\cA_i)} V_{1,i}^{\pi_i', \pi_{-i}}(s) - V_{1,i}^{\pi}(s)\right)$ value of the OMWU, OMD, and MWU \textsf{NE-ORACLE}, respectively. $x$-axis denotes  the number of iteration of \textsf{NE-ORACLE} subroutine, and $y$-axis denotes  the \textsf{NE-Gap}.}
\label{fig:experiment}
\end{figure}
}

\numbering{\subsection{Value-iteration with  different \textsf{NE-ORACLE}s}}
\tnumbering{\section{Value-iteration with  different \textsf{NE-ORACLE}s}}
We present an experiment for zero-sum NMGs with different \textsf{NE-ORACLE}s in \Cref{section:OMWU}. We experimented with a zero-sum NMG $(\cG = (\cN, \cE_Q), \cS, \cA, H, (\PP_h)_{h \in [H]}, (r_{h, i})_{h \in [H], i \in \cN})$ where $\cN = \{0,1,2\}$, $\cE_Q = \{(1,2), (1,0), (2,0)\}$, $\cS$ = \{0,1\}, $\cA_i = \{0, 1\}$ for all $i \in \cN$, $H = 5$, $\PP_h(s' \given  s, \ba) := \sum_{i \in \cN}\FF_{i}(s'\given s, a_i)$, $\FF_0(s'\given s, a_0) = 0$, $\FF_{1}(s'\given s, a_1) = \frac{1}{3} \PP_{1} (s' \given s, a_1)$, $\FF_{2}(s'\given s, a_2) = \frac{2}{3} \PP_2(s'\given s, a_2)$, $\PP_1(0\given 0,1) = \PP_1(0\given 1, 0) = \PP_2(0\given 0,1) = \PP_2(0 \given 1, 0)$ is determined randomly, $\PP_1(0\given 1,1) = \PP_1(0\given 0, 0) = \PP_2(0\given 1,1) = \PP_2(0 \given 0, 0)$ is determined randomly, and $r_{h,i}$ is determined randomly such that it makes $(\cG, \cA, (r_{h,i}(s))_{i \in \cN})$ a zero-sum NG for every $h$. We set $\tau = 0.05$ for both OMWU and MWU. We set $\eta = 1/(36H)$ for both OMWU and OMD. We iterated the algorithm for $T = 5000$ times. The result is plotted in \Cref{fig:experiment} (b).

\arxiv{
\numbering{\section{Concluding Remarks}}
\tnumbering{\chapter{Concluding Remarks}}

We studied a new class  of  non-cooperative Markov games, i.e., multi-player zero-sum  Markov games with networked separable interactions.  We established the structural properties of reward and transition dynamics under this model,  and showed that marginalizing Markov CCE per state leads to Markov NE. Furthermore, we established the computational hardness of finding Markov stationary CCE in infinite-horizon discounted zero-sum NMGs, unless the underlying network has a star topology, which is in contrast to the tractability of CCE computation for the normal-form case of zero-sum NMGs{, i.e., zero-sum polymatrix games}. In light of this hardness result, we focused on: 1) star-shaped zero-sum NMGs, and then developed fictitious-play learning dynamics that provably converge to Markov stationary NE; 2) non-stationary Markov NE computation for general zero-sum NMGs,  with finite-iteration last-iterate convergence guarantees. 

Our work has opened up many venues  for future research.  
Firstly, further study of the fictitious-play property for general zero-sum NMGs beyond star-shaped cases would be interesting, as {\tt PPAD}-hardness does not necessarily  imply the impossibility of \emph{asymptotic} convergence of FP dynamics  in these cases. 
Additionally, it would be interesting to study model-free online and/or offline RL in zero-sum NMGs, with sample-complexity/regret guarantees, as well as to explore the networked structure beyond the zero-sum setting in non-cooperative Markov games. 
 }

\bibliography{cparkbib}
\bibliographystyle{unsrt}
\clearpage

\onecolumn

~\\
\centerline{{\fontsize{14}{14}\selectfont \textbf{Supplementary Materials}}}

\vspace{30pt}
In \Cref{appendix:literature-review}, we provide a detailed literature review. In \Cref{appendix:MPGdef}, we provide deferred proofs for the results on the MZNMG formulation in \Cref{sec:MPGdef}. In \Cref{appendix:PPADhard-proof}, we provide deferred proofs for the {\tt PPAD}-hardness of computing Markov stationary CCE in MZNMGs, in \Cref{subsection:PPAD}. In \Cref{appendix:fictitious}, we provide deferred proofs for the fictitious-play property results, in  \Cref{sec:fictitious-play}. In \Cref{appendix:stochastic-approx}, we provide a brief background on stochastic approximation. In \Cref{appendix:OMWU},  
we provide deferred proofs for the results regarding Markov non-stationary NE computation in \Cref{section:OMWU}. 

\vspace{10pt}

\appendix
\numbering{\section{Related Work}}
\tnumbering{\chapter{Related Work}}
\label{appendix:literature-review}

\paragraph{Tabular Markov game.}  
Markov games (MG), which are also referred to as stochastic games,  were initially introduced by \cite{shapley1953stochastic} and have since garnered significant attention within the multi-agent RL literature \cite{busoniu2008comprehensive,zhang2021multi}. Early research, such as \cite{littman1994markov, littman2001friend, littman2001value,hu2003nash},   established asymptotic  convergence of various Q-learning-based dynamics in solving MGs. In contrast, recent studies have mainly focused  on developing  more sample-efficient methods for learning equilibria in two-player zero-sum Markov games, as demonstrated by \cite{bai2020provable, sidford2020solving, xie2020learning, bai2020near, liu2021sharp, zhao2021provably,zhang2020model,li2022minimax}.

Substantial work has also been conducted on learning correlated equilibrium   and coarse correlated equilibrium  in Markov games, including model-based  \cite{liu2021sharp,subramanian2023robustness} and model-free approaches  \cite{song2021can, jin2021v, mao2022improving, mao2022provably}. 
A recently developed algorithm by \cite{daskalakis2022complexity} is able  to learn Markov non-stationary CCE while overcoming the curse of multi-agents, whose sample  complexity has recently been improved in \cite{cui2023breaking,wang2023breaking}. Other studies within the full-information feedback setting have focused on proving convergence to CE/CCE and sublinear individual regret \cite{erez2022regret}. 

\paragraph{Complexity of equilibrium computation.}  
{Computational challenges can occur for Nash equilibrium-finding in even matrix/normal-form games in general. 
Computing such equilibria has been proven  to be {\tt PPAD}-complete 
even for three/two-player general-sum normal-form games \cite{daskalakis2009complexity,chen2009settling}, which is believed to be computationally hard  \cite{papadimitriou1994complexity,rubinstein2017settling}. Nevertheless, linear programming enables the computation of Nash equilibria in {\it two-player zero-sum} games   and {\it zero-sum polymatrix} games \cite{cai2016zero}. Alternative solution concepts including  (coarse) correlated equilibria  are also more favorable than NE when it comes to  computational complexity, as they can also be efficiently computed  \cite{papadimitriou2008computing,cesa2006prediction}. More recently, \cite{daskalakis2022complexity,jin2022complexity} have shown that for infinite-horizon discounted Markov games, computing even the  coarse correlated equilibrium that is Markov stationary can be {\tt PPAD}-hard, which is in stark contrast to the stateless normal-form game case. For a recent  overview of the computational complexity for equilibrium computation,  we refer to \cite{daskalakis2022non}.}

\paragraph{Games with network structure.}
Network Games  \cite{jackson2015games} and Graphical Games \cite{kearns2013graphical} have been extensively studied in the literature {to model the networked interactions among agents}. \cite{kearns2013graphical} introduced treeNash, an algorithm for computing NE in tree-structured graphical games. The algorithm by \cite{kakade2003correlated} can find  correlated equilibrium  in graphical games. Polymatrix games, wherein edges represent two-player games, constitute a particularly intriguing type of network games.  \cite{bergman1998separable} introduced the concept of {\it separable} zero-sum games, where a player's payoff is the sum of their payoffs from pairwise  interactions with other players, and provided equilibrium-finding algorithms. \cite{daskalakis2009network} demonstrated that graphical games with edges representing zero-sum games (also called \emph{pairwise zero-sum polymatrix} games) can be reduced to two-person zero-sum games, streamlining the NE computation for this case. \cite{cai2011minmax} established that separable zero-sum multiplayer games can be transformed into pairwise constant-sum polymatrix games. \cite{cai2016zero} revealed properties of NE in separable zero-sum games, such as non-unique NE payoffs and the reduction of NE computation to CCE computation by  marginalizing the equilibria. 

More recently, researchers have proposed several NE-finding methods that do not depend on linear programming (LP). \cite{leonardos2021exploration} employed a continuous-time version of Q-learning to approximate NE in weighted zero-sum polymatrix games, \cite{anagnostides2022last} utilized optimistic mirror descent to find NE in constant-sum polymatrix games, and \cite{ao2022asynchronous} applied optimistic multiplicative weight updates to find NE in zero-sum polymatrix games.

{In the setting with state transitions,  the networked structure has also been exploited recently in multi-agent RL \cite{zhang2018fully,zhang2018networked,qu2022scalable,liu2022scalable,zhang2023global,zhou2023convergence}, where either the communication or interaction, in terms of reward or transition, were assumed to have some networked structure. However, most of  these results were focused on  the {\it cooperative}  setting (or more generally the {\it potential} game setting). We instead focus on a multi-player while {\it non-cooperative}, specifically, {\it zero-sum}, setting.} 

{In the extensive form games literature, \cite{piliouras2022fast} proved that optimistic gradient ascent provides $\cO(1/T)$ convergence rate to NE in the network zero-sum extensive form games. }

\paragraph{Entropy regularization.} 
Entropy regularization is a common  approach used in reinforcement learning  to foster  exploration and enable faster convergence. Recently, both empirical evidence and provable convergence rate guarantees for entropy-regularized MDPs have been established \cite{williams1991function, peters2010relative,neu2017unified, haarnoja2017reinforcement,mei2020global,cen2022fast}. In addition to its applications in single-agent RL, entropy regularization has been investigated in game-theoretic settings, including two-player zero-sum matrix games \cite{cen2022faster}, multi-player zero-sum games \cite{leonardos2021exploration, ao2022asynchronous},  potential games \cite{cen2022independent}, and extensive-form games \cite{liu2023the,sokota2023a}. 

\paragraph{Fictitious play.}
Fictitious play is a  classical learning dynamics in game theory introduced by \cite{brown1951iterative}, in which players develop a belief in their opponent's policy and use a greedy approach to the belief they hold about the opponent's policy. (Stochastic) fictitious-play property ((S-)FPP) is a property of a game that ensures the convergence of (stochastic) fictitious play to a Nash equilibrium of the game. In the case of static games, (S-)FPP holds for two-player zero-sum games \cite{robinson1951iterative}, 2x$n$ games \cite{miyasawa1961convergence, berger2005fictitious}, $n$-player potential games \cite{monderer1996fictitious}, zero-sum polymatrix games  \cite{ewerhart2020fictitious}. However, FPP normally does not hold for 3x3 games \cite{shapley1964some}. For stochastic games \cite{shapley1953stochastic}, (S-)FPP holds for zero-sum and identical payoff games \cite{sayin2022fictitious, leslie2020best, baudin2022fictitious, sayin2021decentralized, sayin2022global}. Recently, \cite{sayin2022global} proved that any stochastic game with turn-based controllers on state transitions has S-FPP, as long as the stage payoffs have S-FPP. For a more detailed overview of fictitious play in stochastic/Markov  games, we refer to \cite{ozdaglar2021independent}.

{
\paragraph{Comparison with independent work \cite{kalogiannis2023zero}.} While preparing our work, we noticed an independent preprint \cite{kalogiannis2023zero}, which also studied the polymatrix zero-sum structure in Markov games. Encouragingly, they also showed the collapse of Markov CCE to Markov NE and thus their computational tractability.  However, there are several key differences that may be summarized as follows. First, the model in \cite{kalogiannis2023zero} is \emph{defined} as a combination of zero-sum polymatrix {\it reward functions} and \emph{switching-controller} dynamics, under which the desired property of equilibria collapse holds; in contrast, we define the model based on the payoffs of the {\it auxiliary games} at each state, which, by our \Cref{prop:MPGcond}, is \emph{equivalent} to the reward being zero-sum polymatrix and the dynamics being \emph{ensemble} (c.f. \Cref{rem:implication}). Our ensemble dynamics covers the switching controller case, and our model is more general in this sense. Second, our proof for equilibria collapse is different from that in \cite{kalogiannis2023zero}, which is based on characterizing the solution to some nonlinear program. We instead directly exploit the property of ensemble transition dynamics in marginalizing the joint policies{, and its effect on dynamic programming in finding the equilibria}.  
Third, {in terms of equilibrium computation, we investigate a series of value-iteration-based algorithms, based on both existing and our new algorithms for solving zero-sum polymatrix games,} with finite-iteration last-iterate convergence guarantees, including an  $\tilde{\cO}(1/\epsilon)$ rate result. In comparison, \cite{kalogiannis2023zero} uses existing algorithms for {\it learning} Markov CCE due to equilibria collapse, i.e., \cite{daskalakis2022complexity}. Finally, we have additionally provided hardness results for \emph{stationary} equilibria computation in infinite-horizon discounted settings, fictitious-play dynamics with convergence guarantees, as well as several examples of our model.} 
\numbering{\section{Omitted Details in \Cref{sec:MPGdef}}}
\tnumbering{\chapter{Omitted Details in \Cref{sec:MPGdef}}}

\label{appendix:MPGdef}
    \numbering{\subsection{Omitted proof for \Cref{prop:MPGcond} and \Cref{def:canonical}}}
    \tnumbering{\section{Omitted proof for \Cref{prop:MPGcond} and \Cref{def:canonical}}}    
\MZNMGcond* 
\begin{proof} 
     Firstly, in the case with  $\cN_C \neq \emptyset$, we prove that if an MG satisfies the decomposability of the reward function $r_i(s, a_i, \cdot)$ and the transition {dynamics} $\PP(s'\mid s,\cdot)$, then $Q_{i,j}^V$ can be decomposed  as follows: 
\begin{align*}
    &Q_i^V(s, \ba) = r_i(s, \ba) + \gamma  \EE_{s' \sim \PP(\cdot \mid s, \ba)}V(s') 
     = \sum_{j\in\cE_{Q, i}} r_{i,j}(s, a_i,a_j) + \gamma \sum_{j\in\cN_C} \sum_{s' \in \cS} \FF_j(s' \mid s, a_j) V(s')
    \\
    & = \sum_{j\in\cE_{Q, i}}\left( r_{i,j}(s, a_i,a_j) + \gamma \sum_{s' \in \cS} \left(\lambda_{i,j}(s)  \pmb{1}(i \in \cN_C)\FF_i(s' \mid s, a_i) + \pmb{1}(j \in \cN_C)\FF_j(s' \mid s, a_j)\right) V(s')\right)\\
    & =: \sum_{j \in \cE_{Q, i}} Q_{i,j}^V(s, a_i, a_j),
\end{align*}
for any non-negative $(\lambda_{i,j}(s))_{(i,j)\in\cE_Q}$ such that $\sum_{j \in \cE_{Q, i}}\lambda_{i,j}(s) = 1$, since by definition  $\cN_C \subseteq \cE_{Q, i}$ for every $i \in \cN$. 

In the case when $\cN_C = \emptyset$, we prove that if the MG satisfies decomposability of $r_i(s, a_i, \cdot)$ and $\PP(\cdot|s,\ba) = \FF_o (\cdot|s)${, then} $Q_{i,j}^V$ can be decomposed as follows:
\begin{align*}
    &Q_i^V(s, \ba) = r_i(s, \ba) + \gamma  \EE_{s' \sim \PP(\cdot \mid s, \ba)}V(s') 
       = \sum_{j\in\cE_{Q, i}} \left(r_{i,j}(s, a_i,a_j) + \gamma \sum_{s' \in \cS}\lambda_{i,j}(s)  \FF_o(s'\mid s)  V(s')\right)\\
    &\quad
     =: \sum_{j \in \cE_{Q, i}} Q_{i,j}^V(s, a_i, a_j),
\end{align*} 
for any non-negative $(\lambda_{i,j}(s))_{(i,j)\in\cE_Q}$ such that $\sum_{j \in \cE_{Q, i}}\lambda_{i,j}(s) = 1$.

Next, we prove the necessary conditions for an MG to be an NMG. By definition, we have
\begin{align*}
    Q_i^V(s, \ba) = \sum_{j \in \cE_{Q, i}} Q_{i,j}^V(s, a_i, a_j) = r_{i}(s,\ba) + \gamma \langle \PP(\cdot \mid s, \ba), V(\cdot) \rangle
\end{align*}
{for any $V$, } 
which indicates that 
\begin{align}
    \sum_{j \in \cE_{Q, i}} (Q_{i,j}^V(s, a_i, a_j) - Q_{i,j}^{V'}(s, a_i, a_j)) = \gamma \langle \PP(\cdot \mid s, \ba), V(\cdot) - V'(\cdot) \rangle, \label{eqn:subtract}
\end{align} 
for any $V, V'$ and any $(s, \ba)$. 

For every $s \in \cS$, define $B_s:\cS \to \RR$ such that $B_s(s')= \pmb{1}(s= s')$. {We define $\mathbb{G}_{i,j}(s'\mid s,a_i,a_j) := 1/\gamma \sum_{j \in \cE_{Q, i}}(Q_{i,j}^{B_{s'}}(s, a_i, a_j) - Q_{i,j}^{\pmb{0}}(s, a_i, a_j)) $, then by plugging in $V= B_{s'}$ and $V'= \pmb{0}$, we can derive  $\PP(\cdot \mid s, \ba) = \sum_{j \in \cE_{Q, i}} \mathbb{G}_{i,j}(\cdot \mid s, a_i, a_j)$ for every $i$ from \Cref{eqn:subtract}. Alternatively,  we can take a functional derivative of \Cref{eqn:subtract} } 
with respect to $V - V'${, which} means that there exist some functions $\{\mathbb{G}_{i,j}\}_{j \in \cE_{Q, i}}$ such that $\PP(\cdot \mid s, \ba) = \sum_{j \in \cE_{Q, i}} \mathbb{G}_{i,j}(\cdot \mid s, a_i, a_j)$ for every $i$. {Note that the decomposability of $\PP$ with respect to $\cE_{Q,i}$ above has to hold for all $i\in\cN$.} 
Therefore, {when $\cN_C\neq \emptyset$,} if $j \notin \cN_C$, {then} there exists some $i\in\cN$ such that $(i,j) \notin \cE_Q$. {In this case,} $\PP$ is not {dependent on} {this} $j$. So $\PP$ should be a function of the players in $\cN_C$, which indicates that $\PP(\cdot \mid s, \ba) = \sum_{j \in \cN_C, j \neq i} \FF_{i,j}(\cdot \mid s, a_i, a_j)$ for every $i$ unless $\cN_C = \emptyset$. If $\cN_C = \emptyset$, then it directly concludes that $\PP(\cdot \mid s, \ba) = \FF_o(\cdot\mid s)$ for some $\FF_o$, since by the argument above, it should not depend on any player $j$. {Next, we focus on the case when $\cN_C\neq \emptyset$, 
and there are more than two players.}

Specifically, in this case, 
if there exist   some  $k_1 \neq k_2$ and $k_1\notin \cN_C$, such that $$\PP(\cdot|s, \ba) = \sum_{i \in \cN_C, i \neq k_1} \FF_{k_1, i}(\cdot | s, a_{k_1}, a_i) = \sum_{i \in \cN_C, i \neq k_2} \FF_{k_2, i}(\cdot | s, a_{k_2}, a_i),$$  then we  choose a fixed  $a_{k_1}$ but changing  $a_{k_2}$ arbitrarily. This preserves the equality, indicating  that  
\begin{align}
\sum_{i \in \cN_C, i \neq k_1} \FF_{k_1, i}(\cdot | s, a_{k_1, \text{fix}}, a_i) = \sum_{i \in \cN_C, i \neq k_2} \FF_{k_2, i}(\cdot | s, a_{k_2}, a_i)    \label{eqn:two-different-player}
\end{align}
for any $a_{k_2}$. Since $k_1 \notin \cN_C$, we can have the left-hand side of \eqref{eqn:two-different-player}  written as $\sum_{i\in\cN_C}\FF_{i}(\cdot|s, a_{i})$ by setting  $\FF_{i}(\cdot|s, a_{i}) := \FF_{k_1, i}(\cdot|s,a_{k_1, \text{fix}}, a_i)$. {Meanwhile, the right-hand  side of \eqref{eqn:two-different-player} is also $\PP(\cdot|s, \ba)$ by definition, which concludes that $\PP(\cdot|s, \ba) = \sum_{i \in \cN_C} \FF_i(\cdot|s, a_i)$, and  proves the theorem.}  
In other words, as long as at least one $k_1$ does not belong to $\cN_C$, we can conclude the theorem.

If no such a $k_1\notin\cN_C$ exists, then it means that all players are in $\cN_C$. 
In this case, for any $k_1\neq k_2$, for a fixed $a_{k_1, \text{fix}}$, we have
\begin{align}
\sum_{i \in \cN_C/\{k_1\}} \FF_{k_1, i}(\cdot | s, a_{k_1, \text{fix}}, a_i) = \sum_{i \in \cN_C/\{k_1, k_2\}} \FF_{k_2, i}(\cdot | s, a_{k_2}, a_i)  + \FF_{k_2, k_1}(\cdot | s, a_{k_2}, a_{k_1,\text{fix}}).  \label{eqn:two-different-player-2}
\end{align}
Therefore, if we define $\mathbb{G}_i(\cdot|s, a_i): =  \FF_{k_1, i}(\cdot | s, a_{k_1, \text{fix}}, a_i)$ for $i \in \cN_C/\{k_1, k_2\}$, and $\mathbb{G}_{k_2}(\cdot|s, a_{k_2}): =  \FF_{k_1, k_2}(\cdot | s, a_{k_1, \text{fix}}, a_{k_2}) - \FF_{k_2, k_1}(\cdot | s, a_{k_2}, a_{k_1,\text{fix}})$,  then we have 
\begin{align}
    \sum_{i \in \cN_C/\{k_1, k_2\}} \FF_{k_2, i}(\cdot | s, a_{k_2}, a_i) = \sum_{i \in \cN_C/\{k_1\}}\mathbb{G}_i(\cdot|s, a_i).
    \label{eqn:two-sum-decomp2}
\end{align}
Let $k_3 \in \cN_C$ such that $k_3 \neq k_1, k_2$. By definition of $\FF_{i,j}$, we have 
\begin{align*}
    \PP(\cdot| s,\ba) &= \sum_{i \in \cN_C / \{k_3\}} \FF_{k_3, i} (\cdot |s, a_{k_3}, a_i) = \sum_{i \in \cN_C / \{k_2\}} \FF_{k_2, i} (\cdot |s, a_{k_2}, a_i). 
\end{align*}
Plugging \eqref{eqn:two-sum-decomp2}, we have
\begin{align}
    \PP(\cdot| s,\ba) &= \sum_{i \in \cN_C / \{k_3\}} \FF_{k_3, i} (\cdot |s, a_{k_3}, a_i) = \sum_{i \in \cN_C / \{k_1\}} \mathbb{G}_{i} (\cdot |s, a_i) + \FF_{k_2, k_1}(\cdot|s, a_{k_2}, a_{k_1}) \label{eqn:two-sum-decomp3}
\end{align}
for any $\ba \in \cA$. 
If we now fix $a_{k_3}$ as $a_{k_3, \text{fix}}$, then from \eqref{eqn:two-sum-decomp3}  we know that 
\begin{align}
    \sum_{i \in \cN_C / \{k_3\}} \FF_{k_3, i} (\cdot |s, a_{k_3, \text{fix}}, a_i) = \sum_{i \in \cN_C / \{k_1, k_3\}} \mathbb{G}_{i} (\cdot |s, a_i) +  \mathbb{G}_{k_3} (\cdot |s, a_{k_3, \text{fix}}) + \FF_{k_2, k_1}(\cdot|s, a_{k_2}, a_{k_1}). \label{eqn:two-sum-decomp4}
\end{align}
Plugging \eqref{eqn:two-sum-decomp4} to \eqref{eqn:two-sum-decomp3}, we have 
\begin{align*}
    \PP(\cdot| s,\ba) & = \sum_{i \in \cN_C / \{k_1, k_3\}} \mathbb{G}_{i} (\cdot |s, a_i) + \FF_{k_2, k_1}(\cdot|s, a_{k_2}, a_{k_1}) + \mathbb{G}_{k_3} (\cdot |s, a_{k_3})
    \\
    &= \sum_{i \in \cN_C / \{k_3\}} \FF_{k_3, i} (\cdot |s, a_{k_3, \text{fix}}, a_i) -\mathbb{G}_{k_3} (\cdot |s, a_{k_3, \text{fix}}) + \mathbb{G}_{k_3} (\cdot |s, a_{k_3})
     =: \sum_{i \in \cN_C } \FF_i(\cdot|s, a_i)
\end{align*}
where $\FF_i(\cdot | s, a_i) := \FF_{k_3, i}(\cdot |s, a_{k_3, \text{fix}}, a_i)$ for $i \in \cN_C / \{k_3\}$ and $\FF_{k_3}(\cdot| s, a_{k_3}) := -\mathbb{G}_{k_3} (\cdot |s, a_{k_3, \text{fix}}) + \mathbb{G}_{k_3} (\cdot |s, a_{k_3})$, which concludes the decomposability of the transition dynamics. 
Finally, note that we can ensure  the non-negativity  of $\FF_i$, since we can iterate the following procedure:
\begin{algorithm}[H]
\caption{Procedure for constructing non-negative $\{\FF_i\}_{i \in \cN_C}$}
\label{procedure:non-negative}
\begin{algorithmic}
\WHILE{there exists $s, s' \in \cS$ and $i\in \cN_C$ such that $\min_{a_i \in \cA_i} \FF_i(s'\mid s, a_i) < 0$}
\STATE{Set $s, s' \in \cS$ and $i\in \cN_C$ that satisfying $\min_{a_i \in \cA_i} \FF_i(s'\mid s, a_i) < 0$}
\STATE{Sort $\cN_C$ according to  the descending  order of  $b_k: = \min_{a_k\in \cA_k} \FF_k(s'\mid s, a_k)$, and denote it as $\{j_1, j_2, \dots, j_{|\cN_C|}\}$}
\STATE{Define $\text{tmp1} = 0, t= 1$}
\WHILE{$\text{tmp1} <  -\min_{a_i \in \cA_i} \FF_i(s'\mid s, a_i)$}
\STATE{Define $\text{tmp2} =  -\min_{a_i \in \cA_i}\FF_i(s'\mid s, a_i) - \text{tmp1} > 0 $}
\STATE Define $\text{tmp3} = \min(\min_{ a_{j_t} \in \cA_{j_t}} \FF_{j_t}(s'\mid s, a_{j_t}), \text{tmp2}) > 0$;
\STATE Update $\FF_{j_t}(s' | s, a_{j_t}) \leftarrow  \FF_{j_t}(s'\mid s, a_{j_t}) -  \text{tmp3}$ for all $a_{j_t} \in \cA_{j_t}$
\STATE Update $\FF_i(s'\mid s, a_i) \leftarrow \FF_i(s'\mid s, a_i) + \text{tmp3} $ for all $a_i \in \cA_i$
\STATE Update $\text{tmp1} \leftarrow \text{tmp1} + \text{tmp3}$ and $t \leftarrow t + 1$
\ENDWHILE
\ENDWHILE
\end{algorithmic}
\end{algorithm}
In this context, the third and fourth lines of the inner-while loop ensure that $\PP(s' | s, \ba) = \sum_{j \in \cN_C} \FF_{j}(s'\mid s, a_j)$ always holds for $s, s' \in \cS$, $\ba \in \cA$. Furthermore, $\text{tmp}3$ remains greater than 0 within the inner-while loop. This is because if $\min_{a_{j_t} \in \cA_{j_t}} \FF_{j_t}(s'\mid s, a_{j_t}) \leq 0$, then it implies that $\min_{ a_{j_k} \in \cA_{j_k}} \FF_{j_k}(s'\mid s, a_{j_k}) = 0$ for all $k \in [t]$. Also, $\min_{ a_{j_k} \in \cA_{j_k}} \FF_{j_k}(s'\mid s, a_{j_k}) \leq \min_{ a_{j_t} \in \cA_{j_t}} \FF_{j_t}(s'\mid s, a_{j_t}) \leq 0$ for every $k \in [|\cN_C|]$. Meanwhile, since we assumed $\min_{a_i \in \cA_i} \FF_i(s'\mid s, a_i) < 0$, we can define a vector $\tilde{a}$ with elements $\tilde{a}_j \in \argmin_{a_j\in \cA_j} \FF_j(s'\mid s, a_j)$. This implies that $\PP(s'\mid s, \tilde{a}) =\sum_{i \in \cN_C} \min_{a_i \in \cA_i} \FF_i(s'\mid s, a_i) < 0 $, which contradicts the condition $\PP(s'\mid s, \tilde{a}) \geq 0$. 

Therefore, our procedure ensures that the number of pairs $(s, s') \in \cS \times \cS$ and indexes $i \in \cN_C$ for which $\min_{a_i \in \cA_i} \FF_i(s'\mid s, a_i) < 0$ is consistently reduced.
 
If there are only two players, then both of them belong to $\cN_C$. However, one cannot decompose the transition dynamics as above, as there is no such a $k_3\neq k_1,k_2$ to construct the aforementioned formula. Indeed, any two-player (zero-sum) MG satisfies our definition of zero-sum NMGs in \Cref{sec:def_MZNG}.

For the reward decomposition, if $\cN_C \neq \emptyset$, we have
\begin{align*}
     r_{i}(s,\ba)  = \sum_{j \in \cE_{Q, i}} \left( Q_{i,j}^V(s, a_i, a_j) - \gamma   \langle \pmb{1}(j \in \cN_C)\FF_j(\cdot \mid s, a_j), V(\cdot) \rangle\right)
\end{align*}
and if $\cN_C = \emptyset$, we have 
\begin{align*}
     r_{i}(s,\ba)  = \sum_{j \in \cE_{Q, i}} \left( Q_{i,j}^V(s, a_i, a_j) - \gamma   \langle \frac{1}{|\cE_{Q, i}|}\FF_o(\cdot \mid s), V(\cdot) \rangle\right).
\end{align*}
{Hence, $r_i(s, a_i, \cdot)$ can be represented as $r_i(s, \ba)  = \sum_{j \in \cE_{Q, i}} r_{i, j}(s, a_i, a_j)$ for some functions $(r_{i, j})_{(i,j)\in\cE_Q}$. The same procedure as  \Cref{procedure:non-negative} provides that we can ensure the non-negativity of $r_{i, j}$, so that $r_i(s, a_i, \cdot)$ is decomposable with respect to $\cE_{Q, i}$. {In fact, adding any large-enough constant to $r_{i,j}$ does not change the solution to the problem, while ensuring the non-negativity of $r_{i,j}$.}}
\end{proof}

\vspace{0.2in}
\noindent \textbf{Proposition 1 - Finite-horizon version. } For a given graph $\cG=(\cN, \cE_Q)$, an MG $(\cN, \cS, \cA, (\PP_h)_{h \in [H]}, (r_{h, i})_{i \in \cN, h \in [H]})$ with more than two players is an NMG with $\cG$ if and only if: (1) \textbf{$r_{h, i}(s, a_i, \cdot)$ is decomposable with respect to} $\cE_{Q, i}$ for each $i \in \cN\,, s \in \cS\,, a_i \in \cA_i,\, h \in [H]$, i.e., $r_{h, i}(s, \ba) = \sum_{j \in \cE_{Q,i}} r_{h, i,j}(s, a_i, a_j)$ for a set of functions $\{r_{h, i, j}(s, a_i ,\cdot)\}_{j \in \cE_{Q, i}}$ and (2)  \textbf{the transition dynamics $\PP_h(s'\mid s, \cdot)$ is decomposable with respect to}$\cN_C$ corresponding to this $\cG$ for all $h \in [H]$, {i.e.,} $\PP_h(s'\mid s, \ba) = \sum_{i \in \cN_C} \FF_{h,i}(s'\mid s, a_i)$ for a set of functions $\{\FF_{h,i}(s'\mid s,\cdot)\}_{i\in\cN_C}$ if $\cN_C \neq \emptyset$, or $\PP_h(s'\mid s, \ba) = \FF_{h, o}(s'\mid s)$ for some constant function (of $\ba$) $\FF_{h, o}(s'\mid s)$ {if $\cN_C = \emptyset$}. 
Moreover, an MG qualifies as a zero-sum NMG if and only if it satisfies an additional condition: the NG, characterized by $(\cG, \cA, (r_{h, i,j}(s))_{(i,j) \in \cE_Q})$, must be a zero-sum NG {for all $s\in\cS, h \in [H]$}. In the case of two players, every (zero-sum) Markov game becomes a (zero-sum) NMG.   

\vspace{0.2in}
\noindent \textbf{Proposition 2} (Decomposition of $(Q_i^V)_{i \in \cN}$)\textbf{.}
For {an infinite-horizon $\gamma$-discounted} NMG with $\cG = (\cN, \cE_Q)$ such that $\cN_C \neq \emptyset$, if we know that $\PP(s'\mid s, \ba) = \sum_{i \in \cN_C} \FF_i(s'\mid s, a_i)$, and $r_{i}(s, \ba) = \sum_{j \in \cE_{Q,i}} r_{i,j}(s, a_i, a_j)$ for some $\{\FF_i\}_{i \in \cN_C}$ and $\{r_{i,j}\}_{(i,j) \in \cE_Q}$,  then {the $Q_{i,j}^V$  given in \Cref{sec:def_MZNG} can be represented as}
{\neurips{\fontsize{9}{9}\selectfont}
\numbering{\begin{align*}
Q_{i,j}^V(s, a_i, a_j) = r_{i,j}(s,a_i, a_j) + \sum_{s' \in \cS} \gamma \left(\pmb{1}(j \in \cN_C) \FF_j(s'\mid s, a_j) + \pmb{1}(i \in \cN_C)\lambda_{i,j}(s)\FF_i(s'\mid s,a_i) \right) V(s')
\end{align*}}}
\tnumbering{\begin{align*}
Q_{i,j}^V(s, a_i, a_j) &= r_{i,j}(s,a_i, a_j) 
\\
&\quad+ \sum_{s' \in \cS} \gamma \left(\pmb{1}(j \in \cN_C) \FF_j(s'\mid s, a_j) + \pmb{1}(i \in \cN_C)\lambda_{i,j}(s)\FF_i(s'\mid s,a_i) \right) V(s')
\end{align*}}
for any {non-negative} {$(\lambda_{i,j}(s))_{(i,j) \in \cE_Q}$}  such that $\sum_{j \in \cE_{Q,i}} \lambda_{i,j}(s) = 1$ for all {$i\in\cN$ and $s\in \cS$}.
For {an infinite-horizon $\gamma$-discounted} NMG with $\cG = (\cN, \cE_Q)$ such that  $\cN_C = \emptyset$, if we know that $\PP(s'\mid s, \ba) = \FF_o(s'\mid s)$, and $r_{i}(s, \ba) = \sum_{j \in \cE_{Q,i}} r_{i,j}(s, a_i, a_j)$ for some $\FF_o$ and $\{r_{i,j}\}_{(i,j) \in \cE_Q}$, then
{the $Q_{i,j}^V$  given in \Cref{sec:def_MZNG} can be represented as}
{
\cpdelete{\fontsize{9}{9}\selectfont}
\begin{align*}
Q_{i,j}^V(s, a_i, a_j) = r_{i,j}(s,a_i, a_j) + \sum_{s' \in \cS} \gamma \left(\lambda_{i,j}(s)\FF_o(s'\mid s) \right) V(s')
\end{align*}}
\noindent for any {non-negative} {$(\lambda_{i,j}(s))_{(i,j) \in \cE_Q}$}  such that $\sum_{j \in \cE_{Q,i}} \lambda_{i,j}(s) = 1$ for all $i\in\cN$ and $s\in\cS$. We call it the \textit{canonical} decomposition of $\{Q_i^V\}_{i \in \cN}$  {when $Q_{i,j}^V$ can be represented as above with} $\lambda_{i,j}(s) = {1}/{|\cE_{Q,i}|}$ for $j \in \cE_{Q, i}$.

\Cref{def:canonical} naturally follows from the proof of \Cref{prop:MPGcond} (a).  

\vspace{0.2in}
\noindent \textbf{Proposition 2 - Finite-horizon version} (Decomposition of $(Q_{h,i}^V)_{i \in \cN}$)\textbf{.}
For {a finite-horizon} NMG with $\cG = (\cN, \cE_Q)$ with $\cN_C \neq \emptyset$, {we know } $\PP_h(s'\mid s, \ba) = \sum_{i \in \cN_C} \FF_{h,i}(s'\mid s, a_i)$, and $r_{h,i}(s, \ba) = \sum_{j \in \cE_{Q,i}} r_{h,i,j}(s, a_i, a_j)$ for some $\{\FF_{h,i}\}_{i \in \cN_C}$ and $\{r_{h, i,j}\}_{(i,j) \in \cE_Q}$, {then} {the $Q_{h, i,j}^V$  given in \Cref{sec:def_MZNG} can be represented as}
\arxiv{\numbering{
\begin{align*}
Q_{h, i,j}^V(s, a_i, a_j) = r_{h, i,j}(s,a_i, a_j) + \sum_{s' \in \cS}\left(\pmb{1}(j \in \cN_C) \FF_{h, j}(s'\mid s, a_j) + \pmb{1}(i \in \cN_C)\lambda_{h, i,j}(s)\FF_{h, i}(s'\mid s,a_i) \right) V_{h+1}(s')
\end{align*}}}
\tnumbering{\begin{align*}
Q_{h, i,j}^V&(s, a_i, a_j) = r_{h, i,j}(s,a_i, a_j) 
\\
&+ \sum_{s' \in \cS}\left(\pmb{1}(j \in \cN_C) \FF_{h, j}(s'\mid s, a_j) + \pmb{1}(i \in \cN_C)\lambda_{h, i,j}(s)\FF_{h, i}(s'\mid s,a_i) \right) V_{h+1}(s')
\end{align*}}
\neurips{
\begin{align*}
Q_{h, i,j}^V&(s, a_i, a_j) = r_{h, i,j}(s,a_i, a_j) 
\\
&+ \sum_{s' \in \cS}\left(\pmb{1}(j \in \cN_C) \FF_{h, j}(s'\mid s, a_j) + \pmb{1}(i \in \cN_C)\lambda_{h, i,j}(s)\FF_{h, i}(s'\mid s,a_i) \right) V_{h+1}(s')
\end{align*}}
for any {non-negative} {$(\lambda_{h, i,j}(s))_{(i,j) \in \cE_Q}$}  such that $\sum_{j \in \cE_{Q,i}} \lambda_{h, i,j}(s) = 1$ for all {$i\in\cN$, $s\in \cS$, and $h\in[H]$}. For {a finite-horizon} NMG with $\cG = (\cN, \cE_Q)$ such that  $\cN_C = \emptyset$, $\PP_h(s'\mid s, \ba) = \FF_{h,o}(s'\mid s)$, and $r_{h,i}(s, \ba) = \sum_{j \in \cE_{Q,i}} r_{h, i,j}(s, a_i, a_j)$ for some $\FF_{h,o}$ and $\{r_{h,i,j}\}_{(i,j) \in \cE_Q}$, 
{the $Q_{h,i,j}^V$  given in \Cref{sec:def_MZNG} can be represented as}
\begin{align*}
Q_{h, i,j}^V(s, a_i, a_j) = r_{h, i,j}(s,a_i, a_j) + \sum_{s' \in \cS} \left(\lambda_{h, i,j}(s)\FF_{h, o}(s'\mid s) \right) V_{h+1}(s')
\end{align*} 
for any {non-negative} {$(\lambda_{h, i,j}(s))_{(i,j) \in \cE_Q}$}  such that $\sum_{j \in \cE_{Q,i}} \lambda_{h, i,j}(s) = 1$ for all $i\in\cN$, $s\in\cS$, and $h\in[H]$. We call it the \textit{canonical} decomposition of $Q$-value functions {when $Q_{h, i,j}^V$ can be represented as above with} $\lambda_{h, i,j}(s) = {1}/{|\cE_{Q,i}|}$ for $j \in \cE_{Q, i}$.

\numbering{\subsection{An alternative definition of NMGs}}
\tnumbering{\section{An alternative definition of NMGs}}
\label{ssec:alt_def}}

\begin{definition}[An alternative definition of NMGs]\label{def:alternative_MNMG}
An infinite-horizon $\gamma$-discounted MG is called a \textit{Multi-player MG with Networked separable interactions (NMG)}  characterized by  a tuple $$(\cG = (\cN, \cE_Q), \cS, \cA, \PP, (r_{i})_{i \in \cN}, \gamma)$$  
{if for any \emph{policy}  $\pi$, there exist a {set of} functions $(Q_{i,j}^\pi)_{(i, j) \in \cE_Q}$ {and an undirected connected graph $\cG = (\cN, \cE_Q)$} such that $
    Q_{i}^\pi(s, \ba) = \sum_{j\in \cE_{Q, i}} Q_{i,j}^\pi(s, a_i, a_j)$ }
holds for every $i \in \cN$, $s \in \cS$, $\ba \in \cA$. A finite-horizon MG is called a \textit{Multi-player MG with Networked separable interactions} {if for any \emph{policy} $\pi$, there exist a set of functions $(Q_{h, i,j}^\pi)_{(i, j) \in \cE_Q,h\in[H]}$ such that $Q_{h, i}^\pi(s, \ba) = \sum_{j\in\cE_{Q, i}} Q_{h, i,j}^\pi(s, a_i, a_j)
$ holds for every $i\in \cN$, $s \in \cS$, $\ba \in \cA$, and $h\in[H]$.} 
\end{definition}

The ``if'' condition of \Cref{prop:MPGcond} also holds for this definition. However, the proof for the ``only if'' condition of \Cref{prop:MPGcond} uses the functional derivative argument. We cannot use the functional derivative here directly, since $\bV^\pi: =(V^\pi(s))_{s \in \cS}$ cannot represent {all vectors in $[0,R/(1-\gamma)]^{|\cS|}$}. To be specific, we have
\begin{align*}
    Q_i^\pi(s, \ba) = \sum_{j \in \cE_{Q, i}} Q_{i,j}^\pi(s, a_i, a_j) = r_{i}(s,\ba) + \gamma \langle \PP(\cdot \mid s, \ba), V_i^\pi(\cdot) \rangle
\end{align*}
{for any policies $\pi,\pi'$, }
which indicates
\begin{align}
    \sum_{j \in \cE_{Q, i}} (Q_{i,j}^\pi(s, a_i, a_j) - Q_{i,j}^{\pi'}(s, a_i, a_j)) = \gamma \langle \PP(\cdot \mid s, \ba), V_i^\pi(\cdot) - V_i^{\pi'}(\cdot) \rangle. \label{eqn:pi-decomposition-definition}
\end{align}

{If we can find a set of policies $(\pi^{(k)})_{k \in [|\cS|]}$ such that the vectors in $\{\bV_i^{\pi^{(k)}} - \bV_i^{\pi'}\}_{k \in [|\cS|]}$ are independent for some fixed $\pi'$, then we can concatenate the vectors $\{\bV_i^{\pi^{(k)}} - \bV_i^{\pi'}\}_{k \in [|\cS|]}$ together as a matrix of size $|\cS|\times|\cS|$ that is full-rank, and solve for $\PP(\cdot\given s,\ba)$ by solving the linear equations   \eqref{eqn:pi-decomposition-definition}, for some  fixed $(s,\ba)$. This way, we can show that $\PP(\cdot\given s,\ba) = \sum_{j \in \cE_{Q, i}} \FF_{i,j}(\cdot\given s, a_i, a_j)$, and the rest of the proof follows from that of \Cref{prop:MPGcond}. However, such a set of policies $(\pi^{(k)})_{k \in [|\cS|]}$ (and $\pi'$) may not exist in some degenerate cases, as to be detailed below.}  
\begin{definition}[Degenerate MG with respect to player $i$] 
\label{def:degnerateMG1}
    We call an MG {\it degenerate} with respect to player $i \in \cN$ if there exists some  $s \in \cS$ such that for any $\pi$,  $Q^\pi_i(s, \ba)$ is a constant function of $\ba \in \cA$. 
\end{definition}

\begin{definition}[Non-degenerate MG]
\label{def:non-dege_MG} 
    We call an MG {\it non-degenerate} if an MG is {not} degenerate with respect to {any} player $i \in \cN$. 
\end{definition}

{Now, we are ready to state the counterpart of \Cref{prop:MPGcond} in these non-degenerate  cases.} The proof for the ``if'' direction is exactly the same as that of \Cref{prop:MPGcond}. We focus on the proof of the ``only if'' statement.  

\vspace{7pt}
\noindent \textbf{Proposition 1 - An alternative definition version.} For a given graph  $\cG=(\cN, \cE_Q)$, a  {\it non-degenerate} MG $(\cN, \cS, \cA, \PP, (r_{i})_{i \in \cN}, \gamma)$ {(in the sense of  \Cref{def:non-dege_MG})} with more than two players is an NMG with $\cG$ if and only if:
(1) \textbf{$r_i(s, a_i, \cdot)$ is decomposable with respect to} $\cE_{Q, i}$ for each $i \in \cN, s \in \cS, a_i \in \cA_i $, {i.e.,} $r_i(s, \ba) = \sum_{j \in \cE_{Q,i}} r_{i,j}(s, a_i, a_j) $ for a set of functions $\{r_{i,j}(s, a_i, \cdot)\}_{j \in \cE_{Q, i}}$, and {(2) \textbf{the transition dynamics $\PP(s'\mid s, \cdot)$ is decomposable with respect to} the $\cN_C$ corresponding to this $\cG$, {i.e.,} $\PP(s'\mid s, \ba) = \sum_{i \in \cN_C} \FF_i(s'\mid s, a_i)$ for a set of functions $\{\FF_i(s'\mid s,\cdot)\}_{i\in\cN_C}$ if $\cN_C \neq \emptyset$, or $\PP(s'\mid s, \ba) = \FF_o(s'\mid s)$ for some constant function (of $\ba$) $\FF_o(s'\mid s)$ {if $\cN_C = \emptyset$}.}
\begin{proof}
Proofs of the claims are deferred to \Cref{ssec:deferred-proof-alternative-definition} to preserve the flow of the argument. 
\begin{restatable}{claim}{independent}
\label{claim:independent}
    For {player} $i \in \cN$, if there exists a policy $\pi$ such that for every $s\in\cS$, there exist $\ba_{s,1}, \ba_{s,2}$ that make  $Q_i^\pi(s, \ba_{s, 1}) \neq Q_i^\pi(s, \ba_{s, 2})$, then we can construct $|\cS|$ number of policies {$(\pi^{(k)})_{k\in[|\cS|]}$} such that  $\{\bV_i^{\pi^{(k)}} - \bV_i^{\pi}\}_{k \in [|\cS|]}$ are independent for any fixed $\pi$. 
\end{restatable}
Therefore, for player $i \in \cN$, if the condition of \Cref{claim:independent} holds, then we can guarantee that there exist some functions $\{\mathbb{G}_{i,j}\}_{j \in \cE_{Q, i}}$ such that $\PP(\cdot|s, \ba) = \sum_{j \in \cE_{Q, i}} \mathbb{G}_{i,j}(\cdot|s, a_i, a_j)$, by the argument after \Cref{eqn:pi-decomposition-definition}. If we assume that the condition of \Cref{claim:independent} holds for every player $i \in \cN$, then there exist some functions $\{\mathbb{G}_{i,j}\}_{(i, j) \in \cE_Q}$ such that $\PP(\cdot|s, \ba) = \sum_{j \in \cE_{Q, i}} \mathbb{G}_{i,j}(\cdot|s, a_i, a_j)$ for every $i \in \cN$. Then, we can prove the decomposability of $\PP(s'\mid s, \cdot)$ with respect to $\cN_C$ with the same steps as Equations~\eqref{eqn:two-different-player} - \eqref{eqn:two-sum-decomp4}. 

Hence, by \Cref{claim:independent}, we only need to prove that if an MG is non-degenerate with respect to player $i$, then there exists a policy $\pi$ such that there exist $\ba_{s,1}, \ba_{s,2}$ that make $Q_i^\pi(s, \ba_{s, 1}) \neq Q_i^\pi(s, \ba_{s, 2})$ for every $s \in \cS$. 
If $|\cA| = 1$, then the transition dynamics are already decomposed, so we do not need to consider this case.

\begin{restatable}{claim}{univsforconstant}
 \label{claim:univ-s-for-constant}
Assume $|\cA| \geq 2$. For player $i \in \cN$, assume that for any policy $\pi$, there exists a state $s_\pi\in\cS$ such that $Q^\pi_i(s_\pi, \ba)$ is a constant function of $\ba$.  Then, there exists {a state} $s\in \cS$ such that {uniformly} for {any policy}  $\pi$, $Q^\pi_i(s, \ba)$ is a constant function of $\ba$.   
\end{restatable}

\Cref{claim:univ-s-for-constant} shows that if the assumption of \Cref{claim:independent} does not hold for player $i \in \cN$, then $$\cS_{\text{const}, i}:= \{s \mid \text{For any }\pi,  Q^\pi_i(s, \ba)\text{ is a constant function of } \ba\}$$ is not an empty set. %
By \Cref{def:degnerateMG1}, 
if an MG is {not} degenerate with respect to player $i$, then {$\cS_{\text{const}, i}$ is empty, and thus} the conditions of \Cref{claim:independent} always hold for player $i$. Therefore, if we assume the non-degeneracy of MG (\Cref{def:non-dege_MG}), by \Cref{claim:independent} and the arguments immediately following it, $\PP(s'\mid s, \cdot)$ is decomposable with respect to $\cN_C$. 
\end{proof}

{\begin{remark}[Degenerate MG with respect to player $i$]
The degeneracy of MGs as defined above 
 can indeed be rare. 
To illustrate, consider a seemingly degenerate scenario where for all $s\in\cS$, \(r_i(s, \ba)\) remains a constant function with respect to \(\ba\); even under such circumstances, it is possible for the MG to be non-degenerate with respect to player \(i\). {For example, assume that $\cN = [2], \gamma = \frac{1}{2}$, $\cS = \{-1, 1\}$ $\cA_i = \{-1, 1\}$, $\PP(s'\given s, a_1, a_2) = \pmb{1}(s' = sa_1a_2), r_1(s, \ba) =r_2(s, \ba) =  (1 + s)$. Let $\pi$ satisfy $\pi(\ba\given 1) = \pmb{1}(a_1 = 1, a_2 = -1), \pi(\ba\given -1) = \pmb{1}(a_1 = -1, a_2 = -1)$. Then, we have $V^\pi_1(1) = 2 + \frac{1}{2}V^\pi_1(-1)$ and $V^\pi_1(-1) = \frac{1}{2}V^\pi_1(-1)$,  which further means $V^\pi_1(-1) = 0$ and $V^\pi_1(1) = 2$, i.e., $V^\pi_1(s)=1+s$ always holds. As a result, $Q^\pi_1(1, \ba) = 1+ s  + \frac{1}{2}\left(sa_1a_2 + 1\right)$,  
which is not a constant function, and the game is thus non-degenerate with respect to agent $i$.} 
Moreover, suppose that the policy, transition dynamics, and reward functions  are randomly chosen. The measure of  the event that the existence of $s \in \cS$ such that $Q_i^\pi(s, \ba)$ is constant for all possible actions $\ba$ under this randomly chosen policy, transition dynamics, and reward function is 0.  This  is primarily because \(Q_i^\pi(s, \cdot) : \cA \to \RR\) must lie on a particular hyperplane in the overall value function space, which takes measure $0$. Oftentimes, different actions will transition to different states and yield different rewards, thereby generating almost unique \(Q_i^\pi(s, \ba)\) values. 
\end{remark}}

\begin{remark}
If we have more than two players that satisfy the condition of \Cref{claim:independent}, i.e., the MG is non-degenerate with respect to more than two players (instead of being non-degenerate with respect to all players), then we can still guarantee the decomposability of $\PP(s' \given s, \cdot)$ with respect to a set that is not necessarily the same as  $\cN_C$. As a byproduct of the decomposability of $\PP(s'\given s, \cdot)$, we can guarantee the decomposability of $r_i(s, a_i, \cdot)$, too. 
\end{remark}

\numbering{\subsection{Deferred proof of the claims in \Cref{ssec:alt_def}}}
\tnumbering{\section{Deferred proof of the claims in \Cref{ssec:alt_def}}}
\label{ssec:deferred-proof-alternative-definition}

 We define $\Pi:=  \begin{bmatrix}
\pi(s_1) & 0 & \cdots & 0 \\
0 & \pi(s_2) & \cdots & 0 \\
\vdots & \vdots & \ddots & 0 \\
0 & 0 & \cdots & \pi(s_{|\cS|})
\end{bmatrix}  \in \RR^{|\cS||\cA| \times |\cS|}$, where $\pi(s)$ is a column vector for the policy at state $s$, $\bP:= (\PP(s'\given s, \ba))_{((s, \ba), s')}$, $\bR := \begin{bmatrix}
\br(s_1)^\intercal & 0 & \cdots & 0 \\
0 & \br(s_2)^\intercal & \cdots & 0 \\
\vdots & \vdots & \ddots & 0 \\
0 & 0 & \cdots & \br(s_{|\cS|})^\intercal
\end{bmatrix}  \in \RR^{|\cS|\times |\cS||\cA|} $, then we have $\bV_i^\pi= (I - \gamma \Pi^\intercal \bP)^{-1} \bR \Pi \pmb{1} \in \RR^{|\cS|}$ where $\pmb{1} \in \RR^{|\cS|}$ is $(1, \dots, 1)^\intercal$. 

\begin{proof}[Proof of \Cref{claim:independent}]
If we differentiate $\bV_i^\pi$ with respect to $\pi$ along the direction $\Delta_{( s, \ba_1, \ba_2)}:= e_{\ba_1, s} - e_{\ba_2, s}{\in\RR^{|\cA||\cS|}}$ {for some $\ba_1,\ba_2\in\cA$},  i.e., the direction that increases (or decreases) $\pi(\ba_1\given s)$ and decreases (or increases) $\pi(\ba_2 \given s)$, respectively, then by computation, we can have the  following directional derivative:  
\begin{align}
     &\nabla_{\Delta_{( s, \ba_1, \ba_2)}} \bV_i^\pi\nonumber
      \\
      &=(I - \gamma \Pi^\intercal \bP)^{-1} \left(\sum_{s' \in \cS} \left(\PP(s'\mid s, \ba_1) V_i^\pi(s') - \PP(s'\given s, \ba_2) V_i^\pi(s')\right)+ r_i(s, \ba_1) - r_i(s, \ba_2)\right) e_s \nonumber
    \\ 
    &= \left(Q_i^\pi(s, \ba_1) - Q_i^\pi (s, \ba_2) \right) (I - \gamma \Pi^\intercal \bP)^{-1} e_s. \label{eqn:derivateive}
\end{align}
Therefore, if there exists a policy $\pi$ such that  for every $s$, there exist $\ba_{s,1}, \ba_{s,2}$ that make $Q_i^\pi(s, \ba_{s, 1}) \neq Q_i^\pi(s, \ba_{s, 2})$, then deviating from $\pi$ along the $\Delta_{(s, \ba_{s,1}, \ba_{s,2})}$  direction for every $s$ provides $\{\bV_i^{\pi^{(k)}} - \bV_i^{\pi}\}_{k \in [|\cS|]}$ vectors that are independent of each other,  since $(I - \gamma \Pi^\intercal \bP)^{-1}$ is an invertible matrix for any $\Pi$.  
\end{proof}
\begin{proof}[Proof of \Cref{claim:univ-s-for-constant}]  
Note that $\pi \in \Delta(\cA)^{|\cS|}$. For all $s \in \cS$, define $\Pi_s$ as 
\begin{align*}
    \Pi_s:= \{ \pi \given  \text{$Q^\pi_i(s, \ba)$ is a  constant function of  $\ba$ } \}, 
\end{align*}
where we omit the dependence on $i$ as we focus on the discussion on a specific $i$ here. 
For any $\pi$, there exists a state $s_\pi$ such that $Q_i^\pi(s_\pi, \ba)$ is a constant function of $\ba$, which means that $\pi \in \Pi_{s_\pi}$ and thus $\Pi_{s_\pi}\neq \emptyset$, which further yields $\sum_{s \in \cS} \left(\text{measure of }(\Pi_{s})  \right) \geq \left(\text{measure of the whole space of }(\Delta(\cA)^{|\cS|})\right) > 0$. By  the pigeonhole principle, we know that there exists some $s \in \cS$ such that $\left(\text{measure of }(\Pi_s)\right) > 0$ since $\cS$ is a finite set. 

Note that {for the above $s$ such that $\left(\text{measure of }(\Pi_s)\right) > 0$}, for any pair of $\ba_1, \ba_2 \in \cA$,  $Q^\pi_i(s, \ba_1) - Q^\pi_i(s, \ba_2)$  can  be represented by the ratio of polynomials of $\pi$ while it has a non-zero measure set of solution, we can conclude that  $Q^\pi_i(s, \ba_1) - Q^\pi_i(s, \ba_2) = 0$ for every $\pi \in \Delta(\cA)^{|\cS|}$, $\ba_1, \ba_2 \in \cA$ {for $s$ such that $\left(\text{measure of }(\Pi_s)\right) > 0$}. Therefore, {for the $s$ such that $\left(\text{measure of }(\Pi_s)\right) > 0$},  we have that for every $\pi$, $Q^\pi_i(s, \ba)$ is a constant function of $\ba$. 
\end{proof}

\numbering{\subsection{Counterexample for Alternative \Cref{def:alternative_MNMG}}}
\tnumbering{\section{Counterexample for Alternative \Cref{def:alternative_MNMG}}}
\label{ssec:counterexample_alternative}

We now show via a counterexample that the alternative definition given in \Cref{def:alternative_MNMG} may not even preserve the networked separable structure of the reward functions in general. 

Consider a Markov game with the following specifications: $\cN = [3]$, $\gamma = \frac{1}{2}$, $\cS = \{s_1,s_2,s_3\}$, $\cA_i = \{0, 1\}$, $\PP(s_3\given s_3, \ba) = \PP(s_2 \given s_2, \ba) = 1$ for any $\ba\in\cA$, $\PP(s_2 \given s_1, \ba) = \frac{1}{2} + \frac{1}{2} \cdot (-1)^{a_1 a_2 a_3}$, $\PP(s_3 \given s_1, \ba) = \frac{1}{2} - \frac{1}{2} \cdot (-1)^{a_1 a_2 a_3}$, $r_1(s_3, \ba) = 1$ and  $r_1(s_2, \ba) = 2$ for any $\ba\in\cA$, and  $r_1(s_1, \ba) =  \frac{1}{2} - \frac{1}{2}\cdot (-1)^{a_1a_2a_3}$. 
Then, we have by definition  that for any $\ba\in\cA$,  $Q_1^\pi(s_2, \ba) =V_1^\pi(s_2) =  \frac{1}{1-\gamma} r_1(s_2,\ba)  = 4$, $Q_1^\pi(s_3, \ba)=V_1^\pi(s_3) = \frac{1}{1-\gamma} r_1(s_3, \ba) = 2$, and  
\begin{align*}
 Q_1^\pi(s_1, \ba) &= r_1(s_1, \ba) + \gamma \PP(s_3 \given s_1, \ba) V^\pi_1(s_3) + \gamma \PP(s_2 \given s_1, \ba) V^\pi_1(s_2)
 \\
 & = r_1(s_1, \ba) + \frac{1}{2} \cdot (\frac{1}{2} - \frac{1}{2} \cdot (-1)^{a_1 a_2 a_3}) \cdot 2 + \frac{1}{2} \cdot (\frac{1}{2} + \frac{1}{2}\cdot  (-1)^{a_1 a_2 a_3}) \cdot 4
 = 2. 
\end{align*}
Hence, $Q^\pi_1(s, a_1, \cdot)$ is decomposable with respect to \{2, 3\}, however, the reward $r_1(s,a_1,\cdot)$ is not, due to $r_1(s_1,\ba)$. Note that this counterexample exactly exhibits the importance of the ergodicity of the Markov chain in removing the degeneracy.

\numbering{\subsection{Reviewing existing results for zero-sum NGs}}
\tnumbering{\section{Reviewing existing results for zero-sum NGs}}
In a zero-sum NG, i.e., a zero-sum polymatrix game, for each player $i\in\cN$, the reward $r_i\left(a_i, \pi_{-i}\right)$ of  using a pure strategy $a_i \in\cA_i$ is a linear function of $\pi_{-i}$. The following linear program, which involves variables $\pi \in \prod_{i \in \cN} \Delta(\cA_i)$ and $\bv = \left(v_i\right)_{i \in \cN}$, aims to minimize the sum of the variables $v_i$: 
$$
\begin{array}{ll}
\min_{\bv, \pi} & \sum_{i \in \cN} v_i \\
\text {subject to } & v_i \geq r_i\left(e_{a_i}, \pi_{-i}\right), \quad \text { for all } i \in \cN, a \in \cA_i, \\
& \pi \in \prod_{i \in \cN} \Delta(\cA_i).
\end{array}
$$ 
Reference \cite{cai2016zero} states that if $(\pi^\star, \bv^\star)$ is an optimal solution to the above linear program, then $\pi^\star$  is an NE of the zero-sum NG, {and the optimal value of the above linear program is 0.}  Conversely if $\pi^\star$ is an NE, then there exists a $\bv^\star$ such that $(\pi^\star, \bv^\star)$ is an optimal solution to the above linear program and $\bv^\star$ is the expected reward vector under $\pi^\star$. {By observation, we additionally have the following proposition as an extension:
\begin{proposition}
\label{prop:optima-and-value}
If $(\pi^\star, \bv^\star)$ is an $\epsilon$-optimal solution to the above linear program, then $\pi^\star$ is an $\epsilon$-NE and $r_i(\pi^\star) \leq v_i^\star \leq r_i(\pi^\star) + \epsilon$ {for all $i\in\cN$}. Conversely, if $\pi^\star$ is an $\epsilon$-NE, then there exists a $\bv^\star$ such that $(\pi^\star, \bv^\star)$ is an $n\epsilon$-optimal solution to the above linear program.     
\end{proposition}
\begin{proof}
If $(\pi^\star, \bv^\star)$ is an $\epsilon$-optimal solution to the above linear program (whose optimal solution is exactly $0$), we have   
\begin{align*}
  \epsilon \geq \sum_{i \in \cN} v_i^\star \underset{(i)}{=} \sum_{i \in \cN} \left(v_i^\star - r_i(\pi^\star)\right)
  \underset{(ii)}{\geq}  \sum_{i \in \cN} \left(\max_{\mu_i \in \Delta(\cA_i)} r_i\left(\mu_i, \pi_{-i}^\star\right) -  r_i\left(\pi^\star\right) \right)
\end{align*}
which proves that $\pi^\star$ is an $\epsilon$-NE: here $(i)$ holds since the sum of reward over players is zero, and $(ii)$ holds due to the constraint of the given linear program. Moreover, we can also observe that $r_i(\pi^\star) \leq v_i^\star \leq r_i(\pi^\star) + \epsilon$ holds, since $ 0 \leq \max_{\mu_i \in \Delta(\cA_i)} r_i\left(\mu_i, \pi_{-i}^\star\right) -  r_i\left(\pi^\star\right)\leq v_i^\star - r_i(\pi^\star) \leq \epsilon$ for each $i\in\cN$.    

Conversely, suppose that $\pi^\star$ is an $\epsilon$-NE. Then, defining $v^\star_i: = \max_{\mu_i \in \Delta(\cA_i)} r_i\left(\mu_i, \pi_{-i}^\star\right)$ satisfies the constraints of the given linear program. In addition, we have  
\begin{align*}
    \sum_{i \in \cN} v_i^\star = \sum_{i \in \cN} \left(v_i^\star - r_i(\pi^\star)\right) =\sum_{i \in \cN} \left(\max_{\mu_i \in \Delta(\cA_i)} r_i\left(\mu_i, \pi_{-i}^\star\right) - r_i(\pi^\star)\right) \leq n\epsilon,
\end{align*}
which concludes the theorem.
\end{proof}

\begin{proposition}
\label{prop:polymatrix}
Suppose $\pi^\star$ is an $\epsilon$-approximate CCE of the zero-sum NG. Then {the product of its} marginalized policy $\hat{\pi}^\star$ is an $n\epsilon$-approximate NE of the zero-sum NG. 
Moreover, it holds that $r_i(\pi^\star) \geq r_i(\hat{\pi}^\star) \geq r_i(\pi^\star) - n\epsilon$ for every $i \in \cN$. 
\end{proposition}

\begin{proof}
A similar method with \cite{cai2016zero}'s Theorem 2 can provide proof of \Cref{prop:polymatrix}. To be specific, define $v^\star_i:= r_i(\pi^\star)$, then we have $(\hat{\pi}^\star, \bv^\star)$ is an $n\epsilon$-optimal solution. By \Cref{prop:optima-and-value}, we can conclude that $\hat{\pi}^\star$ is an $n\epsilon$-approximate NE and $r_i(\pi^\star) \geq r_i(\hat{\pi}^\star) \geq r_i(\pi^\star) - n\epsilon$ for all $i \in \cN$.
\end{proof}

We note that \Cref{prop:optima-and-value} and \Cref{prop:polymatrix} are not in \cite{cai2016zero}, but it plays an important role in proving  \Cref{prop:MZNMG}. 
} 

\numbering{\subsection{Omitted proof of \Cref{prop:MZNMG}}}
\tnumbering{\section{Omitted proof of \Cref{prop:MZNMG}}}
Before introducing the proof of \Cref{prop:MZNMG}, we provide the relationship between approximate Markov stationary CCE and approximate auxiliary-game CCE in Markov games.  

\begin{claim}
\label{claim:inf-stage}
For an infinite-horizon $\gamma$-discounted MG, an $\epsilon$-approximate Markov stationary CCE $\pi$ of this MG  {makes $\pi(s)$} 
an $\epsilon$-approximate CCE {of the auxiliary game at each state $s\in\cS$},  where the auxiliary game payoff matrix {at each $s\in\cS$} is defined {as} $(r_i(s, \ba) + \gamma \langle  \PP(\cdot |s, \ba), V_i^\pi(\cdot) \rangle)_{\ba \in \cA}$  for player $i \in \cN$.    
\end{claim}
\begin{proof}
By definition of $\epsilon$-approximate Markov (perfect) CCE, we have for all $s\in\cS$ and all $i\in\cN$ that
\#\label{equ:tmp_1}
V_{i}^{\pi}(s)\leq \max_{\mu_i \in \Delta(\cA_i){^{|\cS|}}} V_{i}^{\mu_i, \pi_{-i}}(s) \leq V_{i}^\pi(s) +\epsilon.
\#
Then, by one-step of Bellman equation for $\max_{\mu_i \in \Delta(\cA_i){^{|\cS|}}} V_{i}^{\mu_i, \pi_{-i}}(s)$, we also know that
\#\label{equ:tmp_2}
\max_{\mu_i \in \Delta(\cA_i){^{|\cS|}}} V_{i}^{\mu_i, \pi_{-i}}(s)=\max_{\nu\in\Delta(\cA_i)}\EE_{\nu,\pi_{-i}(s)}\bigg[r_i(s,\ba)+\gamma\la\PP(\cdot\given s,\ba),~\max_{\mu_i \in \Delta(\cA_i){^{|\cS|}}}V_i^{\mu_i,\pi_{-i}}(\cdot)\ra\bigg].
\# 
Moreover, we have   by the left inequality of \Cref{equ:tmp_1} that 
\arxiv{\begin{align}
\max_{\nu\in\Delta(\cA_i)}&\EE_{\nu,\pi_{-i}(s)}\bigg[r_i(s,\ba)+\gamma\la\PP(\cdot\given s,\ba),~\max_{\mu_i \in \Delta(\cA_i){^{|\cS|}}}V_i^{\mu_i,\pi_{-i}}(\cdot)\ra\bigg]
\nonumber \\&\geq \max_{\nu\in\Delta(\cA_i)}\EE_{\nu,\pi_{-i}(s)}\bigg[r_i(s,\ba)+\gamma\la\PP(\cdot\given s,\ba),~V_i^{\pi}(\cdot)\ra\bigg].    \label{equ:tmp_3}
\end{align}}
\neurips{\begin{align}
\label{equ:tmp_3}
\max_{\nu\in\Delta(\cA_i)}&\EE_{\nu,\pi_{-i}(s)}\bigg[r_i(s,\ba)+\gamma\la\PP(\cdot\given s,\ba),~\max_{\mu_i \in \Delta(\cA_i){^{|\cS|}}}V_i^{\mu_i,\pi_{-i}}(\cdot)\ra\bigg] \nonumber
\\&\geq \max_{\nu\in\Delta(\cA_i)}\EE_{\nu,\pi_{-i}(s)}\bigg[r_i(s,\ba)+\gamma\la\PP(\cdot\given s,\ba),~V_i^{\pi}(\cdot)\ra\bigg].    
\end{align}}
Also, we have by one-step Bellman consistency equation for $V_{i}^{\pi}(s)$ that
\#\label{equ:tmp_4}
V_{i}^{\pi}(s)=\EE_{\pi(s)}\bigg[r_i(s,\ba)+\gamma\la\PP(\cdot\given s,\ba),~V_i^{\pi}(\cdot)\ra\bigg].
\#
Combining \Cref{equ:tmp_1,equ:tmp_3,equ:tmp_4}, we have
\$
\max_{\nu\in\Delta(\cA_i)}\EE_{\nu,\pi_{-i}(s)}\bigg[r_i(s,\ba)+\gamma\la\PP(\cdot\given s,\ba),~V_i^{\pi}(\cdot)\ra \bigg]\leq \EE_{\pi(s)}\bigg[r_i(s,\ba)+\gamma\la\PP(\cdot\given s,\ba),~V_i^{\pi}(\cdot)\ra\bigg]+\epsilon,
\$
for all $i\in\cN$ and $s\in\cS$, which proves that $\pi(s)$ is an $\epsilon$-approximate  CCE of the auxiliary game, where the game payoff matrix at each state $s\in\cS$ is $(r_i(s, \ba) + \gamma \langle  \PP(\cdot |s, \ba), V_i^\pi(\cdot) \rangle)_{\ba \in \cA}$  for player $i \in \cN$. 
\end{proof}

\begin{claim}
\label{claim:fin-stage}
For an $H$-horizon MG and $h \in [H]$, an $\epsilon$-approximate Markov CCE {$\pi=\{\pi_h\}_{h\in[H]}$}  of this MG {makes $\pi_h(s)$} an $\epsilon$-approximate  CCE of the {auxiliary game at each state $s\in\cS$ and each $h\in[H]$,} where the  auxiliary game payoff matrix {at $(s,h)$} is defined as  $(r_{h, i} (s, \ba) + \langle  \PP_h(\cdot |s, \ba), V_{h+1, i}^\pi(\cdot) \rangle)_{\ba \in \cA}$  for player $i \in \cN$, where $V_{H+1, i}^{\tilde\pi}(s) = 0$ for any policy $\tilde \pi$, and for all $s \in \cS$ and $i \in \cN$. 
\end{claim}
\begin{proof}
By definition of $\epsilon$-approximate Markov (perfect) CCE, we have that for all $s\in\cS$, $h\in[H]$, and $i\in\cN$
\#\label{equ:tmp_1_finite}
V_{h, i}^{\pi}(s)\leq \max_{\mu_i \in \Delta(\cA_i){^{|\cS| \times H}}} V_{h, i}^{\mu_i, \pi_{-i}}(s) \leq V_{h, i}^\pi(s) +\epsilon.
\#
Then, by one-step of Bellman equation for $\max_{\mu_i \in \Delta(\cA_i){^{|\cS| \times H}}} V_{h, i}^{\mu_i, \pi_{-i}}(s)$, we also know that
\arxiv{
\#\label{equ:tmp_2_finite}
\max_{\mu_i \in \Delta(\cA_i){^{|\cS| \times H}}} V_{h, i}^{\mu_i, \pi_{-i}}(s)=\max_{\nu\in\Delta(\cA_i)}\EE_{\nu,\pi_{{h,-i}}(s)}\bigg[r_{h, i}(s,\ba)+\la\PP_h(\cdot\given s,\ba),~\max_{\mu_i \in \Delta(\cA_i){^{|\cS| \times H}}}V_{h+1, i}^{\mu_i,\pi_{-i}}(\cdot)\ra\bigg].
\#     
}
\neurips{
\begin{align}
&\max_{\mu_i \in \Delta(\cA_i){^{|\cS| \times H}}} V_{h, i}^{\mu_i, \pi_{-i}}(s) \nonumber
\\
&\qquad=\max_{\nu\in\Delta(\cA_i)}\EE_{\nu,\pi_{{h,-i}}(s)}\bigg[r_{h, i}(s,\ba)+\la\PP_h(\cdot\given s,\ba),~\max_{\mu_i \in \Delta(\cA_i){^{|\cS| \times H}}}V_{h+1, i}^{\mu_i,\pi_{-i}}(\cdot)\ra\bigg].\label{equ:tmp_2_finite}
\end{align}
}
Moreover, we have   by the left inequality of \Cref{equ:tmp_1_finite} and $V_{H+1, i}^{\tilde\pi}(s) = 0$ for any $\tilde \pi$ that  
\#\label{equ:tmp_3_finite}
&\max_{\nu\in\Delta(\cA_i)}\EE_{\nu,\pi_{{h,-i}}(s)}\bigg[r_{h, i}(s,\ba)+\la\PP_h(\cdot\given s,\ba),~\max_{\mu_i \in \Delta(\cA_i){^{|\cS| \times (H-1)}}}V_{h+1, i}^{\mu_i,\pi_{-i}}(\cdot)\ra\bigg]
\\&\qquad \geq \max_{\nu\in\Delta(\cA_i)}\EE_{\nu,\pi_{{h,-i}}(s)}\bigg[r_{h, i}(s,\ba)+\la\PP_h(\cdot\given s,\ba),~V_{h+1,i}^{\pi}(\cdot)\ra\bigg]. \nonumber
\#
Also, we have  by one-step Bellman equation for $V_{h, i}^{\pi}(s)$ that
\#\label{equ:tmp_4_finite}
V_{h, i}^{\pi}(s)=\EE_{\pi_h(s)}\bigg[r_{h, i}(s,\ba)+\la\PP_h(\cdot\given s,\ba),~V_{h+1,i}^{\pi}(\cdot)\ra\bigg].
\#
Combining \Cref{equ:tmp_1_finite,equ:tmp_2_finite,equ:tmp_4_finite}, we can conclude the theorem.
\end{proof}

\CCENE*
\begin{proof} 
For {an} arbitrary joint Markov  policy $\mu$, we define  
$d_{s' ,\mu}(s):=(1-\gamma) \EE_{s_1 = s', \mu}[\sum_{t=1}^\infty\gamma^{t-1}\pmb{1}(s_t =s)]$ which is the {discounted visitation measure} of states 
when we follow policy $\mu$ and start from state $s'$. Here, $s_t$ denotes the state at timestep $t$. 
We use $\pi$ to denote an $\epsilon$-approximate Markov CCE, and $\hat{\pi}$ to denote {the product policy of the per-state marginalized policies  of $\pi$ for all agents $i\in\cN$}. We define $\pi_i$ as the per-state marginalized policy for player $i$ from $\pi$. With this notation, we have $\hat{\pi}:=\pi_1\times \cdots\times\pi_n$. {Note that for zero-sum NMGs, by \Cref{prop:MPGcond} we know that if $\cN_C\neq \emptyset$, then}  $\PP(s' \mid s, \ba)= \sum_{i\in\cN_C} \FF_i(s'\mid s, a_i) $ or $\PP_h(s' \mid s, \ba)= \sum_{i\in\cN_C} \FF_{h, i}(s'\mid s, a_i) $, for infinite- and finite-horizon cases, respectively; and if $\cN_C = \emptyset$, we have $\PP(s'\given s, \ba) = \FF_o(s'\given s)$ or $\PP_h(s'\given s, \ba) = \FF_{h,o}(s'\given s)$ for the two cases. 

{In the infinite-horizon case, by \Cref{claim:inf-stage}, we know that  $\pi(s)$ also serves as an $\epsilon$-approximate CCE for an auxiliary game with a payoff matrix   $(r_i(s, \ba) + \gamma \langle \PP(\cdot \given s, \ba), V_i^\pi(\cdot) \rangle)_{\ba \in \cA}$  for player $i \in \cN$. Since this auxiliary game is a zero-sum NG by the definition of zero-sum NMG, \Cref{prop:polymatrix} implies that the policy $\hat{\pi}{(s)}$ is an $n \epsilon$-approximate NE of the auxiliary game with the same payoff matrix, and the following inequality is valid for all $i \in \cN$ and $s \in \cS$:
\begin{align}
    V_i^\pi(s) &= r_i(s, \pi) + \gamma \sum_{s' \in \cS} \sum_{\ba \in \cA}  \PP(s' \given s, \ba) \pi(\ba \given  s)V_{i}^\pi(s')   \nonumber
    \\
    &\leq  r_i(s, \hat{\pi}) + \gamma \sum_{s' \in \cS} \sum_{\ba \in \cA}  \PP(s' \given s, \ba) \hat{\pi}(\ba \given s)V_{i}^\pi(s') + n\epsilon. \label{eqn:V-iter}
\end{align}
Applying the inequality $V_i^\pi(s') \leq r_i(s', \hat{\pi}) + \gamma \sum_{\tilde{s} \in \cS} \sum_{\ba \in \cA} \PP(\tilde{s} \given s', \ba) \hat{\pi}(\ba \given s')V_{i}^\pi(\tilde{s}) + n\epsilon$ into the final expression of \Cref{eqn:V-iter}, and applying it  recursively, we have that for every $i \in \cN$ and $s \in \cS$:
\begin{align}
    V_i^\pi(s) &\leq  V_i^{\hat{\pi}}(s) + n\epsilon/(1-\gamma). \label{eqn:CCENE-ineq-inf}
\end{align}}

Moreover, we have that {for any $\mu\in \Delta(\cA_i)^{|\cS|}$}
\begin{align*}
     V^{\hat{\pi}}_i(s) \geq V^\pi_i(s) -n\epsilon/(1-\gamma) &\geq V^{\mu, \pi_{-i}}_i(s) -(n+1)\epsilon/(1-\gamma)
     \\
     &= \EE_{\ba \sim {\mu(s')\times\pi_{-i}(s')}, s' \sim d_{s,\mu, \pi_{-i}}}[r_i(s', \ba)] - (n+1)\epsilon{/(1-\gamma)}
     \\
     &\underset{(i)}{=} \EE_{\ba \sim {\mu(s')\times\hat{\pi}_{-i}(s')}, s' \sim d_{s,\mu, {\pi}_{-i}}}[r_i(s', \ba)] - (n+1)\epsilon/(1-\gamma)
     \\
     &\underset{(ii)}{=} \EE_{\ba \sim {\mu(s')\times\hat{\pi}_{-i}(s')}, s' \sim d_{s,\mu, \hat{\pi}_{-i}}}[r_i(s', \ba)] - (n+1)\epsilon/(1-\gamma)
     \\
     &= V^{\mu, \hat{\pi}_{-i}}_i(s) - (n+1)\epsilon/(1-\gamma),
\end{align*}
{where the second inequality follows from  $\pi$ being an $\epsilon$-approximate Markov CCE, and } $(i)$ holds since for arbitrary $(\nu_i, \nu_{-i})  \in \Delta(\cA_i)^{|\cS|} \times \Delta(\cA_{-i})^{|\cS|}$ and for any $\nu_\cS \in \Delta(\cS)$, 
\$
&\EE_{\ba \sim \nu_i(s')\times \nu_{-i}(s') , s' \sim \nu_\cS } [r_i(s', \ba)]\\
& =\sum_{\ba \in \cA} \sum_{s' \in \cS} \nu_\cS(s') r_i(s', \ba) \nu_i(a_i \given s') \nu_{-i} (\ba_{-i} \given s') \\
&= \sum_{j \in \cE_{Q, i}} \sum_{\ba \in \cA} \sum_{s' \in \cS} \nu_\cS(s') r_{i,j}(s', a_i, a_j) \nu_i(a_i \given s') \nu_{-i} (\ba_{-i} \given s') 
\\
&= \sum_{j \in \cE_{Q, i}} {\sum_{a_i \in \cA_i}\sum_{a_j \in \cA_j}}\sum_{s' \in \cS} \nu_\cS(s') r_{i,j}(s', a_i, a_j) \nu_i(a_i \given s') \nu_{j} (a_j \given s') 
 \\&=\EE_{\ba \sim \nu_i(s') \times \hat{\nu}_{-i}(s') , s' \sim \nu_\cS } [r_i(s', \ba)], 
 \$
 where  $\hat{\nu}:= \nu_1 \times \dots \times \nu_n$ is the product policy of the per-state marginalized policies of $\nu$, and $(ii)$ holds due to the following fact: if $\cN_C\neq \emptyset$
\begin{align}
\PP_\pi(s'\given s)&:=\sum_{\ba\in\cA} \PP(s'\given s,\ba)\pi(\ba\given s)=\sum_{\ba\in\cA}\sum_{i\in\cN_C} \FF_i(s'\given s,a_i)\pi(\ba\given s) \nonumber\\
 &=\sum_{i\in\cN_C} \sum_{\ba\in\cA}\FF_i(s'\given s,a_i)\pi(\ba\given s)=\sum_{i\in\cN_C} \sum_{a_i\in\cA_i}\FF_i(s'\given s,a_i)\pi_i(a_i\given s)=:\PP_{\hat{\pi}}(s'\given s), \label{eqn:marginalizing}   
\end{align}
or if $\cN_C=\emptyset$
\begin{align}
\PP_\pi(s'\given s) =\FF_o(s'\given s) = \PP_{\hat{\pi}}(s'\given s).
\end{align}
In other words, the marginalized policy's state visitation measure $d_{s, {\hat{\pi}}}$ is the same as the original {policy's} state visitation measure $d_{s,{\pi}}$. Therefore, marginalizing $\epsilon$-approximate Markov CCE provides $(n+1)\epsilon/(1-\gamma)$-approximate Markov NE.  

{Moreover, a similar argument holds for the finite-horizon episodic setting. 
In the $H$-horizon case, since $\pi$ is an $\epsilon$-approximate Markov CCE,  by \Cref{claim:fin-stage} $\pi_h(s)$ also serves as an $\epsilon$-approximate CCE for an auxiliary game with a payoff matrix defined by $(r_{h, i}(s, \ba) +  \langle \PP_h(\cdot \given s, \ba), V_{h+1, i}^\pi(\cdot) \rangle)_{\ba \in \cA}$ for player $i \in \cN$. Since this auxiliary game is a zero-sum NG by the definition of zero-sum NMG, \Cref{prop:polymatrix} implies that the policy $\hat{\pi}_h(s)$ is an $n \epsilon$-approximate NE of the auxiliary-game with the same payoff matrix, and the following inequality is valid for all $i \in \cN$ and $s \in \cS$: 
\begin{align}
    V_{h, i}^\pi(s) & = r_{h, i}(s, \pi) + \sum_{s' \in \cS} \sum_{\ba \in \cA}  \PP_h(s' \given s, \ba) \pi{_h}(\ba \given s)V_{h+1, i}^\pi(s')  \notag\\
    & \leq  r_{h, i}(s, \hat{\pi}) + \sum_{s' \in \cS} \sum_{\ba \in \cA}  \PP_h(s' \given s, \ba) \hat{\pi}{_h}(\ba \given  s)V_{h+1, i}^\pi(s') + n\epsilon. \label{eqn:V-iter-fin}
\end{align}
Applying the inequality 
\begin{align*}
    V_{h+1, i}^\pi(s') \leq r_{h+1, i}(s', \hat{\pi}) +  \sum_{\tilde{s} \in \cS} \sum_{\ba \in \cA} \PP_{h+1}(\tilde{s} \given s', \ba) \hat{\pi}{_{h+1}}(\ba \given s')V_{h+2, i}^\pi(\tilde{s}) + n\epsilon
\end{align*}  into  \eqref{eqn:V-iter-fin} continually, iterating this procedure from $h+1$ to $H$, yields that for every $i \in \cN$ and $s \in \cS$:
\begin{align}
    V_{h, i}^\pi(s) &\leq  V_{h, i}^{\hat{\pi}}(s) + hn\epsilon. \nonumber
\end{align}
Moreover, we have that {for any $\mu\in \Delta(\cA_i)^{|\cS| \times H}$}
\begin{align*}
     V^{\hat{\pi}}_{h, i}(s) \geq V^\pi_{h, i}(s) -n\epsilon H &\geq V^{\mu, \pi_{-i}}_{h, i}(s) -(n+1)\epsilon H = V^{\mu, \hat{\pi}_{-i}}_{h, i}(s) - (n+1)\epsilon h,
\end{align*}
with a similar observation on the visitation measure under $\pi$ and $\hat{\pi}$, which concludes the proof.}  
\end{proof}

\numbering{\section{Omitted Details  in \Cref{subsection:PPAD}}}
\tnumbering{\chapter{Omitted Details  in \Cref{subsection:PPAD}}}
\label{appendix:PPADhard-proof}
\PPADHardmain*
\begin{proof} 
We separate the proof for the two cases as follows. 

\paragraph{Case 1. $\cE_Q$ contains a triangle subgraph.}
We will show that for \emph{any} general-sum two-player \emph{turn-based} MG \textbf{(A)}, the problem of computing its Markov stationary CCE, which is inherently a {\tt PPAD}-hard problem \cite{daskalakis2022complexity}, can be reduced to computing the Markov stationary CCE of a three-player zero-sum MG with a triangle structure networked separable interactions  \textbf{(B)}. Consider an MG \textbf{(A)} with two players, players 1 and 2, and a reward function $r_1(s, a_1, a_2)$ and $r_2(s, a_2, a_1)$, where $a_i$ is the action of the $i$-th player and $r_i$ is the reward function of the $i$-th player. The transition dynamics is given by $\PP(s'\,|\,s, a_1, a_2)$. In even {rounds}, player 2's action space is limited to \textsf{Noop2}, and in odd {rounds}, player 1's action space is limited to \textsf{Noop1}, where \textsf{Noop} is an abbreviation of {``no-operation'', i.e., the player does not affect the transition dynamics {nor reward functions} in that round}. We denote player 1's action space in even rounds as $\cA_{1, \text{even}}$ and player 2's action space in odd {rounds} as $\cA_{2, \text{odd}}$.

Now, we construct a {three-player}  zero-sum NMG. We set the reward function as $\tilde{r}_i(s, \ba) = \sum_{j \neq i}\tilde{r}_{i,j}(s, a_i, a_j) $ and $\tilde{r}_{i,j}(s,a_i,a_j)= -\tilde{r}_{j,i}(s,a_j,a_i)$, where the reward functions are designed so that $\tilde{r}_{i,j} = -\tilde{r}_{j,i}$ for all $i, j$, $\tilde{r}_{1,2} + \tilde{r}_{1,3} = r_1$, and $\tilde{r}_{2,1} + \tilde{r}_{2,3} = r_2$, {by introducing} a dummy player{, player 3}. {Here $r_1,r_2$ are the reward functions in game \textbf{(A)}}. In even rounds, player 2's action space is limited to \textsf{Noop2}, and in odd rounds, player 1's action space is limited to \textsf{Noop1}. Player 3's action space is always limited to \textsf{Noop3} in all rounds. The transition dynamics is defined as $\tilde\PP(s'\,|\,s, a_1, a_2, a_3) = \PP(s'\,|\,s, a_1, a_2)$,  since $a_3$ is always \textsf{Noop3}. {In other words, player 3's action does not affect the rewards of the other two players, nor the transition dynamics, and players 1 and 2 will receive the reward as in the two-player turn-based MG. Also, }note that due to the turn-based structure of the game \textbf{(A)}, the transition dynamics satisfy the decomposable condition in our \Cref{prop:MPGcond}, and it is thus a zero-sum NMG. {In fact, every turn-based dynamics can be represented as an ensemble of single-controller dynamics, as we have discussed in \Cref{sec:definition}.}   
We set the reward function values as follows:
\begin{align*}
    &\tilde{r}_{1,3}(s, a_1, \textsf{Noop3}) =-\tilde{r}_{3,1}(s, \textsf{Noop3},a_1) = r_1(s, a_1, \textsf{Noop2}) + r_2(s, \textsf{Noop2}, a_1)
    \\
    &\tilde{r}_{1,3}(s, \textsf{Noop1}, \textsf{Noop3}) = -\tilde{r}_{3,1}(s, \textsf{Noop3}, \textsf{Noop1}) =0 
    \\
    &\tilde{r}_{2,3}(s, a_2, \textsf{Noop3}) = -\tilde{r}_{3,2}(s, \textsf{Noop3}, a_2) = r_1(s, \textsf{Noop1}, a_2) + r_2(s, a_2, \textsf{Noop1})
    \\
    &\tilde{r}_{2,3}(s, \textsf{Noop2}, \textsf{Noop3}) = -\tilde{r}_{3,2}(s, \textsf{Noop3}, \textsf{Noop2}) =0 
    \\
    &\tilde{r}_{1,2}(s, a_1, \textsf{Noop2}) = -\tilde{r}_{2,1}(s, \textsf{Noop2}, a_1) = -r_2(s, \textsf{Noop2}, a_1)
    \\
    &\tilde{r}_{1,2}(s, \textsf{Noop1}, a_2) = -\tilde{r}_{2,1}(s, a_2, \textsf{Noop1}) = r_1(s, \textsf{Noop1}, a_2).
\end{align*}

Note that the new game \textbf{(B)} is still a turn-based game, and thus the Markov stationary CCE is the same as the Markov stationary NE. Also, note that by construction, we know that the equilibrium policies of players $1$ and $2$ at the   Markov stationary CCE of the game \textbf{(B)} constitute a Markov stationary CCE of the game \textbf{(A)}. {If the underlying network is more general and contains a triangle subgraph, we can specify the reward and transition dynamics of these three players as above, and specify all other players to be dummy players, whose reward functions are all zero, and do not affect the reward functions of these three players, nor the transition dynamics.}

\paragraph{Case 2. $\cE_Q$ contains a 3-path subgraph.}
We will show that for \emph{any} general-sum two-player \emph{turn-based} MG \textbf{(A)}, the problem of computing its Markov stationary CCE can also be reduced to computing the Markov stationary CCE of a four-player zero-sum MG with 3-path networked separable interactions  \textbf{(B)}. Consider an MG \textbf{(A)} with two players, players 1 and 2, and a reward function $r_1(s, a_1, a_2)$ and $r_2(s, a_2, a_1)$, where $a_i$ is the action of the $i$-th player and $r_i$ is the reward function of the $i$-th player. The transition dynamics is given by {$\PP(s'\,|\,s, a_1, a_2)$}. In even {rounds}, player 2's action space is limited to \textsf{Noop2}, and in odd {rounds}, player 1's action space is limited to \textsf{Noop1}, where \textsf{Noop} is an abbreviation of {``no-operation'', i.e., the player does not affect the transition dynamics {nor the  reward functions} in that round}. We denote player 1's action space in even rounds as $\cA_{1, \text{even}}$ and player 2's action space in odd {rounds} as $\cA_{2, \text{odd}}$.

Now, we construct a {four-player} zero-sum NMG with a 3-path network structure. We set the reward function as {$\tilde{r}_1(s, \ba) = \tilde{r}_{1,2}(s, a_1, a_2) +  \tilde{r}_{1,3}(s, a_1, a_3) $ and $\tilde{r}_2(s, \ba) = \tilde{r}_{2, 1}(s, a_2, a_1) + \tilde{r}_{2, 4}(s, a_2, a_4)$}  and $\tilde{r}_{i,j}(s,a_i,a_j)= -\tilde{r}_{j,i}(s,a_j,a_i)$. The reward functions  are designed so that 
$\tilde{r}_{1,2} + \tilde{r}_{1,3} = r_1$, and $\tilde{r}_{2,1} + \tilde{r}_{2,4} = r_2$, {where $r_1,r_2$ are the reward functions in game \textbf{(A)},} {by introducing} dummy players{, player 3 and player 4}. In even rounds, player 2's action space is limited to \textsf{Noop2}, and in odd rounds, player 1's action space is limited to \textsf{Noop1}. Player 3's action space is always limited to \textsf{Noop3} in all rounds. Player 4's action space is always limited to \textsf{Noop4} in all rounds. The transition dynamics is defined as {$\tilde{\PP}(s'\,|\,s, a_1, a_2, a_3, a_4) = \PP(s'\,|\,s, a_1, a_2)$}, since $a_3$ is always \textsf{Noop3} and $a_4$ is always \textsf{Noop4}. {In other words, player 3 and player 4's actions do not affect the rewards of the other two players, nor the transition dynamics, and players 1 and 2 will receive the reward as in the two-player turn-based MG. Also, }note that due to the turn-based structure of the game \textbf{(A)}, the transition dynamics satisfy the decomposable condition in our \Cref{prop:MPGcond}, and it is thus a zero-sum NMG. {In fact, every turn-based dynamics can be represented as an ensemble of single-controller dynamics, as we have discussed in \Cref{sec:definition}.}   
We set the reward function values as follows:
\begin{align*}
    &\tilde{r}_{1,3}(s, a_1, \textsf{Noop3}) =-\tilde{r}_{3,1}(s, \textsf{Noop3},a_1) = r_1(s, a_1, \textsf{Noop2}) + r_2(s, \textsf{Noop2}, a_1)
    \\
    &\tilde{r}_{1,3}(s, \textsf{Noop1}, \textsf{Noop3}) = -\tilde{r}_{3,1}(s, \textsf{Noop3}, \textsf{Noop1}) =0 
    \\
    &\tilde{r}_{2,4}(s, a_2, \textsf{Noop4}) = -\tilde{r}_{4,2}(s, \textsf{Noop4}, a_2) = r_1(s, \textsf{Noop1}, a_2) + r_2(s, a_2, \textsf{Noop1})
    \\
    &\tilde{r}_{2,4}(s, \textsf{Noop2}, \textsf{Noop4}) = -\tilde{r}_{4,2}(s, \textsf{Noop4}, \textsf{Noop2}) =0 
    \\
    &\tilde{r}_{1,2}(s, a_1, \textsf{Noop2}) = -\tilde{r}_{2,1}(s, \textsf{Noop2}, a_1) = -r_2(s, \textsf{Noop2}, a_1)
    \\
    &\tilde{r}_{1,2}(s, \textsf{Noop1}, a_2) = -\tilde{r}_{2,1}(s, a_2, \textsf{Noop1}) = r_1(s, \textsf{Noop1}, a_2).
\end{align*}

Note that the new game \textbf{(B)} is still a turn-based game, and thus the Markov stationary CCE is the same as the Markov stationary NE. Also, note that by construction, we know that the equilibrium policies of players $1$ and $2$ at the   Markov stationary CCE of the game \textbf{(B)} constitute a Markov stationary CCE of the game \textbf{(A)}. {If the underlying network is more general and contains a 3-path subgraph, we can specify the reward and transition dynamics of these {four players in the subgraph} as above, and specify all other players to be dummy players, whose reward functions are all zero, and do not affect the reward functions of these three players, nor the transition dynamics. This completes the proof.} 
\end{proof}

\begin{proposition}
\label{prop:star-shaped-only}
A connected graph that does not contain a subgraph of a triangle or a 3-path must be a star-shaped graph.
\end{proposition}
\begin{proof}
If the diameter of a connected graph is exactly 1, then there are only two nodes, which form a star-shaped network. If the diameter of a connected graph is greater than 2, it contradicts the non-existence of a 3-path subgraph. If the diameter of a connected graph is exactly 2, 
we denote the middle node as $c$, and the leftmost and rightmost nodes as $l$ and $r$. {If either $l$ or $r$ has another neighbor other than $c$, it implies the existence of a 3-path subgraph, which contradicts the assumption. Therefore, the additional nodes other than $l,c,r$, if exist, have to be connected to $c$. 
If two neighbors of $c$ are directly connected, then 
it contradicts the non-existence of a triangle  subgraph. Hence, all nodes except $c$ have to be connected to $c$ while not being connected to each other, which leads to a star-shaped graph.} 
\end{proof}

\numbering{\section{Omitted Details in \Cref{sec:fictitious-play}}}
\tnumbering{\chapter{Omitted Details in \Cref{sec:fictitious-play}}}
\label{appendix:fictitious}
We refer to \Cref{appendix:stochastic-approx} for the existing relevant result regarding stochastic approximation. The proof structure for \Cref{appendix:fictitious} follows three steps: (1) find the continuous{-time} dynamics of the fictitious-play learning dynamics, (2) identify a Lyapunov function for the continuous{-time} version of the fictitious  play ($V(\pi)$ or $L(\pi)$), and (3) since the discrete version can be viewed as a perturbed version of the continuous{-time} dynamics  (\Cref{thm:interpolated=perturbation}), the limit point of fictitious play is contained in the level set of a Lyapunov function (\Cref{thm:levelset}). \Cref{thm:interpolated=perturbation} and \Cref{thm:levelset} are stated in \Cref{appendix:stochastic-approx}, and these theorems are restatements of \cite{benaim2005stochastic}. {In this section, with a slight abuse of notation, we interchangeably use $a_i$ to refer to either an action in $\cA_i$, or a pure strategy $\pi_i \in \Delta(\cA_i)$, where $\pi_i(a_i) = 1$ and $\pi_i(a'_i) = 0$ for all $a'_i \neq a_i$.}

\numbering{\subsection{Matrix game case}}
\tnumbering{\section{Matrix game case}}
\label{ssec:Matrix-game-fictitious-play}

\numbering{\subsubsection{Fictitious-play in zero-sum NGs}}
\tnumbering{\subsection{Fictitious-play in zero-sum NGs}}

We first introduce the fictitious-play  dynamics for zero-sum NGs with $\cG = (\cN, \cE)$, i.e., zero-sum polymatrix games \cite{cai2011minmax,cai2016zero}, the very same one as in \cite{brown1951iterative,robinson1951iterative}: at iteration, $k$, each player $i$  maintains a belief of the opponents' policies, $(\hat{\pi}_{-i}^{(k)})$; 
she then takes action by best responding to the belief, and then updates the belief as: 
\$
\text{Take~action:~~}a_{i}^{(k)} \in \argmax_{a_i\in\cA_i} ~~r_i(e_{a_i},\hat{\pi}_{-i}^{(k)}),\quad~~ \text{Update belief:~~}\hat{\pi}_{-i}^{(k+1)} = \hat{\pi}_{-i}^{(k)}  + \alpha^{(k)} (e_{a_{-i}^{(k)}} - \hat{\pi}_{-i}^{(k)}) 
\$
where $r_i(\pi)$ is the expected payoff under joint policy $\pi$ (see \Cref{equ:expected_reward}), and $\alpha^{(k)}\geq 0$ is the stepsize. 
The overall procedure is summarized in \Cref{alg:polymatrix-game-fictitious}. 
\begin{algorithm}[H]
\caption{Fictitious Play in zero-sum NGs  ($i$-th player)}
\label{alg:polymatrix-game-fictitious}
\begin{algorithmic}
\STATE{Choose $\hat{\pi}_{j}^{(0)}$ as a uniform distribution for all $j \in \cN/\{i\}$}
\FOR{each timestep $k = 0, 1, \dots$}
\STATE {Take action $a_{i}^{(k)} \in \argmax_{a_i\in\cA_i} r_i(e_{a_i}, \hat{\pi}_{-i}^{(k)} )$}
\STATE {Observe other players' action $a_{-i}^{(k)}$}
\STATE { Update the policy belief as $\hat{\pi}_{-i}^{(k+1)} = \hat{\pi}_{-i}^{(k)}  + \alpha^{(k)} (e_{a_{-i}^{(k)}} - \hat{\pi}_{-i}^{(k)}) $}
\ENDFOR
\end{algorithmic} 
\end{algorithm} 

We provide the convergence guarantee of the FP dynamics as follows, showing that zero-sum NGs, i.e., zero-sum polymatrix games \cite{cai2011minmax,cai2016zero}, possess the fictitious-play property \cite{monderer1996fictitious}.

\begin{restatable}{theorem}{fppMZNG}
\label{thm:fppMZNG}
Assuming that $ \sum_{k=0}^{\infty}\alpha^{(k)} \to \infty$ and $\alpha^{(k)} \to 0$ as $k\to\infty$, then the limit points of $(\hat{\pi}^{(k)})_{k\geq 0}$  {are} the  NE of the zero-sum NG. 
\end{restatable}

\begin{proof}[Proof of \Cref{thm:fppMZNG}]
To prove the fictitious-play property, we consider a continuous version of \Cref{alg:polymatrix-game-fictitious}. Assuming that $\sum_{k=0}^\infty\alpha^{(k)} \to \infty$ and $\alpha^{(k)} \to 0$, \cite[Proposition 3.27] {benaim2005stochastic} states that we can characterize the limit set of $(\hat{\pi}^{(k)})_{k\geq 0}$ by considering the following dynamics:
\begin{align}\label{eqn:dynamics-fictitious-play}
\pi_i + \frac{d\pi_i}{dt} \in \argmax_{a_i\in\cA_i}~~ r_i(e_{a_i}, \pi_{-i}). 
\end{align}
We define a Lyapunov function as 
\begin{align}
V(\pi) &=\sum_{i \in \cN}  \left(\max_{a_i\in\cA_i} ~~r_i(e_{a_i}, {\pi}_{-i})- r_i(\pi) \right). \label{eqn:dynamics-fictitious-play-lya}
\end{align} 
\begin{restatable}{claim}{Lyapoly}
\label{claim:lyapunov-polymatrix}
$V(\pi(t))$ is a Lyapunov function for \eqref{eqn:dynamics-fictitious-play}.    
\end{restatable}
\begin{proof}
Let $\argmax_{a_i \in \cA_i} r_i(e_{a_i}, \pi_{-i})$ in the formula be ${a_i}^\star$, then we have 
\begin{align*}
    \frac{dV(\pi(t))}{dt} &= \sum_{i \in \cN} \left(\sum_{j \in \cE_i} e_{a_i^\star}^{\intercal} r_{i,j}\pi_j' \right)
     =\sum_{i \in \cN} \left(\sum_{j \in \cE_i } e_{a_i^\star}^{\intercal} r_{i,j}(e_{a_j^\star} - \pi_j) \right)  
    \\&
    =\sum_{i \in \cN}\left(\sum_{j \in \cE_i } -e_{a_i^\star}^{\intercal} r_{i,j}\pi_j \right) = -V(\pi(t))
\end{align*}
where {we use $\pi_j'$ to denote $\frac{d\pi_j}{dt}$, and} the first equality is derived from the envelope theorem. Since $\max_{a_i} r_i(a_i, {\pi}_{-i})\geq r_i(\pi)$, $V$ is guaranteed to  be non-negative. We can thus express $V(t) = V(0)  e^{-t}$, indicating that it is decreasing with a linear rate in continuous time.  
\end{proof}
Consequently, \cite[Proposition 3.27]{benaim2005stochastic} implies 
\begin{align*}
    \lim_{k \to \infty} \left( \sum_{i \in \cN} \left( \max_{a_i\in\cA_i} ~~r_i(e_{a_i}, \hat{\pi}_{-i}^{(k)}) - r_i(\hat{\pi}^{(k)})\right) \right) = 0
\end{align*}
which concludes that every limit point of $(\hat{\pi}^{(k)})_{k\geq 0}$ is an NE.
\end{proof}

{Note that the fictitious-play learning dynamics for zero-sum polymatrix games have also been proposed and analyzed in \cite{ewerhart2020fictitious}, and our result above is a reproduction of it.}

\numbering{\subsubsection{Smooth fictitious play in zero-sum NGs}}
\tnumbering{\subsection{Smooth fictitious play in zero-sum NGs}}
\label{ssec:stochastic-matrix-game}
\label{ssec:sfppqre}
We can also provide guarantees for the learning dynamics of \emph{smooth} fictitious play (may also be referred to as \emph{stochastic} fictitious play later) \cite{ref:Fudenberg93}, with convergence to the quantal response equilibrium (QRE) of the game  \cite{mckelvey1995quantal,mckelvey1998quantal}.

\begin{restatable}{definition}{QREnormalform}
\label{def:QRE}
A policy $\pi_\tau^{\star}=\left(\pi_{\tau, 1}^{\star}, \cdots, \pi_{ \tau,n}^{\star}\right)$ is a quantal response equilibrium of the game  with regularization coefficient ${\tau}$ if the following  condition holds
$$
\pi_{\tau, i}^{\star}(a_i)=\frac{\exp \left(\left[\br_i \pi_\tau^{\star}\right]_{a_i} / \tau \right)}{\sum_{a_i'
\in \cA_i} \exp \left(\left[\br_i \pi_\tau^{\star}\right]_{a_i'} / \tau\right)}
$$
for all $i\in\cN$ and $a_i \in \cA_i$ \cite{mckelvey1995quantal}. 
\end{restatable}

A QRE always exists in finite games. Moreover, a QRE has an equivalent notion as finding the Nash equilibrium of the game with  entropy-regularized payoffs: i.e., $\pi_\tau^{\star}$ satisfies that 
$$r_{\tau, i}\left(\pi_i^{\prime}, \pi_{\tau, -i}^{\star}\right) \leq r_{\tau,i}\left(\pi_{\tau}^{\star}\right),$$ where 
\begin{align}
 r_{\tau, i}(\pi) := r_i(\pi) + \tau \cH(\pi_i) - \sum_{j \in \cE_{r, i}}\frac{\tau}{|\cE_{r, j}|}\cH(\pi_j)   \label{eqn:rtaudef}
\end{align}
and $\cH(\pi_i):=-\sum_{a_i\in\cA_i}\pi_i(a_i)\log(\pi_i(a_i))$ is the Shannon entropy function  
\cite{mertikopoulos2016learning}. Reference \cite{leonardos2021exploration} provided a novel analysis showing that a unique NE  exists for zero-sum NGs {with entropy regularization (thus the QRE for the unregularized zero-sum NG)}. 

\begin{remark}
    In most existing literature \cite{ao2022asynchronous, leonardos2021exploration}, the entropy regularized reward is defined as $r_i(\pi) + \tau \cH(\pi_i)$. Indeed, note that 
    $$\argmax_{\pi_i \in \Delta(\cA_i)}\left(r_i(\pi) + \tau \cH(\pi_i)\right) = \argmax_{\pi_i \in \Delta(\cA_i)}\left(r_i(\pi) + \tau \cH(\pi_i) - \sum_{j \in \cE_{r, i}}\frac{\tau}{|\cE_{r, j}|}\cH(\pi_j)\right)$$
    for any $i\in\cN$, so it does not affect the equilibria. Moreover, by defining $r_{\tau, i}(\pi) := r_i(\pi) + \tau \cH(\pi_i) - \sum_{j \in \cE_{r, i}}\frac{\tau}{|\cE_{r, j}|}\cH(\pi_j)$, we can  have that $\sum_{i \in \cN} r_{\tau, i}(\pi) = 0$ holds for any joint product policy $\pi$. 
\end{remark}

\begin{algorithm}[H]
\caption{Stochastic fictitious play in zero-sum NGs  ($i$-th player)}
\label{alg:polymatrix-game-fictitious-stochastics}
\begin{algorithmic}
\STATE{Choose $\hat{\pi}_{j}^{(0)}$ as a uniform distribution for all $j \in \cN/\{i\}$}
\FOR{each timestep $k = 0, 1, \dots$}
\STATE {Take action $a_{i}^{(k)} \sim \argmax_{\mu_i\in\Delta(\cA_i)} r_{\tau, i}(\mu_i, \hat{\pi}_{-i}^{(k)} )$}
\STATE {Observe other players'  action $a_{-i}^{(k)}$}
\STATE { Update the policy belief as $\hat{\pi}_{-i}^{(k+1)} = \hat{\pi}_{-i}^{(k)}  + \alpha^{(k)} (e_{a_{-i}^{(k)}} - \hat{\pi}_{-i}^{(k)}) $}
\ENDFOR
\end{algorithmic}
\end{algorithm}

In \Cref{alg:polymatrix-game-fictitious-stochastics}, players initialize their beliefs for other players $(\hat{\pi}_{-i})$ as a uniform distribution. They sample from the best-response policy with respect to the entropy-regularized reward, given the beliefs of other players' policies. 
Subsequently, each player   observes other players' actions and updates her beliefs.

\begin{theorem}
\label{thm:sfppMZNG}
Assuming that $ \sum_{k=0}^\infty\alpha^{(k)} \to \infty$ and $\lim_{k\to \infty}\alpha^{(k)} \to 0$, $(\hat{\pi}^{(k)})_{k\geq 0}$ converges to a QRE of the zero-sum NG with probability 1.  
\end{theorem}
\begin{proof}[Proof of \Cref{thm:sfppMZNG}]
To prove the fictitious-play property, we consider a continuous-time version of the learning dynamics in \Cref{alg:polymatrix-game-fictitious-stochastics}. Assuming $\sum_{k=0}^\infty\alpha^{(k)} \to \infty$ and $\alpha^{(k)} \to 0$, \cite[Proposition 3.27]{benaim2005stochastic} states that we can characterize the limit set of $(\hat{\pi}^{(k)})_{k\geq 0}$ by considering the following dynamics 
\begin{align}
\pi_i + \frac{d\pi_i}{dt} =  \argmax_{\mu_i\in\Delta(\cA_i)} ~~r_{\tau, i}(\mu_i, \pi_{-i}). \label{eqn:dynamics-stochastic-fictitious-play}
\end{align}
We define a Lyapunov function as 
\begin{align*}
V_\tau(\pi) &=\sum_{i \in \cN}  \left(\max_{\mu_i\in\Delta(\cA_i)}~~r_{\tau, i}(\mu_i, {\pi}_{-i})- r_{\tau, i}(\pi) \right) .
\end{align*}
\begin{claim}
\label{claim:lyapunov-polymatrix-tau}
$V_\tau (\pi(t))$ is a Lyapunov function for \eqref{eqn:dynamics-stochastic-fictitious-play}.    
\end{claim}
\begin{proof}
Let the maximizer of $ r_{\tau, i}(\mu_i, \pi_{-i})$ in the formula be $\mu_i^\star${, which we know is unique due to the regularization}. Thus,  we have 
\small
\begin{align*} 
    \frac{dV_\tau(\pi(t))}{dt} &= \sum_{i \in \cN} \left(\left(\sum_{j \in \cE_i} \mu_i^{\star\intercal} r_{i,j}\pi_j' \right) - \left(\tau(\cH(\pi_i))'\right) \right)
     \\&=\sum_{i \in \cN} \left(\left(\sum_{j \in \cE_i} \mu_i^{\star\intercal} r_{i,j}(\mu_j^\star - \pi_j) \right) - \left(\tau(\cH(\pi_i))'\right) \right)
    \\
    &=\sum_{i \in \cN} \left(\left(\sum_{j \in \cE_i} -\mu_i^{\star\intercal} r_{i,j}\pi_j \right) + \tau (1 + \log \pi_i)^\intercal \pi_i' \right)
    \\&=\sum_{i \in \cN} \left(\left(\sum_{j \in \cE_i} -\mu_i^{\star\intercal} r_{i,j}\pi_j \right) + \tau (1 + \log \pi_i)^\intercal (\mu^\star_i - \pi_i) \right)
    \\
    &=\sum_{i \in \cN} \left(\left(\sum_{j \in \cE_i} -\mu_i^{\star\intercal} r_{i,j}\pi_j \right) + \tau(\cH(\pi_i)) + \tau (\log \pi_i)^\intercal (\mu^\star_i) \right)
     \\&\leq \sum_{i \in \cN} \left(\left(\sum_{j \in \cE_i} -\mu_i^{\star\intercal} r_{i,j}\pi_j \right) + \tau(\cH(\pi_i)) - \tau (\cH(\mu_i)) \right) = -V_\tau(\pi(t))
\end{align*}
\normalsize
where the first equality is derived from the envelope theorem and the last inequality is from Gibbs' inequality. 
Since $\max_{\mu_i\in\Delta(\cA_i)} r_{\tau, i}(\mu_i, {\pi}_{-i})\geq r_{\tau,i}(\pi)$, $V$ is guaranteed to be non-negative. Therefore, we have $0 \leq V(t) \leq V(0)  e^{-t}$, indicating that it is decreasing.  
\end{proof}
Consequently, \cite[Proposition 3.27]{benaim2005stochastic} implies that 
\begin{align*}
    \lim_{k \to \infty} \left( \sum_{i \in \cN} \left( \max_{\mu_i\in\Delta(\cA_i)}~~ r_{\tau, i}(\mu_i, \hat{\pi}_{-i}^{(k)}) - r_{\tau,i}(\hat{\pi}^{(k)})\right) \right) = 0
\end{align*}
which concludes that every limit point is a QRE. Since the QRE is unique for zero-sum NGs, we conclude that $(\pi^{(k)})_{k\geq 0}$ converges to the QRE of the zero-sum NG.
\end{proof}

\paragraph{Remark.} \Cref{alg:polymatrix-game-fictitious} converges to an NE, and \Cref{alg:polymatrix-game-fictitious-stochastics} converges to a QRE. Since the QRE is unique in zero-sum NG for a fixed $\tau$, we can identify the converging point in \Cref{alg:polymatrix-game-fictitious-stochastics}, while we cannot determine which NE is the converging point in \Cref{alg:polymatrix-game-fictitious}.

\numbering{\subsection{Fictitious-play property of  infinite-horizon zero-sum NMGs of a star-shape}}
\tnumbering{\section{Fictitious-play property of  infinite-horizon zero-sum NMGs of a star-shape}}
\label{sec:FP_inf_case}

Before presenting the results, 
we examine some properties of a star-shaped zero-sum NG (i.e., the polymatrix case). We define player 1 as the center player without loss of generality. 
First, we can view a star-shaped zero-sum NG as a constant-sum separable star-shaped game, as detailed below. 

\begin{proposition}\label{prop:prop_star_MZNG}
There exist some $\{c_i\}_{i\in\cN/\{1\}}$ with $c_i\in\RR$ such that a star-shaped zero-sum NG satisfies the following identities: 
\begin{align*}
&r_{i,1}^\intercal + r_{1,i} = c_i\pmb{1}\pmb{1}^\intercal \qquad \text{for every } i \in \cN/\{1\}, \qquad \sum_{i \in \cN / \{1\}} c_i = 0.
\end{align*}
\end{proposition}
\begin{proof}
    By the definition of a zero-sum NG, for arbitrary $\pi_1 \in \Delta(\cA_1), \{\pi_i \}_{i \in \cN/\{1\}} \in \prod_{i \in \cN/\{1\}} \Delta(\cA_i)$, the following holds:
    \begin{align}
        \sum_{i \in \cN/\{1\}} \left(\pi_1^\intercal r_{1,i} + \pi_1^\intercal r_{i,1}^\intercal \right) \pi_i = 0 \label{eqn:star-MZNG-eq}
    \end{align}
    which implies $\pi_1^\intercal (r_{1,i} + r_{i,1}^\intercal) = c_i \pmb{1}^\intercal $  for some constant $c_i$, since 
    \Cref{eqn:star-MZNG-eq} holds for any $\pi_i \in \Delta(\cA_i)$. {To be specific, $\pi_1^\intercal (r_{1,i} + r_{i,1}^\intercal) \pi_i$ 
should be the same when we plugging $\pi_i = e_{a_i}$ for any $a_i \in \cA_i$,  so that every element of $\pi_1^\intercal (r_{1,i} + r_{i,1}^\intercal)$ is the same, i.e., there exists some $c_i\in\RR$ such that $\pi_1^\intercal (r_{1,i} + r_{i,1}^\intercal) = c_i \pmb{1}^\intercal $.} Plugging to \Cref{eqn:star-MZNG-eq}, we have $\sum_{i \in \cN/\{1\}} c_i = 0$. Moreover, this again implies $r_{1,i} + r_{i,1}^\intercal = c_i\pmb{1}\pmb{1}^\intercal$, since   $\pi_1^\intercal (r_{1,i} + r_{i,1}^\intercal) = c_i \pmb{1}^\intercal $ always holds for any $\pi_1 \in \Delta(\cA_1)$ {by a similar argument as above.} 
\end{proof}

Second, we define the Nash equilibrium value for the center player in a star-shaped zero-sum NG, which is different from the general zero-sum NG case, where there may not exist a unique Nash value  \cite{cai2016zero}.

\begin{proposition}\label{prop:NE_unique}
There exists a unique Nash equilibrium value for the center player $1$ in a star-shaped zero-sum NG (i.e., $r_{i,j} = 0$ if $i \neq 1$ and $j \neq 1$).
\end{proposition}
\begin{proof}
    Player 1 aims to maximize $\sum_{i \in \cN/\{1\}} \pi_1^\intercal r_{1,i} \pi_i $ while player $i\neq 1$  aims to maximize $\pi_1^\intercal r_{i,1}^\intercal \pi_i = \pi_1^\intercal (c_i \pmb{1}\pmb{1}^\intercal - r_{1,i}) \pi_i = c_i - \pi_1^\intercal r_{1,i} \pi_i $, with $c_i$ given in \Cref{prop:prop_star_MZNG}.  We can solve these problems simultaneously by the following maxmin problem:
    \begin{align*}
        \maximize_{\pi_1 \in \Delta(\cA_1)}\minimize_{(\pi_i)_{i \in \cN/\{1\}} \in \prod_{i \in \cN/ \{1\}} \Delta(\cA_i)}\sum_{i \in \cN/\{1\}} \pi_1^\intercal r_{1,i} \pi_i.
    \end{align*}
    Since $\Delta(\cA_1)$ and $\prod_{i \in \cN/\{1\}} \Delta(\cA_i)$ are compact {and convex} sets, we can use the minimax theorem to show that  $\maximize_{\pi_1 \in \Delta(\cA_1)}\minimize_{(\pi_i)_{i \in \cN/\{1\}} \in \prod_{i \in \cN/\{1\}} \Delta(\cA_i)}\sum_{i \in \cN/\{1\}} \pi_1^\intercal r_{1,i} \pi_i$ is unique. 
\end{proof}

Note that there can be multiple Nash equilibrium values for {each} non-center player, but the {\it sum} of Nash equilibrium values for non-center players is always unique {(see the proof of \Cref{prop:NE_unique})}. 

Now, we start to prove that fictitious play dynamics given in \Cref{alg:markov-polymatrix-game-fictitious} converges to a Markov stationary NE in infinite-horizon zero-sum NMGs {with a star-shaped  network}. 

\fppMZNMG*
\begin{proof}
To prove the result, we consider a continuous-time version of \Cref{alg:markov-polymatrix-game-fictitious}. Using standard two-timescale stochastic approximation techniques \cite[Proposition 3.27]{benaim2005stochastic}, we can show that the limit set of our FP dynamics can be captured by that of a continuous-time differential inclusion. Before, we define several notation: 
\begin{align}
     &\bQ_{ i}(s) := (Q_{i,1}(s), \dots, Q_{ i,i-1 }(s), \bm{0}, Q_{ i,i+1}(s) \dots, Q_{ i, n}(s)) \in  \RR^{|\cA_i| \times \sum_{i \in \cN} |\cA_i| }   \nonumber
     \\     
     &\bQ(s) := ((\bQ_{ 1}(s))^\intercal, (\bQ_{2}(s))^\intercal, \dots, (\bQ_{n}(s))^\intercal)^\intercal \in \RR^{\sum_{i \in \cN} |\cA_i| \times \sum_{i \in \cN} |\cA_i| } \nonumber
     \\
     &h(\bQ(s)):= \max_{\mu\in \prod_{i \in \cN} \Delta (\cA_i)} \left|\left( \sum_{i \in \cN/\{1\}} \mu_1^\intercal Q_{1,i}(s) \mu_i + \sum_{i \in\cN/\{1\} }\mu_i^\intercal Q_{i,1}(s) \mu_1 \right)\right|. \label{eqn:h-function} 
\end{align}
Then, we consider the following  differential inclusion for each $s\in\cS$:
\begin{align}
&\pi_1(s) + \frac{d\pi_1(s)}{dt} \in \argmax_{a_1\in\cA_1} \left(\sum_{i\in\cN/\{1\}} e_{a_1}^\intercal Q_{1,i}(s) \pi_i(s)\right),  \, \nonumber
\\
& \pi_i(s) + \frac{d\pi_i(s)}{dt} \in \argmax_{a_i\in\cA_i} \left(e_{a_i}^\intercal Q_{i,1}(s) \pi_1(s) \right), \qquad  \frac{d Q_{1,i}(s)}{dt} = \frac{d Q_{i,1}(s)}{dt} = \bm{0}, \label{eqn:dynamics-fp-MZNMG}
\end{align}
with a Lyapunov function candidate being  
\numbering{\arxiv{\begin{align*}
L_\lambda (\pi, \bQ, s) &= \Biggl( \max_{a_1\in\cA_1}\left(\sum_{i \in \cN/\{1\}} e_{a_1}^\intercal Q_{1,i}(s) \pi_i(s)\right) + \sum_{i \in \cN/\{1\}}\max_{a_i\in\cA_i} \left(e_{a_i}^\intercal Q_{i,1}(s) \pi_1(s) \right) - \lambda h(\bQ(s)) \Biggr)_+
\end{align*}}}
\tnumbering{
\begin{align*}
L_\lambda (\pi, \bQ, s) &= \Biggl( \max_{a_1\in\cA_1}\left(\sum_{i \in \cN/\{1\}} e_{a_1}^\intercal Q_{1,i}(s) \pi_i(s)\right) 
\\
&+ \sum_{i \in \cN/\{1\}}\max_{a_i\in\cA_i} \left(e_{a_i}^\intercal Q_{i,1}(s) \pi_1(s) \right) - \lambda h(\bQ(s)) \Biggr)_+
\end{align*}
}
\neurips{\begin{align*}
&L_\lambda (\pi, \bQ, s) \\&\quad= \left( \max_{a_1\in\cA_1}\left(\sum_{i \in \cN/\{1\}} e_{a_1}^\intercal Q_{1,i}(s) \pi_i(s)\right) + \sum_{i \in \cN/\{1\}}\max_{a_i\in\cA_i} \left(e_{a_i}^\intercal Q_{i,1}(s) \pi_1(s) \right) - \lambda h(\bQ(s)) \right)_+
\end{align*}}
where $h$ is defined before,  
and $\lambda$ is chosen as {$1 < \lambda < 1/\gamma $. } Then, \Cref{claim:lyapunov-MZNMG}  below proves that $L_\lambda(\pi, \bQ, s)$ is a Lyapunov function for \eqref{eqn:dynamics-fp-MZNMG}.
\begin{claim}
\label{claim:lyapunov-MZNMG} For every $1< \lambda < 1/\gamma$, $L_\lambda (\pi, \bQ, s)$ is a Lyapunov function of \eqref{eqn:dynamics-fp-MZNMG} for the set $\Lambda = \{(\pi, \bQ): L_\lambda(\pi, \bQ, s) = 0\}$.
\end{claim}
\begin{proof}
First, we define $V_\lambda (\pi, \bQ, s)$ as below: 
\begin{align*}
V_\lambda (\pi, \bQ, s) &=  \max_{a_1\in\cA_1}\left(\sum_{i \in \cN/\{1\}} e_{a_1}^\intercal Q_{1,i}(s) \pi_i(s)\right) \tnumbering{\\&\qquad}+ \sum_{i \in \cN/\{1\}}\max_{a_i\in\cA_i} \left(e_{a_i}^\intercal Q_{i,1}(s) \pi_1(s) \right) - \lambda h(\bQ(s)).
\end{align*}
Then, we have $L_\lambda(\pi, \bQ ,s) = (V_\lambda (\pi, \bQ, s))_+$.  Moreover, let the maximizer of $\sum_{i \in \cN/\{1\}} e_{a_1}^\intercal Q_{1,i}(s) \pi_i(s)$ be $a_1^\star$ and let the maximizer of $e_{a_i}^\intercal Q_{i,1}(s) \pi_1(s)$ as $a_i^\star$ for $i \in \cN/\{1\}$. Then, we have
\begin{align*}
    \frac{dV_\lambda (\pi, \bQ, s)}{dt} &= \left(\sum_{i \in \cN/\{1\}} e_{a_1^\star}^{\intercal} Q_{1,i}(s) \pi_i(s)'\right) + \sum_{i \in \cN/\{1\}  }\left(e_{a_i^\star}^{\intercal} Q_{i,1}(s) \pi_1'(s) \right) 
    \\
    &= \left(\sum_{i \in \cN/ \{1\}} e_{a_1^\star}^{\intercal} Q_{1,i}(s) (e_{a_i^\star} - \pi_i(s))\right) + \sum_{i \in \cN/\{1\}   }\left(e_{a_i^\star}^{\intercal} Q_{i,1}(s) (e_{a_1^\star} -\pi_1(s)) \right)  
     \\
     &< - V_\lambda (\pi, \bQ, s) 
\end{align*}
since $\sum_{i \in \cN/\{1\}} e_{a_1^\star}^{\intercal} Q_{1, i}(s) e_{a_i^\star} + \sum_{i \in \cN/\{1\}} e_{a_i^\star}^{\intercal} Q_{i,1}(s) e_{a_1^\star} \leq h(\bQ(s)) < \lambda h(\bQ(s))$ holds by the definition of $h(\bQ(s))$. Therefore, $V_\lambda (\pi, \bQ, s )$ is strictly decreasing with respect to time when $V_\lambda (\pi, \bQ, s) \geq 0$. {To emphasize the time dependence of $V_\lambda$ and $L_\lambda$, we will write $V_\lambda(\pi, \bQ, s, t)$ and $L_\lambda(\pi, \bQ, s, t)$.}

If $V_\lambda (\pi, \bQ, s,t) \geq 0$, then $L_\lambda (\pi, \bQ, s, t) = V_\lambda (\pi, \bQ, s, t)$ is strictly decreasing if $L_\lambda (\pi, \bQ , s, t)>0$. Therefore, we can see  that if $L_\lambda(\pi, \bQ, s,t) > 0$ so that $V_\lambda (\pi, \bQ, s,t)>0$, then $L_\lambda(\pi, \bQ, s, t') < L_\lambda(\pi, \bQ, s, t)$ for all $t' > t$, i.e., $L_\lambda(\pi, \bQ, s, t)$ keeps strictly decreasing in this case.

If $V_\lambda (\pi,\bQ, s, t) < 0$, then $L_\lambda (\pi, \bQ, s, t) = 0$ always holds. Assume that there exists $t_1< t_2$ such that $V_\lambda (\pi, \bQ, s, t_1) < 0$ and $V_\lambda(\pi, \bQ, s, t_2)>0$. Due to the continuity of $V_\lambda$, there exists some $t\in(t_1,t_2)$ such that $V_\lambda(\pi, \bQ, s, t) = 0$. Then, $\frac{dV_\lambda (\pi, \bQ, s,t )}{dt} < -V_\lambda (\pi, \bQ, s, t) = 0$, so it is strictly negative, which prevents it from becoming a positive value, so it is a contradiction. Therefore, if $L_\lambda (\pi,\bQ, s, t) = 0$, then $L_\lambda(\pi,\bQ, s, t') = 0$ for all $t' > t$ in this case. 
\end{proof}
Therefore, \cite[Proposition 3.27]{benaim2005stochastic} implies
\numbering{\arxiv{\begin{align}
    \lim_{k \to \infty} \Biggl( \max_{a_1\in\cA_1}\left(\sum_{i \in \cN/\{1\}} e_{a_1}^\intercal \hat{Q}^{(k)}_{1,i} (s)\hat{\pi}^{(k)}_i(s)\right) + \sum_{i \in \cN/\{1\} }\max_{a_i\in\cA_i}   \left(e_{a_i}^\intercal \hat{Q}^{(k)}_{i,1}(s) \hat{\pi}^{(k)}_1(s) \right) - \lambda h(\hat{\bQ}^{(k)}(s)) \Biggr)_+ = 0   \label{eqn:stochastic-approx-MZNMG}
\end{align}}}
\tnumbering{{\begin{align}
    \lim_{k \to \infty} \Biggl( \max_{a_1\in\cA_1}&\left(\sum_{i \in \cN/\{1\}} e_{a_1}^\intercal \hat{Q}^{(k)}_{1,i} (s)\hat{\pi}^{(k)}_i(s)\right) \nonumber
    \\
    &+ \sum_{i \in \cN/\{1\} }\max_{a_i\in\cA_i}   \left(e_{a_i}^\intercal \hat{Q}^{(k)}_{i,1}(s) \hat{\pi}^{(k)}_1(s) \right) - \lambda h(\hat{\bQ}^{(k)}(s)) \Biggr)_+ = 0   \label{eqn:stochastic-approx-MZNMG}
\end{align}}}
\neurips{\begin{align}
    \lim_{k \to \infty} \Biggl( &\max_{a_1\in\cA_1}\left(\sum_{i \in \cN/\{1\}} e_{a_1}^\intercal \hat{Q}^{(k)}_{1,i} (s)\hat{\pi}^{(k)}_i(s)\right) \nonumber\\
    &\qquad+ \sum_{i \in \cN/\{1\} }\max_{a_i\in\cA_i}   \left(e_{a_i}^\intercal \hat{Q}^{(k)}_{i,1}(s) \hat{\pi}^{(k)}_1(s) \right) - \lambda h(\hat{\bQ}^{(k)}(s)) \Biggr)_+ = 0   \label{eqn:stochastic-approx-MZNMG}
\end{align}}
for every $s \in \cS$. 

In \Cref{claim:almost-zero-sum}, we will prove that an NG with $(\cG =(\cN, \cE_Q), \cA = (\cA_i)_{i \in \cN}, (\hat{Q}^{(k)}(s))_{(i, j) \in \cE_Q})$   asymptotically becomes  a zero-sum NG as $k \to \infty$ for all $s\in\cS$. 
Indeed, we have 
 \begin{align*}
     h(\hat{\bQ}^{(k)}(s)) &=  \max_{\mu\in \prod_{i \in \cN} \Delta (\cA_i)} \left|\left( \sum_{i \in \cN/\{1\}} \mu_1^\intercal \hat{\bQ}^{(k)}_{1,i}(s) \mu_i + \sum_{i \in\cN/\{1\} }\mu_i^\intercal \hat{\bQ}^{(k)}_{i,1}(s) \mu_1 \right)\right|, 
 \end{align*}
 we can conclude that for arbitrary policy $\mu\in \prod_{i \in \cN} \Delta (\cA_i)$, the sum of payoffs in NG with $(\cG =(\cN, \cE_Q), \cA = (\cA_i)_{i \in \cN}, (\hat{Q}^{(k)}(s))_{(i, j) \in \cE_Q})$ goes to 0 as $k \to \infty$. Therefore, $h(\hat{\bQ}^{(k)}(s)) \to 0$ implies that the NG is zero-sum NG, where $h$ is defined as \Cref{eqn:h-function}.

Before stating that $\hat{Q}^{(k)}$ asymptotically becomes zero-sum NGs, we state a lemma from \cite{sayin2022fictitious}.  

\begin{lemma}[\cite{sayin2022fictitious}]
\label{lem:sayincontraction}
Suppose the {sequence  of random variables $(y_k)_{k\geq 0}$ with $y_k\in\RR^d$ satisfies}  
\begin{align*}
    &y_{k+1}[n] \leq\left(1-\beta_{n, k}\right) y_k[n]+\beta_{n, k}\left(\gamma\left\|y_k\right\|_{\infty}+\bar{\epsilon}_k+\omega_{n, k}\right) 
    \\
    &y_{k+1}[n] \geq\left(1-\beta_{n, k}\right) y_k[n]+\beta_{n, k}\left(-\gamma\left\|y_k\right\|_{\infty}+\underline{\epsilon}_k+\omega_{n, k}\right)
\end{align*}
for all $k \geq 0$, where {$y_{k}[n]$ denotes the $n$-th element in $y_k$,} $\gamma \in (0,1)$, $\sum_{k=0}^{\infty} \beta_{n, k}=\infty, \lim _{k \rightarrow \infty} \beta_{n, k}=0$ for each $n$ with probability 1, the error sequence $(\epsilon_k)_{k\geq 0}$ satisfies  $\limsup _{k \rightarrow \infty}\left|\bar{\epsilon}_k\right| \leq c ~~\text {and} ~~ \limsup _{k \rightarrow \infty}\left|\underline{\epsilon}_k\right| \leq c$ {for some $c\geq 0$} with probability 1. Here, $\omega_{n,k}$ is a stochastic approximation term that is zero-mean and has finite variance conditioned on the history. Suppose that $\left\|y_k\right\|_{\infty}$ is bounded for all $k$. Then, we have
$
\limsup _{k \rightarrow \infty}\left\|y_k\right\|_{\infty} \leq \frac{c}{1-\gamma}
$ 
with probability 1, provided that either $\omega_{n, k}=0$ for all $n, k$ or $\sum_{k=0}^{\infty} \beta_{n, k}^2<\infty$ for each $n$ with probability 1.    
\end{lemma}
\begin{claim}
 \label{claim:almost-zero-sum} $h(\hat{\bQ}^{(k)}(s))$ converges to 0 for all $s \in \cS$. In other words, an NG with $(\cG =(\cN, \cE_Q), \cA = (\cA_i)_{i \in \cN}, (\hat{Q}^{(k)}(s))_{(i, j) \in \cE_Q})$   asymptotically becomes a zero-sum NG as $k \to \infty$ for all $s\in\cS$.    
\end{claim}
\begin{proof} 
Rewriting \Cref{eqn:stochastic-approx-MZNMG} with the belief of the value function, we have  that for all $s\in\cS$
\begin{align}
    \lim_{k \to \infty} \left( \sum_{i\in\cN} \hat{V}_i^{(k)}(s)  - \lambda h(\hat{\bQ}^{(k)}(s)) \right)_+ = 0.  
    \label{eqn:v-bound}
\end{align}
By the definition of $\hat{V}_i^{(k)}(s)$ and \Cref{eqn:v-bound}, we have 
\begin{align*}
    -\lambda h(\hat{\bQ}^{(k)}(s)) \leq -h(\hat{\bQ}^{(k)}(s)) \leq \sum_{i \in \cN} \hat{V}_i^{(k)}(s) \leq \lambda h(\hat{\bQ}^{(k)}(s)) + \bar{\epsilon}_k(s)
\end{align*} 
for all $s \in \cS$ and $k \geq 0$ for some $(\bar{\epsilon}_k(s))_{k\geq 0}$, where $\bar{\epsilon}_k(s) \to 0$. Moreover, summing over  all the $Q$-belief estimates over $i$, we have
    $$\sum_{i \in\cN} \hat{Q}_{i}^{(k+1)}(s, \ba) = ( 1- \bar{\beta}_{k}(s) ) \sum_{i \in \cN} \hat{Q}_{i}^{(k)}(s,\ba)  + \gamma\bar{\beta}_{k}(s) \left(   \sum_{s' \in S}\PP(s' \mid s, \ba) \sum_{i \in \cN} \hat{V}_{i}^{(k)}(s') \right) $$
    where $\bar{\beta}_k(s) := \pmb{1}(s = s^{(k)}) \beta^{(N(s))}$. Thus, we have 
    \begin{align*}
    &\sum_{i \in \cN} \hat{Q}_{i}^{(k+1)}(s, \ba) \leq ( 1- \bar{\beta}_{k} \kzedit{(s)}) \sum_{i \in \cN}  \hat{Q}_{i}^{(k)}(s,\ba)  + \bar{\beta}_{k}\kzedit{(s)}\left(  \bar{\gamma}  \max_{s' \in \cS} h(\hat{\bQ}^{(k)}(s')) + \kzedit{\gamma}\bar{\epsilon}^{(k)}\right)          
    \\
    &\sum_{i \in \cN} \hat{Q}_{i}^{(k+1)}(s, \ba) \geq ( 1- \bar{\beta}_{k}\kzedit{(s)}) \sum_{i \in \cN}  \hat{Q}_{i}^{(k)}(s,\ba)  - \bar{\beta}_{k}\kzedit{(s)}\left(  \bar{\gamma}  \max_{s' \in \cS} h(\hat{\bQ}^{(k)}(s'))\right)
    \end{align*}
where $\bar{\gamma} = \gamma \lambda \in (0,1)$. Since $\max_{s' \in \cS} h(\hat{\bQ}^{(k)}(s'))$ is the  maximal value of $\big| \sum_{i \in \cN} \hat{Q}_{i}^{(k)}(s, \ba) \big|$, we can apply 
\Cref{lem:sayincontraction} to this situation. {Let $\cZ:=\cS \times \cA$ be the set of all possible state-action pairs.} Then, we can view this problem as $y_k[z]:=\sum_{i \in \cN} \hat{Q}_{i}^{(k)}(s, \ba)$, for all $k \geq 0$. Note that $y_k$ is always bounded by $2Rn/(1-\gamma)$, since the reward function is bounded by $R$, and, every timestep we update the sum of $Q_i$-value estimates over $i \in \cN$ with a convex combination of the previous sum of $Q_i$-value estimates over $i \in \cN$ and $\left(\sum_{i}r_i(s, \ba) + \gamma \sum_{s' \in \cS} \PP(s'\mid s, \ba) V(s')\right) \leq \left(2Rn + \gamma \sum_{s' \in \cS} \PP(s'\mid s, \ba) \max_{\ba'} Q(s', \ba')\right) \leq 2Rn/(1-\gamma)$, so we can recursively show all the sum of $Q_i$-value estimates over $i \in \cN$ iterates are bounded. Therefore, \Cref{lem:sayincontraction} yields that $\max_{s' \in \cS} h(\hat{\bQ}^{(k)}(s')) = \|y_k\|_{\infty} \rightarrow 0$ as $k \rightarrow \infty$. As a byproduct, we also have  $
     \lim_{k \to \infty} \big|\sum_{i\in\cN} \hat{V}^{(k)}_i (s)\big|  = 0,$ 
completing the proof.
\end{proof}

\kzedit{For any $s\in\cS$,} 
for {a given} $\bQ(s) \in \RR^{\sum_{i \in \cN} |\cA_i| \times \sum_{i \in \cN} |\cA_i|}$  
as defined in \eqref{eqn:h-function}, we {define} $\bQ_1(s):= (Q_{1,2}(s), \dots, Q_{1, n}(s))$  and $\bQ_{-1}(s):= (Q_{2,1}^\intercal(s), \dots, Q_{n,1}^\intercal(s) )^\intercal$ for all $s \in \cS$. Then, we define $\text{Val}_1$ and $\text{Val}_{-1}$, which are the maxmin operators  with respect to $\bQ_1(s)$ and $\bQ_{-1}(s)$, respectively, as follows:
\begin{align*}
    &\text{Val}_1(\bQ_1(s)) = \max_{\mu_1 \in \Delta(\cA_1)}  \min_{\mu_2 \in \Delta(\cA_2), \dots, \mu_{n} \in \Delta(\cA_n) } \sum_{i \in \cN / \{1\}}\mu_1^\intercal {Q}_{1,i}(s) \mu_i 
    \\
    &\text{Val}_{-1}(\bQ_{-1}(s)) = \max_{\mu_2 \in \Delta(\cA_2), \dots, \mu_{n} \in \Delta(\cA_n)}\min_{\mu_1\in \Delta(\cA_1)}   \sum_{i \in \cN/\{1\}}\mu_1^\intercal {Q}_{i,1}^\intercal(s) \mu_i.
\end{align*}
Note that the $\text{Val}_1$ and $\text{Val}_{-1}$ operators can be viewed as the maxmin operator in the two-player zero-sum case, and it is indeed the star-shaped topology that enables us to write out a value iteration operator based on it, whose fixed point corresponds to the NE of the game. In general, it is hard to define value-iteration operators induced by such $\text{Val}_1$ and $\text{Val}_{-1}$ for other network structures. Also, note that since  the maxmin formulas  in $\text{Val}_1$ and $\text{Val}_{-1}$ are by definition \emph{non-expansive}, the induced value iteration operator is  \emph{contracting} (due to the $\gamma\in(0,1)$ discount factor), which is key in showing the convergence of our FP dynamics. 

\begin{claim}
\label{claim:Vtovalhat} 
$|\hat{V}_1^{(k)}(s) - \text{Val}_1(\hat{\bQ}_1^{(k)}(s)) |$  and $|\sum_{i\in\cN / \{1\}}\hat{V}_i^{(k)}(s) - \text{Val}_{-1}(\hat{\bQ}_{-1}^{(k)}(s)) |$ converge  to 0 {for all $s\in\cS$}. 
\end{claim}
\begin{proof}
The definition of $\hat{V}_i^{(k)}$ gives 
\begin{align*}
    &\hat{V}_1^{(k)}(s) = \max_{a_i \in \cA_i} \EE_{a_{-i} \sim \hat{\pi}_{-1}^{(k)}}\{\hat{Q}_1^{(k)}(s, \ba)\} \geq \text{Val}_1(\hat{\bQ}_1^{(k)}(s)) \geq \min_{\mu_2, \dots, \mu_{n}}\sum_{i \in \cN/\{1\}}(\hat{\pi}_1^{(k)})^\intercal \hat{Q}_{1,i}^{(k)}(s) \mu_i 
    \\
    &\quad\geq \min_{\mu_2, \dots, \mu_{n}}\sum_{i \in \cN/\{1\}}(\hat{\pi}_1^{(k)})^\intercal (-(\hat{Q}_{i,1}^{(k)}(s))^\intercal) \mu_i  + \min_{\mu_2, \dots, \mu_{n}}\sum_{i \in \cN/\{1\}}(\hat{\pi}_1^{(k)})^\intercal (\hat{Q}_{1,i}^{(k)}(s) + (\hat{Q}_{i,1}^{(k)}(s))^\intercal) \mu_i
    \\
    &\quad\geq -\max_{\mu_2, \dots, \mu_{n}}\sum_{i \in \cN/\{1\}}(\hat{\pi}_1^{(k)})^\intercal (\hat{Q}_{i,1}^{(k)}(s))^\intercal \mu_i  + \min_{\mu_2, \dots, \mu_{n}}\sum_{i \in \cN/\{1\}}(\hat{\pi}_1^{(k)})^\intercal (\hat{Q}_{1,i}^{(k)}(s) + (\hat{Q}_{i,1}^{(k)}(s))^\intercal) \mu_i
    \\
    &\quad\geq -\max_{\mu_2, \dots, \mu_{n}}\sum_{i \in \cN/\{1\}}(\hat{\pi}_1^{(k)})^\intercal (\hat{Q}_{i,1}^{(k)}(s))^\intercal \mu_i  - h(\hat{\bQ}^{(k)}{(s)}) 
     = - \sum_{i\in\cN / \{1\}} \hat{V}_i^{(k)}(s)   - h(\hat{\bQ}^{(k)}{(s)}),
\end{align*} 
where the third inequality is due to the summation of minimization being no greater than the minimization, and the fifth inequality is from the definition of $h$. The above inequality further implies 
\begin{align*}
    \hat{V}_1^{(k)}(s) + \sum_{i\in \cN/\{1\} }\hat{V}_i^{(k)}(s)  + h(\hat{\bQ}^{(k)}{(s)}) \geq \hat{V}_1^{(k)}(s) - \text{Val}_1({\hat{\bQ}_1^{(k)}(s)}) \geq 0.
\end{align*}
The left-hand side goes to zero when $k \to \infty$, so the lemma is proved for player $1$. The other direction can be proved in the same way. 
\end{proof}

{
Then, we define the  value-iteration operators  $\cT_1: \RR^{|\cS| \times (|\cA_1| \times\sum_{i \in \cN/\{1\})} |\cA_i|)} \to  \RR^{|\cS| \times (|\cA_1| \times\sum_{i \in \cN/\{1\})} |\cA_i|)}$ and $\cT_{-1}:  \RR^{|\cS| \times (|\cA_1| \times \sum_{i \in \cN/\{1\}}|\cA_i|) }\to  \RR^{|\cS| \times (|\cA_1| \times \sum_{i \in \cN/\{1\}}|\cA_i|)}$ as in a two-player zero-sum Markov game \cite{shapley1953stochastic} as follows:
\begin{align*}
    &\left(\cT_{1} \bQ_1\right)(s, a_1, a_i) = r_{1, i}(s,a_1, a_i) + \gamma \sum_{s' \in \cS} \frac{1}{n-1}\PP_1 (s'  \mid s, a_1) \text{Val}_1(\bQ_1(s'))
    \\
    &\left(\cT_{-1} \bQ_{-1}\right)(s, a_1, a_i) =r_{i, 1}(s,a_i, a_1) + \gamma \sum_{s' \in \cS} \frac{1}{n-1}\PP_1(s'  \mid s, a_1) \text{Val}_{-1}(\bQ_{-1}(s')). 
\end{align*} 
Also, we define several norms:
\begin{align*}
    &\norm{\cdot}_{\max}: \RR^{m \times n} \to \RR \text{ such that } \norm{A}_{\max} = \max_{i \in [m], j \in [n]} |A_{i,j}|
    \\
    &\norm{\cdot}_{\max, 1} :  \RR^{|\cA_1| \times \sum_{i \in \cN/\{1\}} |\cA_i|} \to \RR \text{ such that } \norm{X_s}_{\max, 1}:= \sum_{i \in \cN/\{1\}}\norm{X_{s, i}}_{\max}
    \\
    &\norm{\cdot}_{\max, 1, \max} : \RR^{|\cS| \times (|\cA_1| \times\sum_{i \in \cN/\{1\})} |\cA_i|)} \to \RR \text{ such that } \norm{X}_{\max, 1}:= \norm{(\norm{X_s}_{\max,1})_{s \in \cS}}_{\max}.
\end{align*}
\begin{claim}
\label{claim:contraction}
    $\cT_1$ and $\cT_{-1}$ are contracting with respect to the norm $\norm{\cdot}_{\max, 1, \max}$. 
\end{claim}
\begin{proof}
By definition of the $\text{Val}_1$  operator, we have that for any $\bQ_1 := (\bQ_1(s))_{s \in \cS} \in \RR^{|\cS| \times (|\cA_1| \times \sum_{j \in \cN / \{1\}}|\cA_j|)}$ where $\bQ_1(s) = (Q_{1,2}(s), \dots, Q_{1, n}(s)) \in \RR^{|\cA_1| \times \sum_{j \in \cN / \{1\}}|\cA_j|}$ and $\bQ_1'(s) = (\bQ_1'(s))_{s \in \cS} \in \RR^{|\cS| \times (|\cA_1| \times \sum_{j \in \cN / \{1\}}|\cA_j|)}$ where $\bQ_1'(s) =(Q_{1,2}'(s), \dots, Q_{1, n}'(s))  \in \RR^{|\cA_1| \times \sum_{j \in \cN / \{1\}}|\cA_j|}$ for all $s$, and for any $i \in \cN/\{1\}$, $s \in \cS$, 
\begin{align*}
    &\norm{\left((\cT_{1} \bQ_1)(s, a_1, a_i) - (\cT_{1} \bQ_1')(s, a_1, a_i)\right)_{a_1 \in \cA_1, a_i \in \cA_i}}_{\max}
    \\
    &\qquad\leq  \max_{a_1} \left| \frac{\gamma}{n-1} \sum_{s' \in \cS} \PP_1(s'  \mid s, a_1) ( \text{Val}_1(\bQ_1(s'))- \text{Val}_1(\bQ_1'(s')))\right| 
    \\
    &\qquad\leq  \frac{\gamma}{n-1} \sum_{s' \in \cS} \PP_1(s'  \mid s, a_1) \max_{a_1} \left| \max_{a_2, \dots, a_{n}} \sum_{i \in \cN/\{1\}} (Q_{1,i} - Q_{1,i}')(s', a_1, a_i) \right| 
    \\
    &\qquad\leq \frac{\gamma}{n-1}\norm{\bQ_1 - \bQ_1'}_{\max, 1, \max},
\end{align*}
so for any $s \in \cS$, $\norm{\cT_1 \bQ_1(s) - \cT_1 \bQ_1'(s)}_{\max, 1} = \sum_{i \in \cN/\{1\}} \norm{\left(\cT_{1} \bQ_1(s, a_1, a_i) - \cT_{1} \bQ_1'(s, a_1, a_i)\right)_{a_1 \in \cA_1, a_i \in \cA_i}}_{\max} \leq \gamma \norm{\bQ_1 - \bQ_1'}_{\max, 1, \max}$ holds and therefore $ \norm{\cT_1 \bQ_1 - \cT_1 \bQ_1'}_{\max, 1, \max} \leq \gamma \norm{\bQ_1 - \bQ_1'}_{\max, 1, \max}$.

Similarly,  we have that for  any $\bQ_{-1}:= (\bQ_{-1}(s))_{s \in \cS} \in \RR^{|\cS| \times (|\cA_1| \times \sum_{j \in \cN / \{1\}}|\cA_j|)}$  where $\bQ_{-1}(s) = (Q_{2,1}^\intercal(s), \dots, Q_{n, 1}^\intercal(s) ) \in \RR^{|\cA_1| \times \sum_{j \in \cN / \{1\}}|\cA_j| }$,  and $\bQ_{-1}':= (\bQ_{-1}'(s))_{s \in \cS} \in \RR^{|\cS| \times (|\cA_1| \times \sum_{j \in \cN / \{1\}}|\cA_j|)}$ where $\bQ_{-1}'(s) = (Q_{2,1}^{'\intercal}(s), \dots, Q_{n, 1}^{'\intercal}(s)) \in \RR^{|\cA_1| \times \sum_{j \in \cN / \{1\}}|\cA_j| }$
\begin{align*}
    &\norm{\left((\cT_{-1} \bQ_{-1})(s, a_1, a_i) - (\cT_{-1} \bQ_{-1}')(s, a_1, a_i)\right)_{a_1 \in \cA_1, a_i \in \cA_i}}_{\max}
    \\
    &\qquad\leq  \max_{a_1} \left|\frac{\gamma}{n-1} \sum_{s' \in \cS} \PP_1(s'  \mid s, a_1) ( \text{Val}_{-1}(\bQ_{-1}(s'))- \text{Val}_{-1}(\bQ_{-1}'(s')))\right|
    \\
    &\qquad\leq  \frac{\gamma}{n-1} \sum_{s' \in \cS} \PP_1(s'  \mid s, a_1)\max_{a_1} \left|  \max_{a_2, \dots, a_{n}}  \sum_{i \in \cN/\{1\}} (Q_{i,1} - Q_{i,1}')(s', a_i, a_1) \right| 
    \\
    &\qquad\leq \frac{\gamma}{n-1}\norm{\bQ_{-1} - \bQ_{-1}'}_{\max, 1, \max},
\end{align*}
so for any $s \in \cS$, $ \norm{\cT_{-1} \bQ_{-1}(s) - \cT_{-1} \bQ_{-1}'(s)}_{\max, 1} = \sum_{i \in \cN/\{1\}} \norm{\left(\cT_{-1} \bQ_{-1}(s, a_1, a_i) - \cT_{-1} \bQ_{-1}'(s, a_1, a_i)\right)_{a_1 \in \cA_1, a_i \in \cA_i}}_{\max} \leq \gamma \norm{\bQ_{-1} - \bQ_{-1}'}_{\max, 1, \max}$ holds and therefore $ \norm{\cT_{-1} \bQ_{-1} - \cT_{-1} \bQ_{-1}'}_{\max, 1, \max} \leq \gamma \norm{\bQ_{-1} - \bQ_{-1}'}_{\max, 1, \max}$.
\end{proof}}

Since the operators $\cT_1$ and $\cT_{-1}$ are contracting, they each have a unique fixed point denoted by $\bQ^\star_1$ and $\bQ^\star_{-1}$, respectively. Then, by the definition of  fixed point, we have 
\begin{align*}
    \sum_{i \in \cN/\{1\}}  Q_{1, i}^\star(s, a_1, a_i) &= \sum_{i \in \cN/\{1\}} r_{1, i}(s, a_1, a_i) + \gamma \sum_{s' \in \cS} \PP_1(s'\mid s, a_1) \max_{\mu_1 \in \Delta(\cA_1)}  \min_{\mu_2 \in \Delta(\cA_2), \dots, \mu_{n} \in \Delta(\cA_n) } \sum_{i \in \cN/\{1\}} \mu_1^\intercal   {Q}_{1,i}^\star(s') \mu_i
    \\
    \sum_{i \in \cN/\{1\}}  Q_{i, 1}^\star(s, a_i, a_1) &= \sum_{i \in \cN/\{1\}} r_{i, 1}(s, a_i, a_1) + \gamma \sum_{s' \in \cS} \PP_1(s'\mid s, a_1) \max_{\mu_2 \in \Delta(\cA_2), \dots, \mu_{n} \in \Delta(\cA_n)}\min_{\mu_1\in \Delta(\cA_1)}   \sum_{i \in \cN/\{1\}}\mu_i^\intercal {Q}^\star_{i,1}(s') \mu_1, 
\end{align*}
and one can check that
\begin{align*}
\max_{\mu_2 \in \Delta(\cA_2), \dots, \mu_{n} \in \Delta(\cA_n)}&\min_{\mu_1\in \Delta(\cA_1)}   \sum_{i \in \cN/\{1\}}\mu_i^\intercal {Q}^\star_{i,1}(s) \mu_1 = \min_{\mu_1\in \Delta(\cA_1)} \max_{\mu_2 \in \Delta(\cA_2), \dots, \mu_{n} \in \Delta(\cA_n)}  \sum_{i \in \cN/\{1\}}\mu_i^\intercal {Q}^\star_{i,1}(s) \mu_1    
\\&= \min_{\mu_1\in \Delta(\cA_1)} \max_{\mu_2 \in \Delta(\cA_2), \dots, \mu_{n} \in \Delta(\cA_n)}  -\sum_{i \in \cN/\{1\}}\mu_i^\intercal (-{Q}^\star_{i,1}(s)) \mu_1
\\&= -\max_{\mu_1\in \Delta(\cA_1)} \min_{\mu_2 \in \Delta(\cA_2), \dots, \mu_{n} \in \Delta(\cA_n)}  \sum_{i \in \cN/\{1\}}\mu_i^\intercal (-{Q}^\star_{i,1}(s)) \mu_1
\end{align*}
by the minimax theorem. Thus, we have 
\begin{align*}
&\left| \sum_{i \in \cN / \{1\}} \left(Q^\star_{1, i}(s, a_1, a_i) + Q_{i, 1}^\star(s, a_i, a_1)\right)\right| 
\\
&\leq \gamma \sum_{s' \in \cS} \PP_1(s'|s, a_1) \bigg| \max_{\mu_1\in \Delta(\cA_1)} \min_{\mu_2 \in \Delta(\cA_2), \dots, \mu_{n} \in \Delta(\cA_n)}\sum_{i \in \cN / \{1\}} \mu_1^\intercal   {Q}_{1,i}^\star(s') \mu_i \\
&\qquad\qquad -\max_{\mu_1\in \Delta(\cA_1)} \min_{\mu_2 \in \Delta(\cA_2), \dots, \mu_{n} \in \Delta(\cA_n)} \sum_{i \in \cN / \{1\}}\mu_i^\intercal   (-{Q}_{i,1}^\star(s')) \mu_1  \bigg| 
\\
&\leq \gamma \sum_{s' \in \cS} \PP_1(s'|s, a_1)  \max_{\mu_1\in \Delta(\cA_1), \mu_2 \in \Delta(\cA_2), \dots, \mu_{n} \in \Delta(\cA_n)}\left|\sum_{i \in \cN / \{1\}}\left( \mu_1^\intercal   {Q}_{1,i}^\star(s') \mu_i + \mu_i^\intercal   ({Q}_{i,1}^\star(s')) \mu_1  \right)\right| 
\\
&\leq \gamma \max_{s, a_1, a_{-1}} \left|\sum_{i \in \cN / \{1\} } \left( Q^\star_{1, i}(s, a_1, a_i) + Q_{i, 1}^\star(s, a_i, a_1)\right)\right|
\end{align*}
Therefore, we conclude that  
$\sum_{i \in \cN/\{1\}}Q^\star_{1, i} (s, a_1, a_i ) + \sum_{i \in \cN/\{1\}}Q^\star_{i, 1}(s, a_i, a_1) = 0$  for every $(s, \ba)$ and $i \in \cN/\{1\}$,  {by iteratively unrolling the inequality above.}

{Moreover,} the update of beliefs on the $Q$-function can be written as 
\numbering{\arxiv{\begin{align*}
    \hat{Q}_{1}^{(k+1)}(s, \ba) &= (1-\bar{\beta}_{k}(s)) \hat{Q}_{1}^{(k)}(s,\ba)  + \bar{\beta}_{k}(s)\left( \sum_{i \in \cN/\{1\}} \cT_{1} \hat{\bQ}_1^{(k)} (s, a_1, a_i) + \cE_1^{(k)}(s, \ba)  \right)
    \\
    \sum_{i \in \cN/\{1\}} \hat{Q}_{i}^{(k+1)}(s, \ba) &= (1 - \bar{\beta}_{k}(s))\sum_{i \in \cN/\{1\}}\hat{Q}_{i}^{(k)}(s,\ba)  + \bar{\beta}_{k}(s)\left(  \sum_{i \in \cN /\{1\}} \cT_{-1} \hat{\bQ}_{i}^{(k)}(s,a_1, a_i) + \cE_{-1}^{(k)}(s, \ba)  \right).
\end{align*}}}
Here, $\cE_1^{(k)}(s,\ba)$ and $\cE_{-1}^{(k)}(s,\ba)$ are defined as 
\begin{align*}
    &\cE_1^{(k)}(s,\ba) = \gamma \sum_{s' \in \cS} \PP_1(s'\mid s, a_1)\left[\hat{V}_1^{(k)}(s') -  \text{Val}_1(\hat{\bQ}_1^{(k)}(s'))\right]
    \\
    &\cE_{-1}^{(k)}(s,\ba) = \gamma \sum_{s' \in \cS} \PP_1(s'\mid s, a_1)\left[\sum_{i \in\cN/\{1\}}\hat{V}_{i}^{(k)}(s') -  \text{Val}_{-1}(\hat{\bQ}_{-1}^{(k)}(s'))\right]
\end{align*}
where the two values go to 0 by \Cref{claim:Vtovalhat}. {For each $s\in\cS$, we further define}  $\hat{\bQ}_1^{(k)}(s):= \left(\left(\hat{Q}_1^{(k)}(s, a_1, a_i)\right)_{a_1 \in \cA_1, a_i \in \cA_i}\right)_{i \in \cN/\{1\}}${, and similarly define}   $\hat{\bQ}_{-1}^{(k)}(s)$ as well. 

\begin{claim}
\label{claim:hatQtoQstar}
$|\hat{Q}_1^{(k)}(s,\ba) - \sum_{i \in \cN/\{1\}} Q_{1, i}^\star(s,a_1, a_i)|$ and $|\hat{Q}_{-1}^{(k)}(s,\ba) - \sum_{i \in \cN/\{1\}} Q_{i, 1}^\star(s,a_i, a_1)|$ converge to 0 as $k \to \infty$ {for all $s\in\cS$ and $\ba\in\cA$}. 
\end{claim}
\begin{proof}
    Define  $\tilde{Q}_1^{(k)}(s,\ba) :=\hat{Q}_1^{(k)}(s,\ba) - \sum_{i \in \cN/\{1\}} Q_{1, i}^\star(s,a_1, a_i)$ and $\tilde{Q}_{-1}^{(k)}(s,\ba) :=\hat{Q}_{-1}^{(k)}(s,\ba) - \sum_{i \in \cN/\{1\}} Q_{i, 1}^\star(s,a_i, a_1)$. Then, by the fact that $\bQ_1^\star$ and $\bQ_{-1}^\star$ are  the fixed point of $\cT_1$ and $\cT_{-1}$, respectively, we have that for each $s\in\cS$ and $\ba\in\cA$: 
\numbering{\arxiv{\begin{align*}
        &\tilde{Q}_{1}^{(k+1)}(s, \ba) = (1-\bar{\beta}_{k}(s)) \tilde{Q}_{1}^{(k)}(s,\ba)  + \bar{\beta}_{k}(s)\left(  \sum_{i \in \cN/\{1\}} \cT_1 \hat{\bQ}_{1}^{(k)}(s, a_1, a_i) -  \sum_{i \in \cN/\{1\}}\cT_1 \bQ_1^\star(s, a_1,  a_i) + \cE_1^{(k)}(s, \ba)  \right)
        \\
        &\tilde{Q}_{-1}^{(k+1)}(s, \ba) = (1 - \bar{\beta}_{k}(s))\tilde{Q}_{-1}^{(k)}(s, \ba)  + \bar{\beta}_{k}(s)\left( \sum_{i \in \cN/\{1\}} \cT_{-1}  \hat{\bQ}_{-1}^{(k)}(s,a_1, a_i) -  \sum_{i \in \cN/\{1\}} \cT_{-1}  \bQ_{-1}^\star(s, a_1, a_i) + \cE_{-1}^{(k)}(s, \ba)  \right)
    \end{align*}}}
    and \Cref{claim:contraction} implies that  
    \begin{align*}
    &\tilde{Q}_{1}^{(k+1)}(s, \ba) \leq ( 1- \bar{\beta}_{k}(s) ) \tilde{Q}_{1}^{(k)}(s,\ba)  + \bar{\beta}_{k}(s)\left(  {\gamma}  \norm{\tilde{Q}_1}_{\max, 1 ,\max} + \bar{\epsilon}^{(k)}\right)         
    \\
    &\tilde{Q}_{1}^{(k+1)}(s, \ba) \geq ( 1- \bar{\beta}_{k}(s) ) \tilde{Q}_{1}^{(k)}(s,\ba)  + \bar{\beta}_{k}(s)\left(  -{\gamma}  \norm{\tilde{Q}_1}_{\max, 1 ,\max}  - \bar{\epsilon}^{(k)}\right)        
    \end{align*}
    which yields  $\tilde{Q}_{1}^{(k+1)}(s, \ba) \to 0$ and also $\tilde{Q}_{-1}^{(k+1)}(s, \ba) \to 0$ as $k\to\infty$ by \Cref{lem:sayincontraction}. 
\end{proof}
Therefore, we verified that $\hat{V}_1^{(k)}(s) - \text{Val}_1(Q^\star_1 (s)) \to 0$ and $\hat{V}_{-1}^{(k)}(s) - \text{Val}_{-1}(Q^\star_{-1} (s)) \to 0$ for every $s\in \cS$. Therefore, the beliefs on the opponents' policies converge to a  (perfect) {Nash} equilibrium of the underlying zero-sum NMG. 
\end{proof}

\begin{remark}
    This can be also done with the stochastic fictitious-play dynamics, in a similar way as the argument in  \Cref{ssec:sfppqre}. 
\end{remark}

{\begin{remark}[Stationary equilibrium computation via value iteration]
 By \Cref{claim:contraction}, we know that with a star-shaped topology, we   can formulate a contracting value iteration operator, which plays an important role in showing the convergence of fictitious play. In fact, iterating such a contracting operator, which leads to the {\it value iteration}  algorithm, can lead to efficient NE computation in this star-shaped case also, with a fixed constant $\gamma$. This folklore result supplements the hardness results in \Cref{{thm:PPAD-hard-main}}, where stationary equilibria computation in cases other than the star-shaped ones are computationally intractable. This completes the landscape of \emph{stationary}  equilibria computation in zero-sum NMGs. We provide the value-iteration process in \Cref{alg:Value-iteration}.   
One can guarantee that $Q_1(s, \ba):= \sum_{i \in \cN/\{1\}} Q_{1, i}(s, a_1, a_i)$ converges to the $Q_1^\star(s,\ba)$, which corresponds to the Nash equilibrium values of the zero-sum NMG. {Also, by solving the maxmin problem in \eqref{eqn:minmax}, $Q_{1,i}(s)$ provides an approximate NE policy.  }
\end{remark}
}
    
\begin{algorithm}[H]
\caption{Value iteration for zero-sum NMGs of a star-shape}
\label{alg:Value-iteration}
\begin{algorithmic}
\STATE{Initialize $Q_{1,i}(s, a_1, a_i) = 0, Q_{i,1}(s, a_i, a_1) = 0$ for all $s \in \cS, \ba \in \cA, i \in \cN/\{1\}$ and $V_{i}(s) = 0$ for all $s \in \cS$, $i \in \cN$}
\FOR{each iteration $t = 0, 1, \dots$}
\STATE { Find $\mu$ for each  $s \in \cS$ such that 
\begin{align}
\mu(s) \in \argmax_{\mu_1 \in \Delta(\cA_1)}  \argmin_{\mu_2 \in \Delta(\cA_2), \dots, \mu_{n} \in \Delta(\cA_n) } \sum_{i \in \cN / \{1\}}\mu_1^\intercal {Q}_{1,i}(s) \mu_i  \label{eqn:minmax}
\end{align}
}
\STATE {Update $V_1(s) =\sum_{i \in \cN / \{1\}} \mu_1^\intercal(s) {Q}_{1,i}(s) \mu_i(s)$ for all $s \in \cS$} 
\STATE {Update $Q_{1, i}(s, a_1, a_i) = r_{1, i}(s, a_1, a_i) + \gamma \sum_{s' \in \cS} \frac{1}{n-1}\PP_1(s'\mid s, a_1) V_1(s')$ for all $i \in \cN/\{1\}$, $s \in \cS$, $\ba \in \cA$}
\ENDFOR
\end{algorithmic}
\end{algorithm}

\numbering{\section{Background on Stochastic Approximation and Differential Inclusions}}
\tnumbering{\chapter{Background on Stochastic Approximation and Differential Inclusions}}
\label{appendix:stochastic-approx} 
\Cref{appendix:stochastic-approx} introduces the theorem statement of \cite{benaim2005stochastic}. 
Let $F: \RR^m \rightrightarrows \RR^m$ be a set-valued function. Assume that $F$ satisfies the following properties:
\begin{enumerate}
\item $F$ is a closed set-valued map, meaning that its graph $\operatorname{Graph}(F) = \{(x, y) : y \in F(x)\}$ is a closed subset of $\mathbb{R}^m \times \mathbb{R}^m$.
\item $F(x)$ is a non-empty, compact, and convex subset of $\mathbb{R}^m$ for all $x \in \mathbb{R}^m$.
\item There exists a constant $c > 0$ such that for all $x \in \mathbb{R}^m$, we have $\sup_{z \in F(x)} |z| \leq c(1 + |x|)$.
\end{enumerate}
The differential inclusion problem involves finding a solution vector function $\bx: \RR \to \RR^m$ that satisfies the initial condition $\bx(0) = x \in \RR^m$ and the following relationship for almost all $t \in \mathbb{R}$:
$$
\frac{d \bx(t)}{d t} \in F(\bx(t)).
$$

\begin{definition}[Perturbed solutions]
A perturbed solution to $F$ refers to a continuous function $\by: \mathbb{R}_{+}=[0, \infty) \rightarrow \mathbb{R}^m$ that meets the following requirements:
\begin{itemize}
    \item $\by$ is absolutely continuous.
    \item There is a locally integrable function $t \mapsto U(t)$ that satisfies:
    \begin{itemize}
        \item For all $T>0$, the supremum of $\left|\int_t^{t+v} U(s) d s\right|$ over the interval $0 \leq v \leq T$ converges to zero as $t \rightarrow \infty$.
        \item For almost every $t>0$, the expression $\frac{d \by(t)}{d t}-U(t)$ belongs to $F^{\delta(t)}(\by(t))$, where $\delta: [0, \infty) \rightarrow \mathbb{R}$ is a function such that $\delta(t) \rightarrow 0$ as $t \rightarrow \infty$.
    \end{itemize}
\end{itemize}
\end{definition}

\begin{definition}[Stochastic approximations]
A discrete-time process $\left\{x_n\right\}_{n \in \mathbb{N}}$ is a stochastic approximation if it satisfies the following relationship:
$$ x_{n+1}-x_n-\gamma_{n+1} U_{n+1} \in \gamma_{n+1} F\left(x_n\right),$$
where the characteristics $\gamma$ and $U$ meet the following conditions:
\begin{itemize}
\item The sequence $(\gamma_n)_{n \geq 1}$ consists of non-negative numbers such that $\sum_{n=1}^\infty \gamma_n = \infty$ and $\lim _{n \rightarrow \infty} \gamma_n=0$.
\item The elements $U_n \in \mathbb{R}^m$ can be either deterministic or random perturbations.
\end{itemize}
\end{definition}
A continuous-time process can be associated with such a process as follows:
\begin{definition}[Affine interpolated process]\label{def:interpolated_process}
Define the following: 
$\tau_0=0  \text { and }  \tau_n=\sum_{i=1}^n \gamma_i \text { for } n \geq 1.
$ 
The continuous-time {\it affine interpolated process} $\bw: \mathbb{R}_{+} \to \mathbb{R}^m$ is defined as:
$$
\bw\left(\tau_n+s\right):=x_n+s \frac{x_{n+1}-x_n}{\tau_{n+1}-\tau_n}, \quad s \in\left[0, \gamma_{n+1}\right).
$$   
\end{definition}
We define $\Phi_t(x) = \{\bx(t): \bx$ is a solution to $\frac{d\bx(t)}{dt} \in F(\bx(t))$ with $\bx(0) = x\}$. 
\begin{definition}[Lyapunov function]
Lyapunov function for a set $\cS$ is a continuous function $V: \RR^m \to \RR$ if $V(y)< V(x)$ for all $x \in \cS\subseteq \RR^m$, $y \in \Phi_t(x), t>0$, and $V(y) \leq V(x)$ for all $x\in\cS$, $y \in \Phi_t(x), t>0$. 
\end{definition}
\begin{theorem} \label{thm:interpolated=perturbation}
Assume that the following hold:
\begin{itemize}
    \item For all $T>0$, the supremum of $\left\|\sum_{i=n}^{k-1} \gamma_{i+1} U_{i+1}\right\|$ for $k=n+1, \ldots, m\left(\tau_n+T\right)$ converges to zero as $n\to \infty$, where $$
m(t)=\sup \left\{k \geq 0: t \geq \tau_k\right\}.
$$ 
    \item $\sup _n\left\|x_n\right\|=M<\infty$.
\end{itemize}
Then the affine interpolated process (c.f. \Cref{def:interpolated_process}) is a perturbed solution.    
\end{theorem}

\begin{theorem}
\label{thm:levelset}
Suppose that $V$ is a Lyapunov function for a set $\Lambda$. Assume that $V(\Lambda)$ has an empty interior. For every bounded perturbed solution $\by$, define $L(\by)=\bigcap_{t \geq 0} \overline{\{\by(s): s \geq t\}}$, then $L(\fy)$  is contained in $\Lambda$ and $V (L(\fy))$  is constant.
\end{theorem}

\numbering{\section{Omitted Details in  \Cref{section:OMWU}}}
\tnumbering{\chapter{Omitted Details in  \Cref{section:OMWU}}}
\label{appendix:OMWU}
  
 We first recall the folklore result that approximating Markov non-stationary NE in \emph{infinite-horizon} discounted settings can be achieved by finding approximate Markov NE in \emph{finite-horizon} settings, with a large enough horizon length.
\begin{restatable}{proposition}{truncation}
\label{prop:truncation}
A $2\epsilon$-approximate Markov non-stationary NE in an infinite-horizon $\gamma$-discounted MG can be generated by (1) truncating the trajectory at time step $H\geq\frac{\log (R / \epsilon)}{1-\gamma}$ and (2) finding an {$\epsilon$-approximate Markov NE} in the  $H$-horizon MG. 
\end{restatable}
\begin{proof}
    We will execute a policy $\pi$ such that (1) for the first $H$ steps, we follow $\epsilon$-approximate NE in the $H$-truncated MG, and (2) after the $H$ steps, we follow an arbitrary policy.
    Then, we have 
    \begin{align*}
        &V^\pi_{i}(s) = \sum_{h=1}^\infty\EE_\pi [\gamma^{h-1} r_{i}(s_h, a_h) \mid s_1 = s]
        \\
        &\geq \sum_{h=1}^{H}\EE_{\mu_i, \pi_{-i}} [\gamma^{h-1} r_{i}(s_h, a_h) \mid s_1 = s] - \epsilon +  \sum_{h = H+1}^\infty \EE_\pi [\gamma^{h-1} r_{i}(s_h, a_h) \mid s_1 = s]
        \\
        &\geq \sum_{h=1}^{H}\EE_{\mu_i, \pi_{-i}} [\gamma^{h-1} r_{i}(s_h, a_h) \mid s_1 = s] - \epsilon  +  \sum_{h = H+1}^\infty \EE_{\mu_i, \pi_{-i}} [\gamma^{h-1} r_{i}(s_h, a_h) \mid s_1 = s] - \frac{R\gamma^{h-1}}{1-\gamma}
        \\
        &\geq \sum_{h=1}^{\infty}\EE_{\mu_i, \pi_{-i}} [\gamma^{h-1} r_{i}(s_h, a_h) \mid s_1 = s] - 2\epsilon
    \end{align*}
    for an arbitrary policy $\mu_i$. Here, the first inequality comes from the definition of NE in the $H$-truncated MG, the second inequality comes from $0 \leq r_i  \leq R$ and the last inequality comes from the definition of $H$. Therefore, the executed policy is a 2$\epsilon$-approximate NE. 
\end{proof}

In this section, we will utilize two performance metrics: \textsf{Matrix-NE-Gap} and \textsf{Matrix-QRE-Gap}. For the definition of QRE, we refer to \Cref{def:QRE} and \eqref{eqn:rtaudef}.
\begin{definition}
    For an NG $M$ with $(\cG = (\cN, \cE_r), \cA = \prod_{i \in \cN} \cA_i, (r_{i,j})_{(i,j) \in \cE_r})$, we define \textsf{Matrix-NE-Gap} and \text{Matrix-QRE-Gap} of $M$ for {some product policy} $\pi$ as follows:
    \begin{align*}
     &\textsf{Matrix-NE-Gap} (M, \pi) = \max_{i \in \cN} \max_{\pi_i' \in \Delta(\cA_i)} \left(r_i(\pi_i' ,\pi_{-i}) - r_i(\pi)  \right)
     \\
     &\qquad= \max_{i \in \cN}\left[\max_{\pi_i' \in \Delta(\cA_i)} \left(\sum_{j \in \cE_{r,i}} \pi_i' r_{i,j} \pi_j \right) - \left(\sum_{j \in \cE_{r,i}} \pi_i r_{i,j} \pi_j\right) \right]. 
     \\
    &\textsf{Matrix-QRE-Gap}_{\tau} (M, \pi) = \max_{i \in \cN}\max_{\pi_i' \in \Delta(\cA_i)} \left( r_{\tau, i}(\pi_i' ,\pi_{-i}) - r_{\tau, i}(\pi)  \right)
     \\
     &\qquad= \max_{i \in \cN}\left[\max_{\pi_i' \in \Delta(\cA_i)} \left(\sum_{j \in \cE_{r,i}} \pi_i' r_{i,j} \pi_j + \tau \cH(\pi_i' )\right) - \left(\sum_{j \in \cE_{r,i}} \pi_i r_{i,j} \pi_j + \tau \cH(\pi_i )\right) \right].
    \end{align*}
    When the underlying graph and the action space for the NG are clear, we also write $\textsf{Matrix-NE-Gap}(M, \pi)$ or $\textsf{Matrix-QRE-Gap}_{\tau} (M, \pi)$ as $\textsf{Matrix-NE-Gap}(\br, \pi)$ or $\textsf{Matrix-QRE-Gap}_{\tau} (\br, \pi)$, respectively.
\end{definition}
{We now provide the relationship between $\textsf{Matrix-NE-Gap}(\br, \pi)$ and $\textsf{Matrix-QRE-Gap}_{\tau} (\br, \pi)$.
\begin{lemma}[\cite{ao2022asynchronous}, Page 6, Equation (8)]
\label{lem:NE-QRE}
For an NG $M$ with $(\cG = (\cN, \cE_r), \cA = \prod_{i \in \cN} \cA_i, (r_{i,j})_{(i,j) \in \cE_r})$, the following holds:
\begin{align*}
    \textsf{Matrix-NE-Gap}(\br, \pi) \leq \textsf{Matrix-QRE-Gap}_{\tau} (\br, \pi) + \tau \max_{i \in \cN} \log |\cA_i|.
\end{align*}    
\end{lemma}
}
Thus, setting $\tau = \frac{\epsilon}{2\max_{i \in \cN}\log|\cA_i|}$, then an $\epsilon/2$-approximate QRE is also an $\epsilon$-approximate NE. Hence,   finding {an approximate-}QRE for zero-sum NGs is sufficient for finding an {approximate-}NE, with a small enough $\tau$. Now, we define the \textsf{NE-Gap} for an MNMG. 

\begin{definition}\label{def:NE_GAP}
    For an MNMG $M$ with $(\cG = (\cN, \cE_Q), \cS, \cA, H, (\PP_h)_{h \in [H]}, (r_{h,i,j}(s))_{(i, j) \in \cE_Q, s \in \cS})$, \textsf{NE-Gap} for some product policy $\pi$ at  timestep $h \in [H]$ is defined as follows:
    \begin{align*}
     \textsf{NE-Gap}_h(M, \pi) &= \max_{i \in \cN} \max_{s \in \cS} \max_{\pi_i' \in \Delta(\cA_i)^{|\cS| \times H} } \left(V_{h, i}^{\pi_i', \pi_{-i}}(s) - V_{h, i}^{\pi}(s)  \right).        
    \end{align*}
\end{definition}

\neurips{
\begin{algorithm}
\caption{A value-iteration-based  algorithm for finding NE in zero-sum NMGs}
\label{alg:meta-algorithm}
\begin{algorithmic}
\STATE{Update $V_{H+1,i}(s) = 0$ for all $s\in\cS$ and $i\in\cN$}
\FOR{step $h = H, H-1, \dots, 1$}
\IF{$\cN_C \neq \emptyset$} 
\STATE{Update $Q_{h,i, j}(s, a_i, a_j)$ for all $(i, j) \in \cE_Q, s \in \cS, a_i \in \cA_i, a_j \in \cA_j$  as 
\begin{equation}
\begin{aligned}
    &Q_{h,i,j}(s, a_i, a_j)= r_{h,i,j}(s, a_i, a_j) 
    \\
    &\qquad+ \sum_{s' \in \cS} \biggl(\frac{1}{|\cE_{Q, i}|} \pmb{1}(i \in \cN_C) \FF_{h, i} (s' \given  s, a_i) 
     +  \pmb{1}(j \in \cN_C)\FF_{h, j}(s'\given s,a_j)\biggr) \cdot V_{h+1, i}(s')        
\end{aligned}
 \label{eqn:Qupdates}
\end{equation}}
\ELSIF{$\cN_C = \emptyset$} 
\STATE{ Update $Q_{h,i, j}(s, a_i, a_j)$ for all $(i, j) \in \cE_Q, s \in \cS, a_i \in \cA_i, a_j \in \cA_j$ as
\begin{equation}
    \begin{aligned}
    &Q_{h,i,j}(s, a_i, a_j)= r_{h,i,j}(s, a_i, a_j) + \sum_{s' \in \cS} \left(\frac{1}{|\cE_{Q, i}|} \pmb{1}(j \in \cE_{Q, i}) \FF_{h, o} (s'\given s)\cdot V_{h+1, i}(s')\right)    
    \end{aligned}
    \label{eqn:Qupdates2}
\end{equation}
}
\ENDIF
\STATE{Update  $\pi_{h}(s) = \textsf{NE-ORACLE}(\cG, \cA, (Q_{h,i,j}(s))_{(i, j) \in \cE_{Q}})$ for all $ s\in \cS $}
\STATE{Update  $V_{h,i}(s)$ for all $i \in \cN, s \in \cS$ as 
{\begin{equation}
\begin{aligned}
    &V_{h, i}(s)=  \sum_{j \in \cE_{Q, i}}  {\pi}_{h, i}^{\intercal}(s) Q_{h,i,j}(s) {\pi}_{h,j}(s)
\end{aligned}
\label{eqn:Vupdates}
\end{equation}}
}
\ENDFOR
\end{algorithmic} 
\end{algorithm} 
}

Now we summarize the algorithm in \Cref{alg:meta-algorithm}. By \Cref{alg:meta-algorithm}, $Q_{h, i, j}(s, a_i, a_j) = Q_{\tau, h, i, j}^{{\pi}}(s, a_i, a_j)$ and $V_{h, i}(s) = V_{\tau, h, i}^{{\pi}}(s)$ holds, so that an NG characterized by $(\cG, \cA, (Q_{h, i, j}(s))_{(i,j) \in \cE_Q})$ is always a zero-sum NG for all $s \in \cS, h \in [H]$. And by induction, we can show that 
 $|Q_{h,i,j}^\pi(s, a_i, a_j)| \leq  (H + 1 - h)R$ for all $\pi$, $h \in [H]$, $(i, j)\in\cE_Q$, $s \in \cS$, $a_i \in \cA_i$, $a_j \in \cA_j$ (i.e., $\norm{\bQ_{h}^\pi(s)}_{\max} \leq HR$ for all $\pi$, $h \in [H]$, $s \in \cS$).
\NEORACLETONE*
\begin{proof}
Let $M = (\cG, \cS, \cA, H, (\PP_h)_{h \in [H]}, (r_{h,i,j}(s))_{(i, j) \in \cE_Q, s \in \cS})$. For any $\pi$ and $h \in [H]$, we have 
\begin{align*}
    &\textsf{NE-Gap}_h(M, \pi) = \max_{i \in \cN} \max_{s \in \cS} \max_{\pi_i' \in \Delta(\cA_i)^{|\cS| \times H}} \left(V_{h, i}^{\pi_i', \pi_{-i}}(s) - V_{h, i}^{\pi_i}(s)  \right)
    \\
    &=\max_{i \in \cN} \max_{s \in \cS} \max_{\pi_i' \in \Delta(\cA_i)^{|\cS| \times H}}  \left[V_{h, i}^{\pi_{i}', \pi_{-i}}(s) - V_{h, i}^{(\pi_{h, i}', \pi_{h, -i}), \pi_{h+1:H}}(s)    + V_{h, i}^{(\pi_{h, i}', \pi_{h, -i}), \pi_{h+1:H}}(s)  -V_{h, i}^{\pi}(s) \right]
    \\
    &=\max_{i \in \cN} \max_{s \in \cS} \max_{\pi_i' \in \Delta(\cA_i)^{|\cS| \times H}}  \left[ \PP_{h, (\pi'_{h,i}, \pi_{h, -i})} (V_{h+1, i}^{\pi_i', \pi_{-i}} - V_{h+1, i}^{\pi_{h+1:H}})(s)      
    + V_{h, i}^{(\pi_{h, i}', \pi_{h, -i}), \pi_{h+1:H}}(s)  -V_{h, i}^{\pi}(s) \right]
    \\    
    &\leq \textsf{NE-Gap}_{h+1} (M,\pi) + \max_{i \in \cN}\max_{s \in\cS}\max_{\pi_{h, i}'(s) \in \Delta(\cA_i)}\sum_{i \in \cN} \left[V_{h, i}^{(\pi_{h, i}', \pi_{h, -i}), \pi_{h+1:H}}(s) -V_{h, i}^{\pi}(s)  \right]
    \\
    &=  \textsf{NE-Gap}_{h+1} (M,\pi) + \max_{s \in \cS}\textsf{Matrix-NE-Gap} (\bQ_{h}(s), \pi_h) \leq  \textsf{NE-Gap}_{h+1} (M,\pi) + \max_{s \in\cS}\epsilon_{h,  s}.
\end{align*}
Therefore, for any $\pi$ and $h \in [H]$, we have $\textsf{NE-Gap}_h(M, \pi) \leq \sum_{h \in [H]}\max_{s \in\cS}\epsilon_{h, s}$.
\end{proof}

\numbering{\subsection{Several examples of \textsf{NE-ORACLE}}}
\tnumbering{\section{Several examples of \textsf{NE-ORACLE}}}
\paragraph{Example 1. Optimism \& Regularization: OMWU algorithm \cite{ao2022asynchronous}.} 
According to \cite[Theorem 1]{ao2022asynchronous}, if we apply  \Cref{alg:OMWU} to $(\cG, \cA, (Q_{h, i, j}(s))_{(i,j) \in \cE_Q})$ for each $h \in [H], s \in \cS$, 
the required number of iterations for $\textsf{Matrix-NE-Gap}(\bQ_{h}^\pi(s), \pi_h) \leq \epsilon/H$ is $\tilde{\cO}\left({H^2d_{\max}}/{\epsilon}\right)$ where $d_{\max}$ is the maximum degree of underlying graph $\cG$.  Consequently, the overall iteration complexity is  $\tilde{\cO}\left({H^3 d_{\max} |\cS|}/{\epsilon}\right)$. {Note that these results are in terms of last-iterate convergence.} 
\begin{algorithm}[h]
\caption{OMWU for zero-sum NGs with $\tau$-entropy regularization \cite{ao2022asynchronous}}
\label{alg:OMWU}
\begin{algorithmic}
\STATE{Choose $\pi_{i}^{(0)}, \bar{\pi}_{i}^{(0)}$ as uniform distributions for all $i \in \cN$}
\STATE{Define $\tau = 1/(n \max_{i \in \cN} \log |\cA_i|)$ and $\eta = {1}/({8n\norm{\br}_\infty})$}
\FOR{timestep $t = 0,1, \dots, $}
\STATE{Update the policy $\bar{\pi}_{i}^{(t+1)}$ as $  \bar{\pi}_{i}^{(t+1)}(a_i) \propto \bar{\pi}^{(t)}_{i}(a_i)^{1-\eta \tau} \exp\left(\eta [ \br_i {\pi}^{(t)}]_{a_i}\right)$ for all $i \in \cN$ and $a_i \in \cA_i$}
\STATE {Update the policy ${\pi}^{(t+1)}_{i}$ as ${\pi}^{(t+1)}_{i}(a_i) \propto \bar{\pi}^{(t+1)}_{i}(a_i)^{1-\eta \tau} \exp\left(\eta [ \br_i{\pi}^{(t)}]_{a_i}\right)$ for all $i \in \cN$ and $a_i \in \cA_i$
}
\ENDFOR
\end{algorithmic} 
\end{algorithm}

\paragraph{Example 2. Optimism \& Regularization-free: OMD algorithm \cite{anagnostides2022last}.}
According to \cite[Theorem 3.4]{anagnostides2022last}, if we apply  \Cref{alg:OMD} to $(\cG, \cA, (Q_{h, i, j}(s))_{(i,j) \in \cE_Q})$ for each $h \in [H], s \in \cS$, the required number of iterations for $\textsf{Matrix-NE-Gap}(\bQ_{h}^\pi(s), \pi_h) \leq \epsilon/H$ is $\tilde{\cO}\left({H^3n }/{\epsilon^2}\right)$. Consequently, the overall iteration complexity is  $\tilde{\cO}\left({H^4n|\cS| }/{\epsilon^2}\right)$. {Note that these results are in terms of best-iterate convergence.} 

\begin{algorithm}[h]
\caption{OMD {with KL-distance generating function} for zero-sum NGs \cite{anagnostides2022last}}
\label{alg:OMD}
\begin{algorithmic}
\STATE{Choose $\pi_{i}^{(0)}, \bar{\pi}_{i}^{(0)}$ as uniform distributions for all $i \in \cN$}
\STATE{Define $\eta = {1}/({4n \norm{\br}_\infty})$}
\FOR{timestep $t = 0,1, \dots, $}
\STATE{Update the policy ${\pi}_{i}^{(t+1)}$ as $  {\pi}_{i}^{(t+1)}(a_i) \propto \bar{\pi}^{(t)}_{i}(a_i) \exp\left(\eta [ \br_i {\pi}^{(t)}]_{a_i}\right)$ for all $i \in \cN$ and $a_i \in \cA_i$}
\STATE {Update the policy $\bar{\pi}^{(t+1)}_{i}$ as $\bar{\pi}^{(t+1)}_{i}(a_i) \propto \bar{\pi}^{(t)}_{i}(a_i) \exp\left(\eta [ \br_i{\pi}^{(t+1)}]_{a_i}\right)$ for all $i \in \cN$ and $a_i \in \cA_i$
}
\ENDFOR
\end{algorithmic} 
\end{algorithm}

\paragraph{Example 3. Optimism-free \& Regularization: MWU algorithm.} 
We provide MWU for zero-sum NGs {with regularization} in \Cref{ssec:mwu}. According to \Cref{thm:mwu-fixed} and \Cref{thm:mwu-diminish}, if we apply \Cref{alg:polymatrix-game-QRE} or \Cref{alg:polymatrix-game-decreasing-regular} to $(\cG, \cA, (Q_{h, i, j}(s))_{(i,j) \in \cE_Q})$ for each $h \in [H], s \in \cS$, the required number of iterations for $\textsf{Matrix-NE-Gap}(\bQ_{h}^\pi(s), \pi_h) \leq \epsilon/H$ is $\tilde{\cO}\left({H^8n}/{\epsilon^4}\right)$ or $\tilde{\cO}\left({H^{18}n^{3}}/{\epsilon^6}\right)$, respectively. Consequently, the overall iteration complexity is  $\tilde{\cO}\left({H^9n|\cS|}/{\epsilon^4}\right)$ or $\tilde{\cO}\left({H^{19}n^{3}|\cS|}/{\epsilon^6}\right)$, respectively. {Note that these results are in terms of last-iterate convergence.}

\numbering{\subsection{Analysis of MWU for zero-sum NGs with regularization}}
\tnumbering{\section{Analysis of MWU for zero-sum NGs with regularization}}
\label{ssec:mwu}

\numbering{\subsubsection{Fixed regularization}}
\tnumbering{\subsection{Fixed regularization}}

First, we provide an algorithm that has a {\it fixed} coefficient for entropy-regularization (\Cref{alg:polymatrix-game-QRE}). Recall that $0 \leq r_i \leq R$ for some $R>0$, for all $i \in \cN$.
\begin{algorithm}[H]
\caption{MWU for zero-sum NGs with $\tau$-entropy regularization}
\label{alg:polymatrix-game-QRE}
\begin{algorithmic}
\STATE{Choose $K = \ceil{2R/\tau  +  \log(\max_{i \in \cN}|\cA_i|))}$} 
\STATE{Choose ${\pi}_{i}^{(0)}$ as a uniform distribution for all $i \in \cN$}
\FOR{timestep $t = 0, 1, \dots$}
\STATE{ Define $\eta^{(t)} = 1/(\tau(t+K))$ }
\STATE {Update the policy as ${\pi}_{i}^{(t+1)}(a_i) \propto (\pi_i^{(t)}(a_i))^{1-\eta^{(t)}\tau} \exp \left(\eta^{(t)} [\br_i \pi^{(t)}]_{a_i}\right)$ for all $i \in \cN$ and $a_i \in \cA_i$}
\ENDFOR
\end{algorithmic}
\end{algorithm}

 \begin{claim}
\label{claim:insideOmega}
    Define $\Omega_i = \left\{ \pi_i \in \Delta(\cA_i) \mid \pi_i(a_i) \geq \frac{1}{|\cA_i|} \exp\left(-\frac{R}{\tau} \right) \text{ for all } a_i \in \cA_i \right\}$ and $g_i^{(t)} = \br_i \pi^{(t)} - \tau \log \pi_i^{(t)}$ in \Cref{alg:polymatrix-game-QRE}.  Then,  ${\pi}_{i}^{(t+1)} =  \argmax_{\pi_i \in \Omega_i}\left(\pi_i^\intercal g_i^{(t)} - \frac{1}{\eta^{(t)}} \KL(\pi_i, {\pi}_{i}^{(t)})\right)  $  holds for all $i \in \cN$ and $t \geq 0$. 
\end{claim}
\begin{proof}
The equation ${\pi}_{i}^{(t+1)}(a_i) \propto (\pi_i^{(t)}(a_i))^{1-\eta^{(t)}\tau} \exp \left(\eta^{(t)}[\br_i \pi^{(t)}]_{a_i}\right)$ implies that $${\pi}_{i}^{(t+1)} =  \argmax_{\pi_i \in \Delta(\cA_i)}\left(\pi_i^\intercal g_i^{(t)} - \frac{1}{\eta^{(t)}} \KL(\pi_i, {\pi}_{i}^{(t)})\right)$$ by a simple algebra. The remaining part to establish is that $\pi_i^{(t+1)} \in \Omega_i$ holds true for all $t \geq 0$ and for every $i \in \cN$. To prove the remaining part, we use induction. As $\pi_i^{(0)}$ is chosen to be a uniform distribution, it is clear that $\pi_i^{(0)} \in \Omega_i$. Under the assumption that $\pi_i^{(t)} \in \Omega_i$, we have
    \begin{align*}
        {\pi}_{i}^{(t+1)}(a_i) &= \frac{(\pi_i^{(t)}(a_i))^{1-\eta^{(t)}\tau} \exp \left([\eta^{(t)} \br_i \pi^{(t)}]_{a_i}\right)}{\sum_{a_i' \in \cA_i} (\pi_i^{(t)}(a_i'))^{1-\eta^{(t)}\tau} \exp \left([\eta^{(t)} \br_i \pi^{(t)}]_{a_i'}\right)}
        \\
        &\underset{(i)}{\geq} \frac{\left(\frac{1}{|\cA_i|} \exp\left(-\frac{R}{\tau} \right)\right)^{1- \eta^{(t)}\tau}}{\sum_{a_i' \in \cA_i} (\pi_i^{(t)}(a_i'))^{1-\eta^{(t)}\tau} \exp \left([\eta^{(t)} \br_i \pi^{(t)}]_{a_i'}\right)}
        \\
        &\underset{(ii)}{\geq} \frac{\left(\frac{1}{|\cA_i|} \exp\left(-\frac{R}{\tau} \right)\right)^{1- \eta^{(t)}\tau}}{\exp(\eta^{(t)}R) \sum_{a_i' \in \cA_i} (\pi_i^{(t)}(a_i'))^{1-\eta^{(t)}\tau}} 
        \underset{(iii)}{\geq} \frac{\left(\frac{1}{|\cA_i|} \exp\left(-\frac{R}{\tau} \right)\right)^{1- \eta^{(t)}\tau}}{\exp(\eta^{(t)}R) |\cA_i|^{\eta^{(t)}\tau}}
        \\
        &= \frac{1}{|\cA_i|} \exp\left(-\frac{R}{\tau} \right)
    \end{align*}
 thereby concluding the proof of our claim. In the above, $(i)$ is derived from $ \exp \left([\eta^{(t)} \br_i \pi^{(t)}]_{a_i}\right) \geq 1$ and the induction hypothesis; $(ii)$ is the result of $\exp \left([\eta^{(t)} \br_i \pi^{(t)}]_{a_i}\right) \leq \exp(\eta^{(t)} R)$. Finally, $(iii)$ comes from $\sum_{a_i \in \cA_i} (\pi_i^{(t)}(a_i))^{1-\eta^{(t)}\tau} \leq |\cA_i| \left(\frac{1}{|\cA_i|}\sum_{a_i \in \cA_i} \pi_i^{(t)}(a_i)\right)^{1-\eta^{(t)}\tau}$, by Jensen's inequality.
\end{proof}
In the forthcoming analysis of \Cref{alg:polymatrix-game-QRE} and \Cref{alg:polymatrix-game-decreasing-regular}, our first step bounds $\KL(\pi_\tau^\star, \pi)$. Following this, we employ \Cref{prop:QREvsKL} below to bound the term $\textsf{Matrix-QRE-Gap}_\tau$. Finally, we leverage \Cref{lem:NE-QRE} to bound the $\textsf{Matrix-NE-Gap}$.
\begin{proposition}
\label{prop:QREvsKL}
For any $\tau >0$, $\pi$, and $\br$, the following holds:
    $$\textsf{Matrix-QRE-Gap}_\tau (\br, \pi) \leq \cO\left( \left(\tau \max_{i \in \cN} \log |\cA_i| + R\right) \sqrt{\KL(\pi_\tau^\star, \pi)}\right).$$
\end{proposition}
\begin{proof}
For any $\pi$, $\tau$, and $\br$, we have
   \numbering{\arxiv{
   \begingroup
\allowdisplaybreaks
   \begin{align*}
    &\textsf{Matrix-QRE-Gap}_\tau (\br, \pi) 
    \\
    &= \max_{i \in \cN}\left[\max_{\pi_i' \in \Delta(\cA_i)} \left(\sum_{j \in \cE_{r,i}} \pi_i' r_{i,j} \pi_j + \tau \cH(\pi_i' )\right) - \left(\sum_{j \in \cE_{r,i}} \pi_i r_{i,j} \pi_j + \tau \cH(\pi_i )\right) \right]
        \\
        &= \max_{i \in \cN}\Bigg[\max_{\pi_i' \in \Delta(\cA_i)} \left(\sum_{j \in \cE_{r,i}} \pi_i' r_{i,j} \pi_{\tau, j}^\star + \tau \cH(\pi_i') + \sum_{j \in \cE_{r,i}} \pi_i' r_{i,j} \pi_j - \sum_{j \in \cE_{r,i}} \pi_i' r_{i,j} \pi_{\tau, j}^\star \right) - \left(\sum_{j \in \cE_{r,i}} \pi_i r_{i,j} \pi_j + \tau \cH(\pi_i )\right) \Bigg]
        \\
        &\underset{(i)}{\leq} \max_{i \in \cN}\left[\max_{\pi_i' \in \Delta(\cA_i)} \left(\sum_{j \in \cE_{r,i}} \pi_i' r_{i,j} \pi_{\tau, j}^\star + \tau \cH(\pi_i')\right) + \sum_{j \in \cE_{r,i}} R \norm{\pi_j - \pi_{\tau, j}^\star}_1  - \left(\sum_{j \in \cE_{r,i}} \pi_i r_{i,j} \pi_j + \tau \cH(\pi_i )\right) \right]
        \\
        &\underset{(ii)}{\leq} \max_{i \in \cN}\Bigg[\max_{\pi_i' \in \Delta(\cA_i)} \left(\sum_{j \in \cE_{r,i}} \pi_i' r_{i,j} \pi_{\tau, j}^\star + \tau \cH(\pi_i')\right)  - \left(\sum_{j \in \cE_{r,i}} \pi_{\tau, i}^\star r_{i,j} \pi_{\tau, j}^\star + \tau \cH(\pi_i )\right) 
        \\
        &\qquad \qquad + 2\sum_{j \in \cE_{r,i}} R \norm{\pi_j - \pi_{\tau, j}^\star}_1 + R \norm{\pi_i - \pi_{\tau, i}^\star}_1  \Bigg]
        \\
        &\underset{(iii)}{\leq} \cO\left(R\sqrt{\KL(\pi_\tau^\star, \pi)}\right) +  \max_{i \in \cN}\left[\max_{\pi_i' \in \Delta(\cA_i)} \left(\sum_{j \in \cE_{r,i}} \pi_i' r_{i,j} \pi_{\tau, j}^\star + \tau \cH(\pi_i')\right)  - \left(\sum_{j \in \cE_{r,i}} \pi_{\tau, i}^\star r_{i,j} \pi_{\tau, j}^\star + \tau \cH(\pi_i )\right) \right]
        \\
            &\underset{(iv)}{\leq} \cO\left(R\sqrt{\KL(\pi_\tau^\star, \pi)}\right) 
        \\
        &\quad +  \max_{i \in \cN}\Biggl[\max_{\pi_i' \in \Delta(\cA_i)} \left(\sum_{j \in \cE_{r,i}} \pi_i' r_{i,j} \pi_{\tau, j}^\star + \tau \cH(\pi_i')\right)  \\
        &\qquad \qquad \qquad- \left(\sum_{j \in \cE_{r,i}} \pi_{\tau, i}^\star r_{i,j} \pi_{\tau, j}^\star + \tau \cH(\pi_{\tau, i}^\star )\right)+ \sqrt{3}\tau\log |\cA_i| \sqrt{\KL(\pi_{\tau,i}^\star, \pi_i)} \Biggr]
        \\
        &\leq \cO\left( \left(\tau \max_{i \in \cN} \log |\cA_i| + R\right) \sqrt{\KL(\pi_\tau^\star, \pi)}\right), 
        \end{align*}
        \endgroup}}

    \noindent where $(i)$ and $(ii)$ are due to the definition of $R$. Meanwhile, $(iii)$ holds by the Pinsker inequality, and $(iv)$ holds by bounding the difference in Shannon entropy via KL divergence, \cite[Theorem 2]{reeb2015tight}. 
\end{proof}
\begin{lemma}[\cite{leonardos2021exploration, ao2022asynchronous}]
\label{lem:sumdiff}
    For any zero-sum NG with $( \cG = (\cN, \cE_r), \cA, (r_{i, j})_{(i, j)  \in \cE_r})$, for any {joint product} policies $\mu, \nu$, the following holds:
    \begin{align*}
        \sum_{i \in \cN}\left(r_i(\mu_i, \nu_{-i}) + r_i(\nu_i, \mu_{-i})\right) = 0.
    \end{align*}
\end{lemma}
By \Cref{lem:sumdiff} and the definition of $r_{\tau, i}$ (\Cref{eqn:rtaudef}), we can derive 
\begin{align}
    \sum_{i \in \cN}\left( r_{\tau, i} (\pi_i^{(t)}, \pi^{\star}_{\tau, -i}) +  r_{\tau, i} (\pi^{\star}_{\tau,i}, \pi_{-i}^{(t)})\right) = 0. \label{eqn:sumzero}
\end{align}
Analogous to the method used in \Cref{claim:insideOmega}, we can show that the QRE, $\pi^\star_\tau$, belongs to $\prod_{i \in \cN} \Omega_i$. Furthermore, we can bound $g_i^{(t)}(a_i)$ for all $i \in \cN$ and $a_i \in \cA_i$ by
\begin{align*}
    g_i^{(t)}(a_i) &\leq R- \tau \log {\pi}_{i}^{(t)}(a_i) \leq R - \tau \log \left( \frac{1}{|\cA_i|} \exp\left(-\frac{R}{\tau} \right)\right) 
    \\
    &= 2R + \tau \log(|\cA_i|)
\end{align*}
so that $\eta^{(t)} g_i^{(t)}(a_i) \leq 1$ for all $i \in  \cN$, $a_i \in \cA_i$, and $t \geq 0$. Therefore, we can apply \Cref{lem:reg-ineq} below. 

\begin{lemma}[\cite{luo2022online}, Theorem 2]
\label{lem:reg-ineq}
    For a convex set $\Omega \subseteq \Delta(\cA)$, $\eta \mu 
\preceq \pmb{1} \in  \RR^{|\cA|}$, and any $\pi \in \Omega$, define 
$$\pi' = \argmax_{\tilde\pi \in \Omega}\left({\tilde\pi}^{\intercal} \mu - \frac{1}{\eta}\KL(\tilde\pi, \pi)\right).$$ 
Then, for any $\nu \in \Omega$, the following holds:
    \begin{align*}
        (\nu - \pi)^\intercal g \leq \frac{\KL(\nu, \pi) - \KL(\nu, \pi')}{\eta} + \eta \sum_{ a \in \cA} \pi(a) g(a)^2. 
    \end{align*}
\end{lemma}

Now, we state the theorem on the iteration complexity of \Cref{alg:polymatrix-game-QRE} to obtain an $\epsilon$-NE. 
\begin{theorem}
\label{thm:mwu-fixed}
The last iterate of \Cref{alg:polymatrix-game-QRE} requires no more than $\tilde{\cO}\left(nR^4 / \epsilon^4 \right)$ iterations to achieve an $\epsilon$-NE of $(\cG, \cA, (r_{i,j})_{(i, j) \in \cE_{r}})$.
\end{theorem}
\begin{proof} 
First, we bound the difference between $ r_{\tau, i} (\pi^{\star}_{\tau,i}, \pi_{-i}^{(t)})$ and $r_{\tau, i} (\pi^{(t)})$ as follows:
$$
\begin{aligned}
 &r_{\tau, i} (\pi^{\star}_{\tau,i}, \pi_{-i}^{(t)})  - r_{\tau, i} (\pi^{(t)})   \\&=\left(\pi^{\star}_{\tau,i} - \pi^{(t)}_{i}\right)^{\intercal} \br_i \pi^{(t)} +\tau\left(\sum_{a_i \in \cA_i} \pi_i^{(t)}(a_i) \log \pi_i^{(t)}(a_i)-\sum_{a_i \in \cA_i} \pi^{\star}_{\tau,i}(a_i) \log \pi^{\star}_{\tau,i}(a_i)\right)
\\
& =\left(\pi^{\star}_{\tau,i} - \pi^{(t)}_{i}\right)^{\intercal} g_i^{(t)}-\tau \mathrm{KL}\left(\pi^{\star}_{\tau,i}, \pi_i^{(t)}\right)
\\
& \leq \frac{\mathrm{KL}\left(\pi^{\star}_{\tau,i}, \pi_i^{(t)}\right)-\mathrm{KL}\left(\pi^{\star}_{\tau,i}, \pi_i^{(t+1)}\right)}{\eta^{(t)}}+\eta^{(t)} \sum_{a_i \in \cA_i} \pi_i^{(t)}(a_i)\left(g_i^{(t)}(a_i)\right)^2-\tau \mathrm{KL}\left(\pi^{\star}_{\tau,i}, \pi_i^{(t)}\right)
\\
&\leq \frac{\mathrm{KL}\left(\pi^{\star}_{\tau,i}, \pi_i^{(t)}\right)-\mathrm{KL}\left(\pi^{\star}_{\tau,i}, \pi_i^{(t+1)}\right)}{\eta^{(t)}}+\eta^{(t)} \left(2R + \tau \max_{i \in \cN}\log(|\cA_i|)\right)^2 - \tau \mathrm{KL}\left(\pi^{\star}_{\tau,i}, \pi_i^{(t)}\right).
\end{aligned}
$$
Here, the penultimate inequality holds by \Cref{lem:reg-ineq} since $\pi_{\tau, i} \in \Omega_i$ and $\eta^{(t)} g_i^{(t)}(a_i) \leq 1$ holds for all $t > 0$ and $a_i \in \cA_i$. Therefore, we have  
\numbering{\arxiv{\begin{align}
    &\mathrm{KL}\left(\pi^{\star}_{\tau,i}, \pi_i^{(t+1)}\right) \nonumber
    \\
    &\leq \left(1-\eta^{(t)} \tau\right) \mathrm{KL}\left(\pi^{\star}_{\tau,i}, \pi_i^{(t)}\right) + \eta^{(t)}(r_{\tau, i} (\pi^{(t)}) - r_{\tau, i} (\pi^{\star}_{\tau,i}, \pi_{-i}^{(t)}))  + (\eta^{(t)})^2 \left(2R + \tau \max_{i \in \cN}\log(|\cA_i|)\right)^2 \label{eqn:KLi}
\end{align}}}
\tnumbering{\begin{align}
    &\mathrm{KL}\left(\pi^{\star}_{\tau,i}, \pi_i^{(t+1)}\right) \nonumber
    \\
    &\leq \left(1-\eta^{(t)} \tau\right) \mathrm{KL}\left(\pi^{\star}_{\tau,i}, \pi_i^{(t)}\right) + \eta^{(t)}(r_{\tau, i} (\pi^{(t)}) - r_{\tau, i} (\pi^{\star}_{\tau,i}, \pi_{-i}^{(t)})) \nonumber  
    \\
    &\qquad + (\eta^{(t)})^2 \left(2R + \tau \max_{i \in \cN}\log(|\cA_i|)\right)^2 \label{eqn:KLi}
\end{align}}
\neurips{\begin{align}
    &\mathrm{KL}\left(\pi^{\star}_{\tau,i}, \pi_i^{(t+1)}\right) \nonumber
    \\
    &\leq \left(1-\eta^{(t)} \tau\right) \mathrm{KL}\left(\pi^{\star}_{\tau,i}, \pi_i^{(t)}\right) + \eta^{(t)}(r_{\tau, i} (\pi^{(t)}) - r_{\tau, i} (\pi^{\star}_{\tau,i}, \pi_{-i}^{(t)})) \nonumber  
    \\
    &\qquad + (\eta^{(t)})^2 \left(2R + \tau \max_{i \in \cN}\log(|\cA_i|)\right)^2 \label{eqn:KLi}
\end{align}}
for all $i \in \cN$ and $t \geq 0$. If we sum \eqref{eqn:KLi} over $i \in \cN$, we have 
\numbering{\arxiv{\begin{align}
    &\mathrm{KL}\left(\pi^{\star}_{\tau}, \pi^{(t+1)}\right) \nonumber
    \\
    &\leq \left(1-\eta^{(t)} \tau\right) \mathrm{KL}\left(\pi^{\star}_{\tau}, \pi^{(t)}\right) + n(\eta^{(t)})^2 \left(2R + \tau \max_{i \in \cN}\log(|\cA_i|)\right)^2 +  \sum_{i \in \cN} \eta^{(t)}(r_{\tau, i} (\pi^{(t)}) - r_{\tau, i} (\pi^{\star}_{\tau,i}, \pi_{-i}^{(t)}))\nonumber
 \\
 &= \left(1-\eta^{(t)} \tau\right) \mathrm{KL}\left(\pi^{\star}_{\tau}, \pi^{(t)}\right) + n(\eta^{(t)})^2 \left(2R + \tau \max_{i \in \cN}\log(|\cA_i|)\right)^2+  \sum_{i \in \cN} \eta^{(t)}(r_{\tau, i} (\pi_{i}^{(t)}, \pi^{\star}_{\tau, -i}) - r_{\tau, i} (\pi^{\star}_{\tau}) ) \label{eqn:change}
 \\
 &\leq \left(1-\eta^{(t)} \tau\right) \mathrm{KL}\left(\pi^{\star}_{\tau}, \pi^{(t)}\right) + n(\eta^{(t)})^2 \left(2R + \tau \max_{i \in \cN}\log(|\cA_i|)\right)^2
 \label{eqn:recursive1}
\end{align}
}}
for all $t \geq 0$ where \eqref{eqn:change} holds due to \eqref{eqn:sumzero} and $\sum_{i \in \cN} r_{\tau, i} (\pi^{(t)}) =\sum_{i \in \cN}r_{\tau, i} (\pi^{\star}_{\tau}) =0$, and \eqref{eqn:recursive1} holds due to that $\pi^\star_\tau$ is the NE for the game that having the payoff $r_{\tau, i}$.
If we recursively apply \eqref{eqn:recursive1}, we have
\begin{align*}
    &\mathrm{KL}\left(\pi^{\star}_{\tau}, \pi^{(t+1)}\right)
    \\
    &\leq \prod_{l= 0}^t (1- \eta^{(l)}\tau) \KL(\pi_\tau^\star, \pi^{(0)})  +  n\left(2R + \tau \max_{i \in \cN}\log(|\cA_i|)\right)^2\sum_{l = 0}^{t} (\eta^{(l)})^2 \prod_{s = l+1}^t (1-\eta^{(s)}\tau)
    \\
    &= \prod_{l= 0}^t (1- 1/(l + K)) \KL(\pi_\tau^\star, \pi^{(0)}) \\&\qquad +  n\left(2R + \tau \max_{i \in \cN}\log(|\cA_i|)\right)^2\sum_{l = 0}^{t} \frac{1}{\tau^2} (1/(l+K)^2) \prod_{s = l+1}^t (1-1/(s+K))
    \\
    &\leq K/(K+t) \KL(\pi_\tau^\star, \pi^{(0)})  \\&\qquad+  n\left(2R + \tau \max_{i \in \cN}\log(|\cA_i|)\right)^2 \frac{1}{\tau^2} \log\left( (t+K+1)/(K)\right) (1/(t+K))
    \\
    &= \tilde{\cO}\left( \left(K \sum_{i \in \cN} \log (|\cA_i|)+ n\max\left(R^2/\tau^2, \max_{i \in \cN} \log^2(|\cA_i|) \right)\right)/t\right) 
    \\
    &= \tilde{\cO}\left( n\max\left(R^2/\tau^2, \max_{i \in \cN} \log^2(|\cA_i|)\right)/t\right).
\end{align*}
Therefore, if we iterate \Cref{alg:polymatrix-game-QRE} for $T$ times, by \Cref{prop:QREvsKL} and  \Cref{lem:NE-QRE}, we obtain an
\begin{align*}
    \tilde{\cO}\left(\tau \max_{i \in \cN} \log |\cA_i|  + \left(\tau \max_{i \in \cN} \log |\cA_i|  + R\right) \sqrt{ n\max\left(R^2/\tau^2, \max_{i \in \cN} \log^2(|\cA_i|)\right)/T} \right) 
\end{align*}
approximate NE.  Therefore, if we want to obtain an $\epsilon$-NE ($\epsilon > 0$) for the matrix game in the last iterate, we need to have $\tilde{\cO}\left(nR^4/ \epsilon^4 \right)$ iteration.
\end{proof}

\numbering{\subsubsection{Diminishing regularization} }
\tnumbering{\subsection{Diminishing regularization}}

One might also consider the algorithm with a diminishing choice of $\tau$. We provide \Cref{alg:polymatrix-game-decreasing-regular}  for this case. 

\begin{algorithm}[H]
\caption{MWU for zero-sum NGs with diminishing entropy regularization}
\label{alg:polymatrix-game-decreasing-regular}
\begin{algorithmic}
\STATE{Choose $K =(R + 2\max_{i \in \cN} \log |\cA_i|)^2$} 
\STATE{Choose ${\pi}_{i}^{(0)}$ as a uniform distribution for all $i \in \cN$}
\FOR{timestep $t = 0, 1, \dots$}
\STATE{ Define $\tau^{(t)} = (t+K)^{-1/6}$ and $\eta^{(t)} = (t+K)^{-1/2}$}
\STATE {Update $g_i^{(t)} =\br_i \pi^{(t)}  - \tau^{(t)} \log {\pi}_{i}^{(t)}$ for all $i \in \cN$ }
\STATE {Define $\Omega_i^{(t)} = \left\{ \pi_i \in \Delta(\cA_i) \mid \pi_i(a_i) \geq \frac{1}{|\cA_i|(t+K)^2} \text{ for all } a_i \in \cA_i \right\}$ }
\STATE {Update the policy as ${\pi}_{i}^{(t+1)} = \argmax_{\pi_i \in \Omega_{i}^{(t)}}\left(\pi_i^\intercal g_i^{(t)} - \frac{1}{\eta^{(t)}} \KL(\pi_i, {\pi}_{i}^{(t)})\right)$ for all $i \in \cN$} 
\ENDFOR
\end{algorithmic}
\end{algorithm}
Let $\pi^{\star}_{\tau^{(t)}}$  be the unique NE in the policy space $\prod_{i \in \cN} \Omega_i^{(t)}$ for the game that has the payoff $r_{\tau^{(t)}, i}$ as 
\begin{align*}
r_{\tau^{(t)}, i}(\pi) &= r_{i}(\pi) + \tau^{(t)}\cH(\pi_i) - \sum_{j \in \cE_{r,i}}\frac{\tau^{(t)}}{|\cE_{r, j}|} \cH(\pi_j) = \pi_i^\intercal \br_i \pi + \tau^{(t)} \cH(\pi_i) - \sum_{j \in \cE_{r,i}}\frac{\tau^{(t)}}{|\cE_{r, j}|} \cH(\pi_j).
\end{align*}
Moreover, we can bound $g_i^{(t)}(a_i)$  for all $i \in \cN$ and $a_i \in \cA_i$:
\begin{align*}
    g_i^{(t)}(a_i) &\leq R- (t+K)^{-1/6}  \log {\pi}_{i}^{(t)}(a_i) \leq R +  2(t+K)^{-1/6} \log (|\cA_i|(t+K)) 
    \\
    &\leq  R + 2\max_{i \in \cN} \log (|\cA_i|),
\end{align*}
so that $\eta^{(t)} g_i^{(t)}(a_i) \leq 1$ for all $i \in  \cN$, $a_i \in \cA_i$, and $t \geq 0$. Therefore, we can apply \Cref{lem:reg-ineq}. 
\begin{theorem}
    \label{thm:mwu-diminish}
    The last iterate of \Cref{alg:polymatrix-game-decreasing-regular} requires no more than $\tilde{\cO}\left(n^3 R^{12}/\epsilon^6\right)$ iterations to achieve an $\epsilon$-NE of $(\cG, \cA, (r_{i,j})_{(i, j) \in \cE_{r}})$. 
\end{theorem}    
\begin{proof}
First, we bound the difference between $ r_{\tau^{(t)}, i} (\pi^{\star}_{\tau^{(t)},i}, \pi_{-i}^{(t)})$ and $ r_{\tau^{(t)}, i} (\pi^{(t)})$ as follows:
$$
\begin{aligned}
 &r_{\tau^{(t)}, i} (\pi^{\star}_{\tau^{(t)},i}, \pi_{-i}^{(t)})  - r_{\tau^{(t)}, i} (\pi^{(t)})  \\&=\left(\pi^{\star}_{\tau^{(t)},i} - \pi^{(t)}_{i}\right)^{\intercal} \br_i \pi^{(t)} +\tau^{(t)}\left(\sum_{a_i \in \cA_i} \pi_i^{(t)}(a_i) \log \pi_i^{(t)}(a_i)-\sum_{a_i \in \cA_i} \pi^{\star}_{\tau^{(t)},i}(a_i) \log \pi^{\star}_{\tau^{(t)},i}(a_i)\right)
\\
& =\left(\pi^{\star}_{\tau^{(t)},i} - \pi^{(t)}_{i}\right)^{\intercal} g_i^{(t)}-\tau^{(t)} \mathrm{KL}\left(\pi^{\star}_{\tau^{(t)},i}, \pi_i^{(t)}\right)
\\
& \leq \frac{\mathrm{KL}\left(\pi^{\star}_{\tau^{(t)},i}, \pi_i^{(t)}\right)-\mathrm{KL}\left(\pi^{\star}_{\tau^{(t)},i}, \pi_i^{(t+1)}\right)}{\eta^{(t)}}\\&\qquad+\eta^{(t)} \sum_{a_i \in \cA_i} \pi_i^{(t)}(a_i)\left(g_i^{(t)}(a_i)\right)^2-\tau^{(t)} \mathrm{KL}\left(\pi^{\star}_{\tau^{(t)},i}, \pi_i^{(t)}\right)
\\
&\leq \frac{\mathrm{KL}\left(\pi^{\star}_{\tau^{(t)},i}, \pi_i^{(t)}\right)-\mathrm{KL}\left(\pi^{\star}_{\tau^{(t)},i}, \pi_i^{(t+1)}\right)}{\eta^{(t)}}
\\&\qquad+\eta^{(t)} \left(R + 2(t+K)^{-1/6}\max_{i \in \cN} \log( |\cA_i|(t+K))\right)^2 -\tau^{(t)} \mathrm{KL}\left(\pi^{\star}_{\tau^{(t)},i}, \pi_i^{(t)}\right).
\end{aligned}
$$
Here, the penultimate inequality holds by \Cref{lem:reg-ineq} since $\pi_{\tau^{(t)}, i} \in \Omega_i^{(t)}$ and $\eta^{(t)} g_i^{(t)}(a_i) \leq 1$ holds for every $t > 0$ and $a_i \in \cA_i$. Therefore, we have  
\begin{align}
    &\mathrm{KL}\left(\pi^{\star}_{\tau^{(t+1)},i}, \pi_i^{(t+1)}\right) \nonumber
    \\
    &\leq \left(1-\eta^{(t)} \tau^{(t)}\right) \mathrm{KL}\left(\pi^{\star}_{\tau^{(t)},i}, \pi_i^{(t)}\right) + \eta^{(t)}(r_{\tau^{(t)}, i} (\pi^{(t)}) - r_{\tau^{(t)}, i} (\pi^{\star}_{\tau^{(t)},i}, \pi_{-i}^{(t)}))  \label{eqn:KLi-2}
    \\
    &\qquad  + (\eta^{(t)})^2 \left(R + 2(t+K)^{-1/6}\max_{i \in \cN} \log( |\cA_i|(t+K))\right)^2\nonumber
    \\
    &\qquad+ \mathrm{KL}\left(\pi^{\star}_{\tau^{(t+1)},i}, \pi_i^{(t+1)}\right)  - \mathrm{KL}\left(\pi^{\star}_{\tau^{(t)},i}, \pi_i^{(t+1)}\right)  \nonumber
\end{align}
for all $i \in \cN$ and $t \geq 0$. If we sum \eqref{eqn:KLi-2} over $i \in \cN$, we have 
\begin{align}
    &\mathrm{KL}\left(\pi^{\star}_{\tau^{(t+1)}}, \pi^{(t+1)}\right) \nonumber
    \\
    &\leq \left(1-\eta^{(t)} \tau^{(t)}\right) \mathrm{KL}\left(\pi^{\star}_{\tau^{(t)}}, \pi^{(t)}\right) + n(\eta^{(t)})^2 \left(R + 2(t+K)^{-1/6}\max_{i \in \cN} \log( |\cA_i|(t+K))\right)^2 \nonumber
    \\
    &\qquad +  \sum_{i \in \cN} \eta^{(t)}(r_{\tau^{(t)}, i} (\pi^{(t)}) - r_{\tau^{(t)}, i} (\pi^{\star}_{\tau^{(t)},i}, \pi_{-i}^{(t)})) +  \mathrm{KL}\left(\pi^{\star}_{\tau^{(t+1)}}, \pi^{(t+1)}\right)  - \mathrm{KL}\left(\pi^{\star}_{\tau^{(t)}}, \pi^{(t+1)}\right) \nonumber
 \\
 &= \left(1-\eta^{(t)} \tau^{(t)}\right) \mathrm{KL}\left(\pi^{\star}_{\tau^{(t)}}, \pi^{(t)}\right) + n(\eta^{(t)})^2 \left(R + 2(t+K)^{-1/6}\max_{i \in \cN} \log( |\cA_i|(t+K))\right)^2 \label{eqn:change-ne}
     \\
    &\qquad +  \sum_{i \in \cN} \eta^{(t)}(r_{\tau^{(t)}, i} (\pi_{i}^{(t)}, \pi^{\star}_{\tau^{(t)}, -i}) - r_{\tau^{(t)}, i} (\pi^{\star}_{\tau^{(t)}}) ) +  \mathrm{KL}\left(\pi^{\star}_{\tau^{(t+1)}}, \pi^{(t+1)}\right)  - \mathrm{KL}\left(\pi^{\star}_{\tau^{(t)}}, \pi^{(t+1)}\right) \nonumber
 \\
 &\leq \left(1-\eta^{(t)} \tau^{(t)}\right) \mathrm{KL}\left(\pi^{\star}_{\tau^{(t)}}, \pi^{(t)}\right) + n(\eta^{(t)})^2 \left(R + 2(t+K)^{-1/6}\max_{i \in \cN} \log( |\cA_i|(t+K))\right)^2\nonumber 
 \\
 &\qquad +  \mathrm{KL}\left(\pi^{\star}_{\tau^{(t+1)}}, \pi^{(t+1)}\right)  - \mathrm{KL}\left(\pi^{\star}_{\tau^{(t)}}, \pi^{(t+1)}\right),
 \label{eqn:recursive1-ne}
\end{align} 
for all $t\geq 0$ where \eqref{eqn:change-ne} holds due to \eqref{eqn:sumzero}, and \eqref{eqn:recursive1-ne} holds due to that $\pi^\star_\tau$ is the NE for the game with the payoff $r_{\tau^{(t)}, i}$. If we recursively apply \eqref{eqn:recursive1-ne}, we have
\numbering{\begin{align*}
    &\mathrm{KL}\left(\pi^{\star}_{\tau^{(t+1)}}, \pi^{(t+1)}\right)
    \\
    &\leq \prod_{l= 0}^t (1- \eta^{(l)}\tau^{(l)}) \KL(\pi_\tau^\star, \pi^{(0)})  \neurips{\\
    &\qquad}+  n\left(R + 2(t+K)^{-1/6}\max_{i \in \cN} \log( |\cA_i|(t+K))\right)^2\sum_{l = 0}^{t} (\eta^{(l)})^2 \prod_{s = l+1}^t (1-\eta^{(s)}\tau^{(s)})
    \\
    &\qquad + \sum_{l = 0}^{t} \prod_{s = l+1}^t (1-\eta^{(s)}\tau^{(s)}) \left( \mathrm{KL}\left(\pi^{\star}_{\tau^{(t+1)}}, \pi^{(t+1)}\right)  - \mathrm{KL}\left(\pi^{\star}_{\tau^{(t)}}, \pi^{(t+1)}\right)\right)
    \\
    &\leq \prod_{l= 0}^t (1- (t+K)^{-2/3}) \KL(\pi_\tau^\star, \pi^{(0)})  \neurips{\\
    &\qquad}+  n\left(R + 2(t+K)^{-1/6}\max_{i \in \cN} \log( |\cA_i|(t+K))\right)^2\sum_{l = 0}^{t} (1/(l+K)) \prod_{s = l+1}^t (1-(s+K)^{-2/3})
    \\
    &\qquad + \sum_{l = 0}^{t} \prod_{s = l+1}^t (1-(s+K)^{-2/3}) \left( \mathrm{KL}\left(\pi^{\star}_{\tau^{(t+1)}}, \pi^{(t+1)}\right)  - \mathrm{KL}\left(\pi^{\star}_{\tau^{(t)}}, \pi^{(t+1)}\right)\right)
    \\
    &\underset{(i)}{\leq} \tilde{\cO} \left(n \max( R^2, \max_{i \in \cN} \log(|\cA_i|t)t^{-1/3}) (t+K)^{-1/3}  + \max_{i \in \cN}\log^3(|\cA_i|t) (t+K)^{-1/3} 
 \right)\neurips{\\
    &}= \tilde{\cO} \left(nR^2t^{-1/3}\right).
\end{align*}}
\tnumbering{\begin{align*}
    &\mathrm{KL}\left(\pi^{\star}_{\tau^{(t+1)}}, \pi^{(t+1)}\right)
    \\
    &\leq \prod_{l= 0}^t (1- \eta^{(l)}\tau^{(l)}) \KL(\pi_\tau^\star, \pi^{(0)})  \\&\qquad+  n\left(R + 2(t+K)^{-1/6}\max_{i \in \cN} \log( |\cA_i|(t+K))\right)^2\sum_{l = 0}^{t} (\eta^{(l)})^2 \prod_{s = l+1}^t (1-\eta^{(s)}\tau^{(s)})
    \\
    &\qquad + \sum_{l = 0}^{t} \prod_{s = l+1}^t (1-\eta^{(s)}\tau^{(s)}) \left( \mathrm{KL}\left(\pi^{\star}_{\tau^{(t+1)}}, \pi^{(t+1)}\right)  - \mathrm{KL}\left(\pi^{\star}_{\tau^{(t)}}, \pi^{(t+1)}\right)\right)
    \\
    &\leq \prod_{l= 0}^t (1- (t+K)^{-2/3}) \KL(\pi_\tau^\star, \pi^{(0)})  \\&\qquad+  n\left(R + 2(t+K)^{-1/6}\max_{i \in \cN} \log( |\cA_i|(t+K))\right)^2\sum_{l = 0}^{t} (1/(l+K)) \prod_{s = l+1}^t (1-(s+K)^{-2/3})
    \\
    &+ \sum_{l = 0}^{t} \prod_{s = l+1}^t (1-(s+K)^{-2/3}) \left( \mathrm{KL}\left(\pi^{\star}_{\tau^{(t+1)}}, \pi^{(t+1)}\right)  - \mathrm{KL}\left(\pi^{\star}_{\tau^{(t)}}, \pi^{(t+1)}\right)\right)
    \\
    &\underset{(i)}{\leq} \tilde{\cO} \left(n \max( R^2, \max_{i \in \cN} \log(|\cA_i|t)t^{-1/3}) (t+K)^{-1/3}  + \max_{i \in \cN}\log^3(|\cA_i|t) (t+K)^{-1/3} 
 \right)\\
    &= \tilde{\cO} \left(nR^2t^{-1/3}\right).
\end{align*}}
Here, (i) holds because   
\begin{align*}
    &\prod_{l= 0}^t (1- (t+K)^{-2/3}) \KL(\pi_\tau^\star, \pi^{(0)}) = \tilde{\cO}\left(\exp\left(-t^{1/3} \right)\right)
\end{align*}
holds, 
\begin{align*}
    \sum_{l = 0}^{t} (1/(l+K)) \prod_{s = l+1}^t (1-(s+K)^{-2/3}) = \tilde{\cO}( (t+K)^{-1/3})
\end{align*} 
holds by \cite[Lemma 4]{cai2023uncoupled}, and  
\numbering{\begin{align*}
    &\sum_{l = 0}^{t} \prod_{s = l+1}^t (1-(s+K)^{-2/3}) \left( \mathrm{KL}\left(\pi^{\star}_{\tau^{(t+1)}}, \pi^{(t+1)}\right)  - \mathrm{KL}\left(\pi^{\star}_{\tau^{(t)}}, \pi^{(t+1)}\right)\right) \neurips{\\
    &\qquad}\leq \max_{i \in \cN}\log^3(|\cA_i|(t+K)) (t+K)^{-1/3}
\end{align*}}
\tnumbering{
\begin{align*}
    &\sum_{l = 0}^{t} \prod_{s = l+1}^t (1-(s+K)^{-2/3}) \left( \mathrm{KL}\left(\pi^{\star}_{\tau^{(t+1)}}, \pi^{(t+1)}\right)  - \mathrm{KL}\left(\pi^{\star}_{\tau^{(t)}}, \pi^{(t+1)}\right)\right) \\
    &\qquad\leq \max_{i \in \cN}\log^3(|\cA_i|(t+K)) (t+K)^{-1/3}
\end{align*}
}
holds by \cite[Lemma 4, Lemma 15]{cai2023uncoupled}. Therefore, if we iterate \Cref{alg:polymatrix-game-decreasing-regular} for $T$ times, by \Cref{prop:QREvsKL} and  \Cref{lem:NE-QRE}, we obtain 
\begin{align*}
    \tilde{\cO}\left(T^{-1/6} \max_{i \in \cN} \log |\cA_i|  + \left(T^{-1/6} \max_{i \in \cN} \log |\cA_i|  + R\right)  \sqrt{nR^2T^{-1/3}} \right) 
\end{align*}
approximate NE.  Therefore, if we want to obtain $\epsilon$-NE ($\epsilon > 0$) for the matrix game, we need to have $\tilde{\cO}\left(n^3R^{12} /\epsilon^6\right)$ iterations. 
\end{proof}

\end{document}